\documentclass[sigconf,nonacm]{aamas}

\newif\ifshort
\shortfalse

\newif\ifnocomment
\nocommentfalse %

\usepackage{multirow}
\usepackage{thm-restate}
\usepackage{paralist,enumitem}
\usepackage[textsize=small,textwidth=4cm]{todonotes}
\usepackage[title]{appendix}
\usepackage{colortbl}
\usepackage{tikz}
\usepackage{amsmath,amsthm,pict2e}
\usepackage{multibib}
\newcites{app}{References for the Appendix}

\usepackage{balance} %

\makeatletter
\gdef\@copyrightpermission{
  \begin{minipage}{0.2\columnwidth}
   \href{https://creativecommons.org/licenses/by/4.0/}{\includegraphics[width=0.90\textwidth]{by}}
  \end{minipage}\hfill
  \begin{minipage}{0.8\columnwidth}
   \href{https://creativecommons.org/licenses/by/4.0/}{This work is licensed under a Creative Commons Attribution International 4.0 License.}
  \end{minipage}
  \vspace{5pt}
}
\makeatother

\setcopyright{ifaamas}
\acmConference[AAMAS '25]{Proc.\@ of the 24th International Conference
on Autonomous Agents and Multiagent Systems (AAMAS 2025)}{May 19 -- 23, 2025}
{Detroit, Michigan, USA}{Y.~Vorobeychik, S.~Das, A.~Nowé  (eds.)}
\copyrightyear{2025}
\acmYear{2025}
\acmDOI{}
\acmPrice{}
\acmISBN{}

\acmSubmissionID{112}

\title{Candidate nomination for Condorcet-consistent voting rules}

\author{Ildik\'o Schlotter}
\affiliation{
  \institution{HUN-REN Centre for Economic and Regional Studies}
  \city{Budapest}
  \country{Hungary}}
\email{schlotter.ildiko@krtk.hun-ren.hu}

\author{Katar\'ina Cechl\'arov\'a}
\affiliation{
  \institution{%
  P.J. \v Saf\'arik University}
  \city{Ko\v sice}
  \country{Slovakia}}
\email{katarina.cechlarova@upjs.sk}

\begin{abstract}
Consider elections where the set of candidates is partitioned into parties, and each party must nominate exactly one candidate.  The \PP{} problem asks whether some candidate of a given party can become the winner of the election for some nominations from other parties.  We perform a multivariate computational complexity analysis of \PP{}  for a range of Condorcet-consistent voting rules, namely for Copeland$^\alpha$ for $\alpha \in [0,1]$ and Maximin. 
The parameters we study are the number of voters, the number of parties, and the maximum size of a party. For all voting rules under consideration, we obtain dichotomies based on the number of voters, classifying $\NP$-complete and polynomial-time solvable cases. Moreover, for each $\NP$-complete variant, we determine the parameterized complexity of every possible parameterization with the studied parameters as either (a) fixed-parameter tractable, (b) $\W1$-hard but in~$\XP$, or (c) $\paraNP$-hard, 
outlining the limits of tractability for these problems. 
\end{abstract}


\newcommand{\grayline}{\arrayrulecolor{lightgray}\hline\arrayrulecolor{black}}

\newcommand{\linkproof}[1]{%
\ifshort
    $\star$%
\else
    \hyperref[#1]{$\star$}%
\fi 
}

\ifnocomment
\newcommand{\KK}[1]{}
\newcommand{\K}[1]{}
\newcommand{\I}[1]{}
\newcommand{\II}[1]{}
\else 
\newcommand{\KK}[1]{\todo[color=red!10!white,fancyline]{K: #1}}
\newcommand{\II}[1]{\todo[color=blue!10!white,fancyline]{I: #1}}
\newcommand{\K}[1]{\textcolor{red}{#1}}
\newcommand{\I}[1]{\textcolor{blue}{#1}}
\fi         

\newcommand{\new}[1]{\textcolor{red}{#1}}

\usetikzlibrary{shapes.geometric} 
    \tikzset{
        every node/.style={inner sep=0.3pt},
        single/.style = {draw,circle, 
        minimum size=6pt},
        couple/.style = {draw,ellipse,
            minimum height=6pt, 
            minimum width=12pt},
        room/.style = {draw,rectangle, 
            minimum height=6pt,
            minimum width=12pt},
    }
    \definecolor{lightcyan}{RGB}{153,221,255} 
    \colorlet{bordercol}{lightcyan}

\newcommand{\Cpl}[1][\alpha]{\mathsf{Cpl}_{\EE}^{#1}}
\newcommand{\Cplodd}[1][\EE]{\mathsf{Cpl}_{#1}}
\newcommand{\MM}{\mathsf{MM}}
\newcommand{\wt}[1]{\widetilde{#1}}
\newcommand{\ora}[1]{\overrightarrow{#1}}
\newcommand{\ola}[1]{\overleftarrow{#1}}
\newcommand{\ol}[1]{\overline{#1}}
\newcommand{\copa}[1]{#1 \odot 1}
\newcommand{\copb}[1]{#1 \odot 2}
\newcommand{\copc}[1]{#1 \odot 3}
\newcommand{\copgen}[2]{#1 \odot #2}
\newcommand{\olsi}[1]{\,\overline{\!{#1}}}
\renewcommand{\P}{\mathcal{P}}
\newcommand{\mytabref}[1]{{\small (#1)}}

\def\PP{\textsc{Possible President}}

\def\LabSubDigraph{\textsc{Partitioned Subdigraph Isomorphism}}
\def\MMC{\textsc{Maximum Matching with Couples}}

\def\E{\mathcal{E}}
\def\C{\mathcal{C}}

\def\mt{t\kern-0.035cm\char39\kern-0.03cm}
\def\ml{l\kern-0.035cm\char39\kern-0.03cm}
\def\md{d\kern-0.035cm\char39\kern-0.03cm}
\def\EE{{\mathcal{E}}}

\def\SS{\mathcal{S}}
\def\pii{\wt{\pi}}
\def\tt{\wt{t}}
\def\cc{\wt{c}}
\def\varphii{\wt{\varphi}}
\def\myDB{\texttt{D} \mkern-3.0mu \texttt{B}}
\def\myDC{\texttt{D} \mkern-3.0mu \texttt{C}}
\def\myCA{\texttt{C} \mkern-3.0mu \texttt{A}}
\def\myBX{\texttt{B} \mkern-3.0mu \texttt{S}}
\def\cop{\textup{copy}}
\def\orig{\textup{orig}}

\def\NP{\mathsf{NP}}
\def\W1{\mathsf{W}[1]}
\def\paraNP{\textup{para-}\mathsf{NP}}
\def\XP{\mathsf{XP}}

\def\abase{\underline{a}}
\def\bbase{\underline{b}}
\def\cbase{\underline{c}}
\def\AlgMM{\textup{\textsf{AlgMM}}}

\def\na{--}
\def\maxsize{\sigma}

\theoremstyle{acmplain}
\newtheorem{observation}{Observation}
\newtheorem{claim}{Claim}

\newcommand{\opentriangle}{%
  \raisebox{0.2pt}{\makebox[0.77778em]{%
    \setlength{\unitlength}{0.6em}%
    \linethickness{0.4pt}\roundjoin
    \begin{picture}(1,1)
    \polygon(0,0)(1,0)(1,1)
    \end{picture}%
  }}%
}

\newenvironment{proofsketch}{\par
  \pushQED{\hfill \opentriangle}%
  \normalfont \topsep 12pt\relax
  \trivlist
  \item[\indent\hskip\labelsep
        {\scshape Proof sketch.}]\ignorespaces
}{%
  \popQED\endtrivlist
}

\newenvironment{claimproof}{\par
  \pushQED{$\hfill \lhd$}%
  \normalfont \topsep 6pt\relax
  \trivlist
  \item[\indent\hskip\labelsep
        {\scshape Claim proof.}]\ignorespaces
}{%
  \popQED\endtrivlist
}

\newtheoremstyle{myrulestyle}%
{}{}{}{}
{\bfseries}{. }%
{ }{}%
\theoremstyle{myrulestyle}
\newtheorem{myrule}{Rule}

\newtheorem{mylatrule}{Rule} %

\begin{document}

\pagestyle{fancy}
\fancyhead{}

\maketitle

\section{Introduction}

Political elections are always preceded by a turbulent phase where parties select their nominated candidates for the upcoming election. Clearly, this process has a great influence on the outcome of the election, and therefore it is only natural that political parties engage in all kinds of strategic behavior when choosing their nominees. We focus on the case which models presidential elections in the sense that each party needs to nominate exactly one person among its possible candidates for presidency. 

A naive approach would expect each party to simply choose its ``best'' candidate---however, in practice it is rarely the case that there is a single candidate that can be considered the best in all scenarios.  
Indeed, a given party may find that different candidates have different chances of winning the upcoming election depending on the nominees of the remaining parties. 
Parties may elect their nominees through primaries (an approach studied by Borodin et al.~\cite{Borodin24}), but a more careful process may take into account the estimated preferences of all voters over the possible nominees, and not only the preferences of party members.

Following the formal model of candidate nomination proposed by Faliszewski et al.~\cite{faliszewski2016}, we assume that the preferences of all voters over all potential candidates are known, and in the {\it reduced election} obtained as a result of each party nominating a unique candidate, the preferences of each voter over these nominees are simply the restriction of its preferences over the whole pool of candidates.  Faliszewski et al.\ asked two natural questions: 
the {\sc Possible President} problem asks whether a given party can nominate some candidate~$c$ in such a way that $c$ can become the winner of the election for \emph{some} nominations from the remaining parties,
and the {\sc Necessary  President} problem
asks whether some nominee~$c$ of the given party will be the winner irrespective of all other nominations. 

In this paper we study
the \PP\ problem in election systems that use some \emph{Condorcet-consistent} voting rule.
A candidate that defeats all other candidates in a pairwise comparison is called the {\it Condorcet winner},
and voting rules that always choose the Condorcet winner if it exists are said to be {\it Condorcet-consistent}. 
We focus our investigations on the Condorcet-consistent voting rules Maximin and Copeland$^\alpha$ for~$\alpha\in[0,1]$.

Condorcet-consistent voting rules are widely used in sports competitions, but have also been applied by e.g., the Pirate Party in Sweden and in Germany, 
and various organizations such as Debian, Gentoo Foundation, and Wikimedia~\cite{Schulze}. 
Foley~\cite{Foley} has suggested to use Condorcet-consistent round-robin voting for primary elections, followed by a general election between the top two candidates, to overcome 
the serious flaw in US presidential elections that the winner may not be the preferred candidate of the majority of voters.

\smallskip
\noindent
{\bf Related Work.}
Faliszewski et al.~\cite{faliszewski2016} dealt only with Plurality, arguably the simplest type of elections, %
and derived several $\mathsf{NP}$-hardness results for both \textsc{Possible} and
\textsc{Necessary President}. 
They also showed that when preferences are single-peaked, {\sc Necessary  President} can be decided in polynomial time. 
By contrast, they found that \PP\ remains $\NP$-complete even for single-peaked preferences, though becomes tractable if the candidates of each party appear consecutively on the societal axis.

Misra \cite{misra2019parameterized} extended the results of Faliszewski et al. by studying the parameterized complexity of {\sc Possible President}.  She examined the number~$t$ of parties as the parameter, and proved that the problem is $\mathsf{W}[2]$-hard and in $\mathsf{XP}$, and 
becomes fixed-parameter tractable (FPT) with parameter~$t$ when restricted to 1D-Euclidean preference profiles. 
She also strengthened previous results by proving that {\sc Possible President} for Plurality is $\NP$-hard even if all parties have size at most two, and preferences are both single-peaked and single-crossing; hence, the problem is $\paraNP$-hard when parameterized by the size of the largest party even on a very restricted domain.
Misra asked whether \PP{} for Plurality is FPT when parameterized by the number of voters;  this questions has been answered negatively by Schlotter et al. in \cite{schlotter2024}.

\PP\ for voting rules other than Plurality have been first treated by Cechl{\'a}rov{\'a} et al.~\cite{cechlarova2023candidates}. Namely, they dealt with    positional scoring rules ($\ell$-Approval, $\ell$-Veto, and Borda)  and  with  Condorcet-consistent rules Copeland, Llull, and Maximin. They proved that {\sc Possible President} is $\mathsf{NP}$-hard for each of these rules, even when the maximum size of a party is two; 
they left the complexity for Copeland$^\alpha$ with $\alpha\in (0,1)$ open. %

Schlotter et al. \cite{schlotter2024} obtained results concerning the parameterized complexity of \PP\ for several classes of positional scoring rules, including Borda and nontrivial generalizations  of $\ell$-Approval and $\ell$-Veto. The parameters they examined were the number of voters, the number of voter types, the number of parties, the maximum size of a party and their combinations.

Further results concerning elections with parties that  nominate candidates have been provided by Lisowski \cite{lisowski2022strategic}. He considered directed graphs called tournaments whose vertices correspond to candidates, and each directed arc~$(a,b)$ indicates that a majority of voters prefers candidate~$a$ to candidate~$b$.
Among others, Lisowski 
observed that it is possible to check whether a given party has a possible Condorcet winner in polynomial time, while the problems to decide whether a Nash equilibrium exists in the associated game and whether a given party has a Condorcet winner in some Nash equilibrium are $\NP$-complete. %

For a broader view on research related to candidate nomination, we refer the reader to 
\ifshort
the full version of our paper\new{~\cite{SC24-fullversion}}.
\else
Appendix~\ref{app:relatedwork}.
\fi

\smallskip
\noindent
{\bf Our contribution.}
We perform a detailed multivariate complexity analysis using the framework of parameterized complexity for \PP\ for two types of Condorcet-consistent voting rules: Copeland$^\alpha$  for every $\alpha\in [0,1]$ and Maximin. 
Our parameters are the following: the number~$n$ of voters, the number~$t$ of parties, and the size~$\maxsize$ of the largest party.
Table~\ref{tab:summary} summarizes our results.

\begin{table}[tb!]
    \centering
    \begin{tabular}{@{\hspace{2pt}}l@{\hspace{6pt}}l@{\hspace{6pt}}l@{\hspace{6pt}}l@{\hspace{6pt}}l@{\hspace{2pt}}}
    \multirow{2}{*}{\# voters}
        & %
        Copeland$^\alpha$ & Copeland$^\alpha$ &
        Maximin & Maximin \\
        & classical & param.~$t$  
         & classical & param.~$t$ \\
        \hline \\[-8pt]
        \multirow{2}{*}{$n=2$ \hspace{16pt} \bigg\{} & $\alpha=1$: $\mathsf{P}$  & \na & \multirow{2}{*}{$\hspace{-9pt}\bigg\} \hspace{5pt} \mathsf{P}$} & \multirow{2}{*}{\na} \\
        & $\alpha<1$: $\NP$-c  
        & open 
        & & \\
        & \mytabref{T\ref{thm:Llull-2voters},T\ref{thm:Copeland-v=s=2}}
        & & \mytabref{T\ref{thm:Maximin_V2}} \\
        \grayline \\[-8pt]
        $n=3$ & $\NP$-c & open & $\mathsf{P}$ & -- \\
        & \mytabref{T\ref{thm:3voters-NPh}}
        & 
        & \mytabref{T\ref{thm:Maximin_V3}}\\
        \grayline \\[-8pt]
        $n\geq 4$ even & $\NP$-c 
        & 
        $\mathsf{W}[1]$-h, $\XP$ 
        & $\NP$-c & $\mathsf{FPT}$ \\
        & \mytabref{T\ref{thm:Copeland-v=s=2},T\ref{thm:4voters-s=2}}
        & \mytabref{T\ref{thm:4voters}} 
        & \mytabref{T\ref{thm:Maximin_hard}}
        & \mytabref{T\ref{thm:MM-fpt}}
        \\
        \grayline \\[-8pt]
        $n\geq 5$ odd & $\NP$-c  & 
        $\mathsf{W}[1]$-h,  $\XP$
        & $\NP$-c 
        & $\mathsf{FPT}$ \\
        & \mytabref{T\ref{thm:3voters-NPh}} 
        & \mytabref{T\ref{thm:5voters-t}} & \mytabref{T\ref{thm:Maximin_hard}}
        & \mytabref{T\ref{thm:MM-fpt}}
    \end{tabular}
    \caption{Summary of our results
    on the classical and parameterized complexity of the 
    \PP\ problem. %
    Our parameterized results for $\NP$-hard cases 
    consider parameter~$t$, the number of parties. ``$\NP$-c'' and ``$\mathsf{W}[1]$-h'' stand for 
    ``$\NP$-complete'' and ``$\mathsf{W}[1]$-hard'', respectively. 
    All our $\NP$-completeness results  hold for maximum party size~$\maxsize=2$.}
    \label{tab:summary}
\end{table}
For Copeland$^\alpha$ elections with $\alpha \in [0,1]$, we obtain a complete computational dichotomy for the complexity of \PP{} as a function of the number of voters: 
\begin{theorem}
\label{thm:fixed-voters}
Let $n$ be a fixed integer and $\alpha \in [0,1]$. 
Then \PP\ for Copeland$^\alpha$ is $\NP$-complete when restricted to instances with~$n$ voters and maximum party size $\maxsize=2$, if
\begin{compactitem}%
    \item[(a)] $n \geq 3$, or 
    \item[(b)] $n = 2$ and $\alpha<1$.
\end{compactitem}
By contrast, \PP\ for Copeland$^1$ (i.e., Llull) restricted to instances with~$2$ voters is polynomial-time solvable.
\end{theorem}

It transpires that \PP\ for Copeland$^\alpha$ for arbitrary~$\alpha \in [0,1]$ is $\paraNP$-hard when parameterized by $n+\maxsize$, i.e., both the number of voters and the maximum party size.
We strengthen this result by showing that parameterizing the problem with $t$, the number of parties, the problem remains $\mathsf{W}[1]$-hard even if the number of voters is a constant $n \geq 4$.
Since the problem is easily solvable in $\maxsize^t n^{O(1)}$ time ~\cite{schlotter2024}, this yields a classification of all parameterized ($\NP$-hard)  variants of \PP\ for Copeland$^\alpha$, with parameters chosen arbitrarily from $\{n,\maxsize,t\}$, as either (i) FPT, (ii) $\W1$-hard and in $\XP$, or (iii) $\paraNP$-hard. 

We remark that despite this complete classification, we leave the computational complexity open for certain constant values of~$n$; namely, we could not resolve the parameterized complexity of \PP\ for Copeland$^\alpha$ for parameter $t$ when $n \in \{2,3\}$.

For the Maximin voting rule,  %
we again obtain a complete dichotomy with respect to the number~$n$ of voters:
\begin{theorem}
\label{thm:Maximin-dichotomy}
    Let $n$ be a fixed  integer. Then  \PP\ for Maximin voting rule for instances with $n$ voters is 
    \begin{itemize}\itemsep0pt
        \item[(a)] polynomial-time solvable if $n\le 3$;
        \item[(b)] $\NP$-complete if $n\ge 4$, even for maximum party size $\maxsize=2$.
    \end{itemize}
\end{theorem}

Contrasting the Copeland$^\alpha$ voting rule, we show that \PP{} for Maximin is FPT when parameterized by the number~$t$ of parties. This tractability result is achieved by a reduction of our problem to a special polynomial-time solvable version of the \textsc{Partitioned Subdigraph Isomorphism} problem.
Thus, our results for Maximin yield a complete classification of all parameterized ($\NP$-hard) variants of the problem as either FPT or $\paraNP$-hard. 
In fact, we settle the complexity of the problem for each variant where $n$ and~$\maxsize$ both may be restricted to \emph{arbitrary} fixed integers as either $\NP$-complete and FPT with~$t$, or polynomial-time solvable.

\smallskip 
\noindent
{\bf Techniques.}
Our algorithmic results use standard techniques from parameterized complexity and algorithmic graph theory. Our hardness results rely on intricate constructions, and we also develop the technique of using so-called \emph{flat elections} with three voters and $m$ candidates where each candidate defeats exactly $\frac{m-1}{2}$  candidates; this method might be of independent interest.

\ifshort
Results marked by ($\star$) have their proofs in the full version~\cite{SC24-fullversion}.
\else 
Results marked with~($\star$)
have their proofs deferred to the appendices; %
the symbol~$\star$ functions as a link to the proof.
\fi

\section{Preliminaries}
\label{sec:prelim}
We use the notation $[i]=\{1,2,\dots,i\}$ for each positive integer $i$.

We assume familiarity with  basic graph theory and
the framework of parameterized complexity.
Besides providing all necessary definitions in 
\ifshort the full version~\new{\cite{SC24-fullversion}},
\else Appendix~\ref{app:graphs},
\fi 
we refer the reader to the books~\cite{downey-fellows-FPC-book,cygan2015parameterized} for an introduction into parameterized complexity, and to the books~\cite{diestel-book,BangjensenGutin-book}
for the standard notation on graphs we adopt.

\smallskip
\noindent
{\bf Elections.}
An election $\E=(C,V,\{\succ_v\}_{v \in V})$ consists of a finite set~$C$ of candidates, a finite set~$V$ of voters, and the preferences of voters over candidates. 
We assume that the preferences of each voter~$v$ are represented by a strict linear order~$\succ_v$  over~$C$, where $c\succ_v c'$ means that voter $v$ \emph{prefers} candidate $c$ to candidate $c'$.
We denote the set of all elections over a set~$C$ of candidates by ${\mathbb E}_C$.
A \emph{voting rule} $f:{\mathbb E}_C\to 2^C$ chooses a set of \emph{winners} of the election. 

Our model also includes a partition ${\mathcal P}=\{P_1,\dots,P_t\}$ of the set~$C$ of candidates; each set $P_j$ is interpreted as a \emph{party} that has to decide whom among its potential candidates to nominate for the election.

Formally, a \emph{reduced election}  arises after all parties have nominated a unique candidate, leading to a reduced candidate set $C'\subseteq C$ such that $|C'\cap P_j|=1$ for each~$j\in [t]$. We can then define the reduced election as ${\mathcal E}_{C'}=(C',V,\{\succ'_v\}_{v \in V})$ where the preference relation~$\succ'_v$ of each voter $v\in V$ is the restriction of her original preference relation~$\succ_v$ %
to~$C'$. 

Now we formulate our problem of interest, as introduced in \cite{faliszewski2016}.

\begin{center}
\fbox{ 
\parbox{0.96\columnwidth}{
Problem {\bf \PP} for voting rule~$f$: \\
{\bf Input:} An election $\E=(V,C,\{\succ_v\}_{v \in V})$ with a set~$V$ of voters and a set~$C$ of candidates, a partition~${\mathcal P}$ of~$C$ into parties, and a distinguished party $P^\star \in \P$. \\
{\bf Question:} Is there a candidate~$p \in P^\star$ such that for some nominations of other parties leading to a reduced candidate set~$C'$, $p$ is the \emph{unique} winner of the reduced election ${\mathcal E}_{C'}$ according to~$f$?
}}
\end{center}

Notice that we consider the \emph{unique winner model}, i.e., we aim for a set of nominations that yield  $f(\E_{C'})=\{p\}$ for the candidate~$p$ nominated by the distinguished party in the reduced election $\E_{C'}$.

\smallskip
\noindent
{\bf Voting rules.}
In this paper we shall concentrate on two Condorcet-consistent rules. 
 For two candidates $a,b\in C$, we let
 $N_\EE(a,b)$ denote the number of voters who prefer candidate~$a$ to candidate~$b$ in election~$\EE$; we shall omit the subscript when $\EE$ is clear from the context. If $N_\EE(a,b)>N_\EE(b,a)$ we say that candidate~$a$ {\it defeats} candidate~$b$ in~$\E$; if  $N_\EE(a,b)=N_\EE(b,a)$ and $a \ne b$, then candidates~$a$ and~$b$ are {\it tied} in $\E$. %
 The \emph{Condorcet winner} %
 is the candidate that defeats all other candidates; %
 a voting rule is \emph{Condorcet consistent}, if it always selects the Condorcet winner whenever it exists.

The \emph{Copeland$^\alpha$ voting rule}, as defined by Faliszewski et al. \cite{faliszewski2009llull},
assigns to some candidate~$a$ a score of~$1$ for each candidate defeated by~$a$, and a score of~$\alpha$ for each candidate tied with~$a$, so 
the Copeland$^\alpha$-score of~$a$  is 
$\Cpl(a)=\sum_{b \textrm{ defeated by }a} 1 + \sum_{b \textrm{ tied with }a} \alpha$ in an election~$\EE$.
The winners of~$\EE$ %
are all candidates with the maximum score. The voting rule obtained for $\alpha=1$ is called the \emph{Llull} rule, and %
we refer to the case $\alpha=0$ as the \emph{Copeland} rule.

In the \emph{Maximin voting rule}, the Maximin-score of candidate~$a$ in election~$\EE$ over candidate set~$C$ is $\MM_\EE(a)=\min_{b\in C \setminus \{ a\}}N_\EE(a,b)$, and the winners  of~$\EE$ are again the candidates with maximum score.

Notice that Copeland$^\alpha$ as well as Maximin winners can be computed efficiently for any election. Therefore it is easy to see that \PP{} for these voting rules belongs to the class $\mathsf{NP}$.

\section{Copeland$^\alpha$ voting rule}
\label{sec:Copeland}

If there are only two voters, in Section~\ref{sec:2voters} we show that \PP{} for Copeland$^\alpha$ is polynomially solvable if $\alpha=1$,  %
(Theorem~\ref{thm:Llull-2voters}), but $\mathsf{NP}$-hard if $\alpha<1$  (Theorem~\ref{thm:Copeland-v=s=2}). 

For three voters, we show in Section~\ref{sec:3voters} that Copeland is $\mathsf{NP}$-complete (Theorem~\ref{thm:3voters-NPh}). If the number of voters is odd, then no ties occur, and hence this result holds for Copeland$^\alpha$ for any $\alpha\in[0,1]$. The proof is quite involved, and provides a reduction from a special variant of the $\mathsf{NP}$-complete problem \textsc{Maximum Matching with Couples}, using the crucial notion of {\it flat} elections. 

The case with four or more voters is treated shortly in Section~\ref{sec:min4voters}.

We address the complexity of \PP\ for Copeland   when parameterized by the number of parties in Section~\ref{sec:parameter-t}.
\subsection{Two Voters}
\label{sec:2voters}
Let us first show that \PP\ for the Llull voting rule is easy if there are only two voters. 
The key observation that yields tractability is that the ``defeat'' relation is transitive for two voters: 

\begin{observation}
    \label{obs:2voters_transitivity}
    In an election with two voters, if candidate $a$ defeats candidate~$b$, and $b$ defeats candidate~$c$, then $a$ also defeats $c$. 
\end{observation}
\begin{proof}
Since $b$ must follow~$a$ in the preference lists of both voters, and $c$ must follow $b$ in both lists too, we immediately know that $c$ follows $a$ in the preference list of both voters.     
\end{proof} 

\begin{restatable}[\linkproof{app:proof-of-Llulltwovoters}]{theorem}{thmLlulltwovoters}
\label{thm:Llull-2voters}
    \PP{} for the Llull voting rule is polynomial-time solvable if there are only two voters.
\end{restatable}

\begin{proofsketch}
Using Observation~\ref{obs:2voters_transitivity}, one can prove that according to the Llull rule, some candidate~$p$ %
can be a unique winner in a reduced election~$\EE$ 
if and only if $p$ defeats every other nominee: intuitively, assuming that $p$ is the unique winner because every nominee other than~$p$ is defeated by some other nominee, we arrive at a cycle in the defeat relation, a contradiction showing that $p$ can become the unique winner only by defeating all nominees. 

This 
offers a quadratic-time algorithm to solve \PP{} for Llull voting with two voters: we check for each candidate~$p$ in the distinguished party~$P$ whether $p$ can become the unique winner, which happens if and only if every other party contains 
at least one candidate that is 
defeated by~$p$.
\end{proofsketch}

By contrast, a reduction from  \textsc{3-Coloring} shows that Copeland$^\alpha$ for $\alpha<1$ is intractable already for two voters, even if $\maxsize=2$.

\begin{restatable}[\linkproof{sec:proof-of-CopelandtwoNPc}]{theorem}{thmCopelandtwovoters}
    \label{thm:Copeland-v=s=2}
    For each $\alpha \in [0,1)$, \PP{} for Copeland$^\alpha$ is $\NP$-complete even for instances with two voters and maximum party size $\maxsize=2$.
\end{restatable}

\subsection{Three Voters}
\label{sec:3voters}

As already mentioned, for an odd number of voters no two candidates can be tied, so the value of~$\alpha$ is irrelevant, and the Copeland and Llull voting rules coincide. %
We show the following.

\begin{restatable}[\linkproof{sec:proof-of-3votersNPh}]{theorem}{thmthreevotersNPh}    
    \label{thm:3voters-NPh}
    \PP\ for Copeland is $\NP$-complete even for three voters and maximum party size $\maxsize=2$.
\end{restatable}

To show Theorem~\ref{thm:3voters-NPh}, we will reduce from a special case of an $\mathsf{NP}$-complete problem \textsc{Maximum Matching with Couples}, described in Section~\ref{sec:MC-variant}.
We present the most important ingredient of the reduction, the notion of \emph{flat elections} %
in Section~\ref{sec:flat-elections}, and 
follow with a sketch of the reduction in Section~\ref{sec:3votersreduction}.

\subsubsection{A special case of \textsc{Maximum Matching with Couples}}
\label{sec:MC-variant}

We are going to reduce from a variant of the following problem called \textsc{Maximum Matching with Couples}. This problem involves a set~$S$ of \emph{singles}, a set $C$ of \emph{couples}\footnote{Although in the context of elections $C$ denotes the set of candidates, this slight clash of notation will not cause any confusion.} and a set~$R$ of rooms. Each room has capacity~2, meaning that it can accommodate either a couple or at most two singles. 
Moreover, we need to match everyone to a room that they find \emph{acceptable}, where acceptability is described by a bipartite graph~$G=((S \cup C)\uplus R,E)$.
A \emph{complete matching}\footnote{Note that we do not require $M$ to be a matching in the classic graph-theoretic sense, since we allow edges in~$M$ to share endpoints in $R$.} in~$G$ is then an edge set~$M \subseteq E$ that contains exactly one edge incident to each vertex in~$S\cup C$ and satisfies $|M(r) \cap S|+2|M(r) \cap C| \leq 2$ for each room~$r \in R$, 
where ${M(r)=\{x \in S \cup C \colon rx \in M\}}$ denotes the set of singles and couples \emph{matched to}~$r$.
It is known that the following problem is $\NP$-complete~\cite{biro-mcdermid-matching-sizes1,biro-mcdermid-matching-sizes}.

\begin{center}
\fbox{ 
\parbox{0.96\columnwidth}{
Problem {\bf \textsc{Maximum Matching with Couples}}: \\
{\bf Input:} Sets~$S$, $C$, and~$R$ of singles, couples and rooms, respectively, and a bipartite graph $G=((S \cup C)\uplus R,E)$. \\
{\bf Question:} Is there a complete matching in~$G$?
}}
\end{center}

We shall use a special case of \textsc{Maximum Matching with Couples} as specified in Theorem~\ref{thm:MMC-deg3}. The proof of its $\mathsf{NP}$-completeness relies on a series of simple reduction rules that transform any instance into an equivalent one, achieving the properties required in Theorem~\ref{thm:MMC-deg3} step by step; see 
\ifshort the full version~\new{\cite{SC24-fullversion}}.
\else Appendix~\ref{app:MC-variant}.
\fi

\begin{restatable}[\linkproof{app:MC-variant}]{theorem}{mmcvariant}
\label{thm:MMC-deg3}    
\textsc{Maximum Matching with Couples} remains $\NP$-complete even if $|R|=|S|/2+|C|$, and 
\begin{itemize}
    \item each vertex in the input graph has degree~$2$ or~$3$, and
    \item each room adjacent to both singles and couples is adjacent to exactly two singles and one couple.
\end{itemize} 
\end{restatable}

\subsubsection{Flat elections with three voters}
\label{sec:flat-elections}

Working towards a reduction from \textsc{Maximum Matching with Couples} to \PP\ for Copeland voting with three voters, we next present a construction for an election~$\EE_q$ over $3^q$ candidates for some $q \in \mathbb{N}^+$ and with three voters, in which every candidate defeats the same number of candidates. We will call such elections \emph{flat}, i.e., an election is flat if all candidates receive the same Copeland-score. An election with $m$ candidates where $m$ is odd can only be flat if each candidate defeats~$\frac{m-1}{2}$ other candidates. To see this, consider the tournament underlying the election: clearly, we can only have all out-degrees equal to some~$d$, 
if the tournament has $m \cdot d$ arcs, i.e., $m \cdot d=\binom{m}{2}$.

We propose a recursive construction for $\EE_q$ in Definition~\ref{def:EEq}. 

\begin{definition}
\label{def:EEq}
Let the candidate set of~$\EE_1$ be $C_1=\{\abase,\bbase,\cbase\}$, and 
let $w$, $w'$, and $w''$ be our three voters with preferences 
\begin{align*}
    w &: \abase,\bbase,\cbase; \\
    w' &: \cbase,\abase,\bbase; \\
    w'' &: \bbase,\cbase,\abase.
\end{align*}
Notice that $\abase$ defeats~$\bbase$, $\bbase$ defeats~$\cbase$, and $\cbase$ defeats~$\abase$. Therefore, each of the candidates obtains a Copeland$^\alpha$-score of~$1$.

For $q \geq 1$, we are going to reuse the candidate set $C_q$ of the election~$\EE_q$ to construct the candidate set~$C_{q+1}$ of~$\EE_{q+1}$ 
by introducing three copies of each candidate~$c \in C_q$ which will be denoted by~$\copa{c},\copb{c}$, and $\copc{c}$. 
Let $L_q(w)$, $L_q(w')$, and $L_q(w'')$ denote the preference lists of voters~$w$, $w'$, and~$w''$, respectively, in~$\EE_q$. 
For a list~$L$ of candidates from~$C_q$ and each $h \in [3]$, let us denote by~$\copgen{L}{h}$ the list obtained from~$L$ by replacing each candidate~$c$ in~$L$ by its {$h$-th} copy~$\copgen{c}{h}$. Using this notation, we are now ready to define the preferences of the voters in~$\EE_{q+1}$:
\begin{equation}
\label{prefs-of-EEq}
\begin{array}{rlll}
    w: & \copgen{L_q(w)}{1}, & \copgen{L_q(w)}{2}, & \copgen{L_q(w)}{3}; \\
    w': & \copgen{L_q(w')}{3}, & \copgen{L_q(w')}{1}, & \copgen{L_q(w')}{2}; \\
    w'': & \copgen{L_q(w'')}{2}, & \copgen{L_q(w'')}{3}, & \copgen{L_q(w'')}{1}.    
\end{array}
\end{equation}
\end{definition}
\smallskip

Notice that each candidate in~$\EE_{q+1}$ is then of the form 
\begin{equation}
\label{eqn:candidate-form}
    (\copgen{(\copgen{(\copgen{x}{h_1})}{h_2})  \cdots \odot h_{q-1})}{h_q}
\end{equation} for some $x \in \{\abase,\bbase,\cbase\}$ and indices $h_1,h_2,\dots,h_q \in [3]$.

We will say that two candidates~$c$ and~$c'$ in~$\EE_{q+1}$, having the form~(\ref{eqn:candidate-form}) for $x$ and~$x'$ in~$\{\abase,\bbase,\cbase\}$ and indices $h_1,\dots,h_q$ and~$h'_1,\dots,h'_q$ from~$[3]$, respectively, 
\emph{belong to the same group at level~$q'$} 
for some $q' \in [q]$, if~$h_i=h'_i$ for each $q' \leq i \leq q$;
accordingly,  we define a \emph{$q'$-level group} as a maximal set of candidates that belong to the same group at level~$q'$.
Notice that restricting the election $\EE_{q+1}$ to a $q'$-level group, we obtain a copy of the election~$\EE_{q'}$.

In particular, restricting $\EE_{q+1}$ to a $q$-level group, that is, to the set of candidates contained in $\copgen{L_q(w)}{h}$ for some $h \in [3]$, 
we obtain a copy of the election~$\EE_q$.
Observing the preferences of the voters as given in~(\ref{prefs-of-EEq}), 
the following facts are immediate:
\begin{observation}
\label{obs:EEq-properties}
    For each $q \in \mathbb{N}$, the election~$\EE_{q+1}$ has the following properties: 
    \begin{itemize}
        \item each candidate in~$\copgen{L_q(w)}{1}$ defeats all candidates in~$\copgen{L_q(w)}{2}$, and is defeated by all candidates in~$\copgen{L_q(w)}{3}$;
        \item each candidate in~$\copgen{L_q(w)}{2}$ defeats all candidates in~$\copgen{L_q(w)}{3}$, and is defeated by all candidates in~$\copgen{L_q(w)}{1}$;
        \item each candidate in~$\copgen{L_q(w)}{3}$ defeats all candidates in~$\copgen{L_q(w)}{1}$, and is defeated by all candidates in~$\copgen{L_q(w)}{2}$.
    \end{itemize}
    Furthermore, for each $h \in [3]$ and each 
    $c,c' \in C_q$,
    candidate~$\copgen{c}{h}$ defeats candidate~$\copgen{c'}{h}$ in~$\EE_{q+1}$ if and only if~$c$ defeats~$c'$ in~$\EE_q$. 
\end{observation}

By Observation~\ref{obs:EEq-properties}, $\EE_{q+1}$ for some $q \in \mathbb{N}^+$ is flat if and only $\EE_q$ is flat. Since $\EE_1$ is flat, we  obtain the following consequence.

\begin{observation}
\label{obs:EEq-number-of-defeats}
    For each integer $q \geq 1$, every candidate in~$C_q$ defeats
    $\frac{|C_q|-1}{2}=\frac{(3^q-1)}{2}$ candidates in~$\EE_q$, so $\EE_q$ is a flat election with $3^q$ candidates.
    Moreover, no candidate is preferred to another candidate by all three voters in~$\EE_q$. 
\end{observation}

\subsubsection{Reduction for Theorem~\ref{thm:3voters-NPh}}
\label{sec:3votersreduction}
    We present a reduction from the variant of \textsc{Maximum Matching with Couples} %
    described in Theorem~\ref{thm:MMC-deg3}. Let $G=((S\cup C)\uplus R,E)$ be the input graph.
    
\medskip
\noindent
{\bf High-level description.}
The main ideas of the reduction are the way flat elections are used. First, we need a large enough set~$T$ of \emph{teams}, over which we have a flat election involving three voters. Each team in~$T$ will be either a single, a couple or its copy, a room, or a dummy, and will be
eventually be replaced by a set of candidates, depending on its type. We will also add a set~$A \cup B$ of \emph{simple candidates}, and we fix a simple candidate~$a_1$ to form the distinguished singleton party in the constructed instance. 

Based on our flat election over~$T$, we do three modifications: (i) we insert the simple candidates, (ii) we substitute each team in~$T$ with the corresponding candidate lists, and (iii) we move our distinguished candidate~$a_1$ ``to the left'' so that it gains one extra point in the election. 
The crux of the reduction is to ensure that in the obtained election, $a_1$ can become the unique winner if and only if restricting the election to the \emph{relevant} candidates (those that are associated with some team in~$T$) yields a flat election. By carefully designing the candidate set corresponding to each team and their ordering within the preference lists (used during the substitution step), we will ensure that the relevant candidates can form a flat election if and only if our instance of \MMC\ admits a complete matching.
    
\smallskip
\noindent
{\bf Candidates and parties.}
First, we define a party~$P_r=\{r,r'\}$ for each room $r \in R$.  
Next, for each vertex~$p \in S \cup C$ adjacent to~$r$ in~$G$, we introduce a party~$P_p^r=\{p^r,\neg p^r\}$.
Additionally, we define two candidates~$p$ and~$p'$ for each $p \in S \cup C$; if $p$ has degree~$3$ in~$G$, then these two candidates form a single party, and if $p$ has degree~$2$ in~$G$, then $p$ and $p'$ both form their own singleton party.
This way, we associate  four parties with each single~$s \in S$:
\begin{itemize}
    \item $P_s^{r_1}, P_s^{r_2}, \{s\}, \{s'\}$ if $N_G(s)=\{r_1,r_2\}$,
    \item $P_s^{r_1}, P_s^{r_2}, P_s^{r_3}, \{s,s'\}$ if $N_G(s)=\{r_1,r_2,r_3\}$
\end{itemize}
where $N_G(v)$ denotes the neighborhood of a vertex~$v$ in~$G$.
Similarly, there are four parties associated with each couple~$c \in C$:
\begin{itemize}
    \item $P_c^{r_1}, P_c^{r_2}, \{c\}, \{c'\}$ if $N_G(c)=\{r_1,r_2\}$;
    \item $P_c^{r_1}, P_c^{r_2}, P_c^{r_3}, \{c,c'\}$  if $N_G(c)=\{r_1,r_2,r_3\}$.
\end{itemize}

Next, for each couple~$c \in C$, we introduce a copy~$\hat{x}$ for each candidate~$x$ associated with the couple~$c$, yielding a candidate set $\{\hat{c}^r,\neg \hat{c}^r: r \in N_G(c)\} \cup \{\hat{c},\hat{c}'\}$.  
We write $\hat{C} =\{\hat{c}:c \in C\}$. 
With each~$\hat{c} \in \hat{C}$ we associate the parties $P^r_{\hat{c}}=\{\hat{x}:x \in P^r_c\}$ for each $r \in N_G(c)$, plus one or two parties formed by $\hat{c}$ and~$\hat{c}'$, depending on whether $c$ has degree two or three in~$G$, so that altogether there are four parties associated with~$\hat{c}$ (as for~$c$).
For practical purposes, we extend the notation by setting $N_G(\hat{c}):=N_G(c)$ for each $c \in C$.

    We also fix an arbitrary set~$D$ of \emph{dummy teams} %
    whose size is the smallest non-negative integer for which 
    $\rho:=|R|+|S|+2|C|+|D|=3^q$ for some~$q \in \mathbb{N}^+$, and
    introduce candidates~$a_d,b_d,c_d$ for each $d \in D$, each of them forming its own singleton party. 
    Since for each positive integer~$n$ there is a power of~$3$ in the interval $[n,3n]$ (this is easily shown by induction on~$n$),
    we get $|D|\leq 2(|S|+2|C|+|R|)$.

    We call the candidates defined so far  \emph{relevant candidates}, and denote their set as~$X$.
    We further define \emph{simple candidates} $a_1,\dots,a_{3\rho}$ and $b_1,\dots,b_{3\rho}$, each of them forming its own singleton party. %
    We will write $A=\{a_1,\dots,a_{3\rho}\}$ %
    and
    $B=\{b_1,\dots,b_{3\rho}\}$.
    Notice that the maximum party size is $\maxsize=2$ in~$G$,
    and the number of parties is $|R|+4|S|+8|C|+3|D|+6\rho=9\rho$.
    Our distinguished party is $\{a_1\}$.

\smallskip
\noindent
{\bf Teams and their lists.}
We refer to the set 
$T=S \cup C \cup \hat{C} \cup R \cup D$
as the set of \emph{teams}.
To define the preferences of our voters, $v$, $v'$, and~$v''$, we introduce for each team $t \in T$ three lists that we call \emph{team lists} and denote by $F_t,F'_t$, and~$F''_t$.
Each of these three lists contains the same set candidates that we associate with~$t$.
        
Let us start with defining the team lists for each room team~$r \in R$. 
First, if room~$r$ is adjacent to singles~$s_1$ and~$s_2$ and a couple~$c$ in~$G$, then we set its team list according to~(\ref{eqn:def-Fr-mixed})
below (to the left). Second, if room~$r$ is adjacent to two or three singles, $s_1, s_2$ and possibly~$s_3$, and no couples in~$G$, then we set its team list as in~(\ref{eqn:def-Fr-singles}).

\begin{minipage}{.47\columnwidth}
    \begin{equation}
    \begin{tabular}{@{}r@{\hspace{2pt}}l@{}}    
    \multicolumn{2}{@{}l@{}}{if $N_G(r)=\{s_1,s_2,c\}$:} \\[3pt]
    $F_r$ & $=s_1^r, r, s_2^r, c^r, r', \hat{c}^r$; \\
    $F'_r$ & $= s_2^r, s_1^r, \hat{c}^r, c^r, r, r'$; \\    
    $F''_r$ & $= r, r', s_2^r, s_1^r, \hat{c}^r, c^r$;    
    \end{tabular}
    \label{eqn:def-Fr-mixed}
    \end{equation}
\end{minipage}
\hfill
\begin{minipage}{.47\columnwidth}
    \begin{equation}
    \begin{tabular}{@{}r@{\hspace{2pt}}l@{}}    
    \multicolumn{2}{@{}l@{}}{if $N_G(r)=\{s_1,s_2,(s_3)\}$:} \\[3pt]
    $F_r$ & $= s_1^r, r, s_2^r, r' , (s^r_3)$; \\
    $F'_r$ & $= (s^r_3), s_2^r, s_1^r,  r, r'$; \\    
    $F''_r$ & $= r, r', (s_3^r), s_2^r, s_1^r$.
    \end{tabular}
    \label{eqn:def-Fr-singles}
    \end{equation}
\end{minipage}    
\vspace{2pt}     

    Third, if room~$r$ is adjacent to two or three couples, $c_1, c_2$ and possibly~$c_3$, and no singles in~$G$, then we set
    \begin{equation}
    \begin{array}{@{}r@{\hspace{2pt}}l@{}}    
    F_r &= c^r_1,c^r_2,(c^r_3), r, r', \hat{c}^r_1, \hat{c}^r_2, (\hat{c}^r_3); \\
    F'_r &= \hat{c}^r_1, c^r_1, \hat{c}^r_2, c^r_2, (\hat{c}^r_3), (c^r_3), r, r'; \\    
    F''_r &= r, r', (\hat{c}^r_3), (c^r_3), \hat{c}^r_2, c^r_2, \hat{c}^r_1, c^r_1.
    \end{array}
    \label{eqn:def-Fr-couples}
    \end{equation}
   In lists~(\ref{eqn:def-Fr-singles}) and~(\ref{eqn:def-Fr-couples}), candidates written within parenthesis may not exist, in which case they should be ignored.
    
    Next, consider a team $p \in S \cup C \cup \hat{C}$. 
    We set the team lists for~$p$ depending on the degree of~$p$ in~$G$:
    
    \begin{minipage}{.5\columnwidth}
    \begin{equation}
    \begin{tabular}{@{}r@{\hspace{2pt}}l@{}}    
    \multicolumn{2}{@{}l@{}}{if $N_G(p)=\{r_1,r_2,r_3\}$:} 
     \\[3pt]
    $F_p$ & $=p, p', \neg  p^{r_1}, \neg  p^{r_2}, \neg  p^{r_3}$; \\
    $F'_p$ & $= \neg  p^{r_3}, p', \neg  p^{r_2}, p, \neg  p^{r_1}$; \\    
    $F''_p$ & $ = \neg  p^{r_1},\neg  p^{r_2}, \neg  p^{r_3}, p,p'$. \\[2pt]
    \end{tabular}
    \label{eqn:def-Fp-3choices}
    \end{equation}
    \end{minipage}
    \hfill%
    \begin{minipage}{.4\columnwidth}
    \begin{equation}\hspace{-8pt}
    \begin{tabular}{@{}r@{\hspace{2pt}}l@{}}    
    \multicolumn{2}{@{}l}{if $N_G(p)=\{r_1,r_2\}$:} 
     \\[3pt]
    $F_p$ & $= p, p', \neg  p^{r_1}, \neg  p^{r_2}$; \\
    $F'_p$ & $= \neg  p^{r_1}, \neg  p^{r_2}, p, p'$; \\    
    $F''_p$ & $=  p', \neg  p^{r_1}, \neg  p^{r_2}, p$. \\[2pt]
    \end{tabular}
    \label{eqn:def-Fp-2choices}
    \end{equation}
    \end{minipage}
    \vspace{2pt}

    Finally, for each dummy team $d \in D$, we let 
    \begin{equation}
    \begin{array}{@{}r@{\hspace{2pt}}l@{}}    
    F_d &= a_d,b_d,c_d; \\
    F'_d &= c_d,a_d,b_d; \\
    F''_d &= b_d,c_d,a_d.
    \end{array}
    \label{eqn:def-Fd}
    \end{equation}    
    This finishes the definition of the team lists~$F_t,F'_t$, and~$F''_t$ for each team $t \in T$.
Observe that the sets of candidates in~$F_t$ taken over each $t \in T$ form a partition of the set~$X$ of relevant candidates. 
    
\medskip
\noindent
{\bf Preferences.}
    In what follows, it will be convenient to fix an ordering over~$T$ and use the notation $T=\{t_1,\dots,t_\rho\}$. 
    Consider the election $\EE_{q}$ introduced in Definition~\ref{def:EEq} over $3^q=\rho$ candidates.
    Since $|T|=3^q$, there exists a bijection~$\psi:C_q \rightarrow T$ between candidates of~$\EE_q$ and teams in~$T$ that maps $t_i \in T$ to the $i$-th candidate in the preference list of~$w$.
    Using the alias $\tt_i=\psi^{-1}(t_i)$ for each team~$t_i \in T$, the election~$\EE_q$ can be written as 
    \begin{equation}
        \label{rewritten-EEq}
    \begin{array}{rl}
        \text{election }\EE_q: \quad 
        w :& \tt_1, \tt_2, \dots,\tt_{\rho}; \\
        w' :& \tt_{\pi(1)},\tt_{\pi(2)},\dots,\tt_{\pi(\rho)}; \\
        w'' :& \tt_{\pii(1)},\tt_{\pii(1)},\dots,\tt_{\pii(\rho)}
    \end{array}
    \end{equation}
    for some permutations~$\pi$ and~$\pii$ over~$[\rho]$.

    We define the permutations~$\varphi$ and~$\varphii$ over~$[3\rho]$ based on the election~$\EE_{q+1}$ similarly: after renaming the candidates in the %
    election~$\EE_{q+1}$ as $\cc_1,\cc_2,\dots,\cc_{3\rho}$, the election~$\EE_{q+1}$ can be re-written as
    \begin{align*}
        \text{election }\EE_{q+1}: \quad 
        w &: \cc_1, \cc_2, \dots,\cc_{3\rho}; \\
        w' &: \cc_{\varphi(1)},\cc_{\varphi(2)},\dots, \cc_{\varphi(3\rho)}; \\
        w'' &: \cc_{\varphii(1)},\cc_{\varphii(1)},\dots, \cc_{\varphii(3\rho)};
    \end{align*}
    for some permutations~$\varphi$ and~$\varphii$ over~$[3\rho]$.

    Now we are ready to give the preferences of voters~$v$, $v'$, and~$v''$:
    \begin{align*}    
    v &: 
    F_{t_1},F_{t_2},
    \dots, F_{t_\rho}, 
    b_1,b_2,
    \dots,b_{3\rho-1}, a_1, b_{3\rho}, 
    a_2,a_3,\dots,a_{3\rho}; \\
    v' &: 
    a_{\varphi(1)},%
    \dots, a_{\varphi(3\rho)}, 
    F'_{t_{\pi(1)}},%
    \dots, F'_{t_{\pi(\rho)}}, 
    b_{\varphi(1)},%
    \dots,b_{\varphi(3\rho)}; \\
    v'' &: 
    b_{\varphii(1)},%
    \dots,b_{\varphii(3\rho)},
    a_{\varphii(1)},%
    \dots, a_{\varphii(3\rho)}, 
    F''_{t_{\pii(1)}},%
    \dots, F''_{t_{\pii(\rho)}}.
    \end{align*}
Hence, the constructed election is obtained from~(\ref{rewritten-EEq}) by substituting each candidate corresponding to some team~$t_i$ with the team lists for~$t_i$, and adding the simple candidates in the appropriate manner.

    It is clear that the construction takes polynomial time, since building the elections~$\EE_q$ and~$\EE_{q+1}$ takes time polynomial in~$3^q$, and $q=\log_3(|T|)$.
    Therefore, it remains to prove its correctness. 

\medskip
\noindent
{\bf Connection between solvability of the input instance and flatness of the election restricted to relevant candidates.}
To prove the correctness of our reduction, let us start with the following facts, which rely on Observation~\ref{obs:EEq-number-of-defeats}.\begin{itemize}
    \item Candidate~$a_1$ defeats~$b_{3\rho}$, all relevant candidates, and no candidate in~$B \setminus \{b_{3\rho}\}$; additionally $a_1$ also defeats exactly half of the candidates in~$A \setminus \{a_1\}$. Therefore, 
    \begin{equation}
    \label{eq:score-of-a1}
    \Cplodd(a_1)=1 +3\rho+\frac{3\rho-1}{2}=\frac{9\rho+1}{2} 
    \end{equation}
    because $X$ is the union of~$3\rho$ parties.\footnote{Henceforth, we write $\Cplodd(x)$ for the score of candidate~$x$ whenever $\alpha$ is irrelevant.}
    \item Candidate~$a_i \in A \setminus \{a_1\}$ defeats all relevant candidates, no candidates in~$B$,  and %
    half of the candidates in~$A \setminus \{a_i\}$. Thus, 
    $\Cplodd(a_i)=3\rho+\frac{3\rho-1}{2}=\frac{9\rho-1}{2}$.
    \item Candidate~$b_{3\rho}$ defeats all candidates in~$A$ except for~$a_1$, no relevant candidates,  and
    half of the candidates in~$B \setminus \{b_{3\rho}\}$. Thus, 
    $\Cplodd(b_{3\rho})=|A|-1+\frac{3\rho-1}{2}=\frac{9\rho-3}{2}$.
    \item Candidate~$b_i \in B \setminus \{b_{3\rho}\}$ defeats all candidates in~$A$, no relevant candidates,  and %
    half of the candidates in~$B \setminus \{b_i\}$. Thus, 
    $\Cplodd(b_i)=|A|+\frac{3\rho-1}{2}=\frac{9\rho-1}{2}$.
    \item Relevant candidates defeat all candidates in~$B$ and no candidates in~$A$.
\end{itemize}
    Due to (\ref{eq:score-of-a1}), the above observations imply that $a_1$ is the unique winner in of the election~$\EE$ resulting from some nominations if and only if all relevant nominees defeat at most $\frac{3\rho-1}{2}$ relevant nominees, i.e., if the election~$\EE$ restricted to relevant nominees is flat.
In other words, our instance of \PP\ is a ``yes''-instance if and only if there exist nominations of all parties corresponding to singles, couples, and rooms for which the \emph{relevant election}~$\EE_X$ below reduced to these nominations becomes flat: 
    \begin{equation}
    \label{prefs-EEX}
    \begin{array}{r@{\hspace{4pt}}l}
    \text{relevant election $\EE_X$: }\quad v &: 
    F_{t_1},F_{t_2},\dots, F_{t_\rho}; \\
    v' &: 
    F'_{t_{\pi(1)}},F'_{t_{\pi(2)}},\dots, F'_{t_{\pi(\rho)}}; \\
    v'' &: 
    F''_{t_{\pii(1)}},F''_{t_{\pii(2)}},\dots, F''_{t_{\pii(\rho)}}.
    \end{array}
    \end{equation}

Recall that~$\psi:C_q \rightarrow T$ is a bijection between candidates of~$\EE_q$ and teams in~$T$. 
Comparing~(\ref{rewritten-EEq}) and~(\ref{prefs-EEX}), we get the following. 
\begin{observation}
\label{obs:bijection-property}
    Replacing each candidate~$z$ in the preference lists of~$w$, $w'$, and~$w''$ in the election~$\EE_q$ with~$F_{\psi(z)}$, $F'_{\psi(x)}$, and~$F''_{\psi(x)}$, respectively, yields exactly the preference lists of voters~$v$, $v'$, and~$v''$ in the relevant election~$\EE_X$.
\end{observation} %
Observation~\ref{obs:bijection-property} enables us to take advantage of the structure of election~$\EE_q$ to establish analogous properties of the constructed instance.
Using the specifics of the team list definitions, we can show that 
our instance of \textsc{Maximum Matching with Couples} admits a complete matching
if and only if 
$\EE_X$ admits nominations resulting in a flat election;
as we have seen, the latter happens
 if and only if the constructed instance of \PP\ is a ``yes''-instance. See 
 \ifshort the full version~\new{\cite{SC24-fullversion}}
 \else Appendix~\ref{app:proof-of-correctness-3voters}
 \fi 
 for the rest of the proof.

\subsection{Four or More Voters}
\label{sec:min4voters}

Contrasting Theorem~\ref{thm:Llull-2voters}, showing the tractability of \PP\ for Llull with two voters, a reduction from \textsc{3-Coloring}
yields $\NP$-hardness for four voters.
As it is possible to add two voters with opposite preferences without changing the election outcome, Theorems~\ref{thm:Copeland-v=s=2}, \ref{thm:3voters-NPh} and~\ref{thm:4voters-s=2} imply Theorem~\ref{thm:fixed-voters}.

\begin{restatable}[\linkproof{sec:proof-of-4votersNPc}]{theorem}{thmfourvotersmaxsizetwo}
    \label{thm:4voters-s=2}
    \PP\ for Copeland$^1$ (i.e., Llull) is $\NP$-complete even for four voters and maximum party size~$\maxsize=2$.    
\end{restatable}

\subsection{Few Parties}
\label{sec:parameter-t}

In this section we consider the parameterization of \PP\ by~$t$, the number of parties. As we will see, intractability persists even if the number of voters is four, and $t$ is a parameter. Our starting point is Theorem~\ref{thm:Copeland-easy-alpha} which shows that \PP\ for Copeland$^\alpha$ for $\alpha<1$ is $\mathsf{W}[1]$-hard with parameter~$t$.

\begin{theorem}
\label{thm:Copeland-easy-alpha}
For any constant $\alpha \in [0,1)$, \PP{} for Copeland$^{\, \alpha}$ is $\mathsf{W}[1]$-hard when parameterized by~$t$.
\end{theorem}
\begin{proof}
We provide a reduction from the \textsc{Multicolored Clique} problem. An instance of this problem consists of a graph $G=(U,E)$ with its vertex set partitioned into $k$ independent sets~$U_1,\dots, U_k$, and the question is whether $G$ contains a clique of size~$k$. \textsc{Multicolored Clique} is $\mathsf{W}[1]$-hard when parameterized by~$k$ \cite{fellows-hermelin-rosamond-vialette-multicolored-hardness,pietrzak-multicolored-2003}.

We construct an instance of \PP\ as follows. The set of candidates is $C=U\cup \{p,p'\}$, our distinguished party is $P=\{p\}$, and we have further parties $P'=\{p'\}$ and  $U_i$ for each $i\in[k]$.
Thus, we have $t=k+2$ parties. 

The set of voters corresponds to the set of ``non-edges'' in~$G$, that is, to $\olsi{E}=\{ u u': u \in U_i, u' \in U_j, i<j, uu' \notin E\}$. Namely, for each $e=uu' \in \olsi{E}$, 
we create two voters~$v_e$ and~$v'_e$ with  preferences as in~(\ref{eq:prefs-of-ve}).
We fix an arbitrary ordering over~$C$, and write  $\ora{X}$ for listing a set~$X$ of candidates according to this order, and $\ola{X}$  for its reverse.
\begin{equation}
\label{eq:prefs-of-ve}
    \begin{array}{@{}l@{\hspace{2pt}}l@{}}
 v_e: & u,u', \ora{U\setminus \{u,u'\}}, p, p'\\[4pt]
 v'_e: & p, p', \ola{U\setminus \{u,u'\}}, u, u'\\
     \end{array}
\end{equation}
    
Consider a reduced election~$\EE$ obtained by some nominations of all parties. Notice that $\Cpl(p)=\alpha k+1$ and $\Cpl(p')=\alpha k$,  since %
$p$ defeats %
$p'$, and %
both are tied with every other candidate.

Assume that $G$ admits a multicolored clique $S=\{u^{(i)},\dots,u^{(k)}\}$ with $u^{(i)}\in U_i$ for $i\in[k]$. Let each party $U_i$ nominate  $u^{(i)}$. As $S$ is a clique in~$G$, it is an independent set in the complement of~$G$, so there is no $e \in \olsi{E}$ containing two vertices of~$S$ corresponding to two nominated candidates. Thus, each nominee from~$U$ obtains a Copeland$^\alpha$ score of~$\alpha(k+1)$. Since $\alpha<1$, this is strictly smaller than $\Cpl(p)$, so $p$ is the unique winner of the resulting election.

Conversely, assume for the sake of contradiction that $p$ is the unique winner of some reduced election~$\EE$, but the nominated candidates in~$U$ do not form a clique. 
Let $u^{(i)} \in U_i$ be a nominee such that 
there is an edge $e \in \olsi{E}$ in the complement of~$G$ between $u^{(i)}$ and some nominee $u^{(j)}\in U_j$ with $i<j$; we choose $i$ as the  minimal index where this happens.
Then, due to the two voters corresponding to $e \in \olsi{E}$ we know that candidate~$u^{(i)}$ defeats candidate~$u^{(j)}$, and due to our choice of~$i$, there is no nominated candidate that defeats~$u^{(i)}$. Hence,  $\Cpl(u^{(i)}) \geq \alpha k+1$, a 
contradiction to our assumption that $p$ is the unique winner in~$\EE$.
\end{proof}

We can strengthen Theorem~\ref{thm:Copeland-easy-alpha} as follows:

\begin{restatable}[\linkproof{sec:param-t-even}]{theorem}{thmCopelandfourteven}
\label{thm:4voters}
For any constant $\alpha \in [0,1]$, \PP{} for Copeland$^{\, \alpha}$ is 
$\mathsf{W}[1]$-hard when parameterized by~$t$, the number of parties, even if there are only four voters.    
\end{restatable}

We prove Theorem~\ref{thm:4voters} in two steps, first for $\alpha<1$, and then filling the gap with a more involved reduction for $\alpha=1$;
see 
\ifshort the full version~\new{\cite{SC24-fullversion}}. 
\else Appendix~\ref{sec:param-t-even}.
\fi
Regarding elections with an odd number of voters, we were able to prove the following:
\begin{restatable}[\linkproof{sec:proof-of-5voterst}]{theorem}{thmfivevoterst}
    \label{thm:5voters-t}
    For any constant~$\alpha \in [0,1]$, \PP\ for Copeland$^\alpha$ is $\mathsf{W}[1]$-hard when parameterized by~$t$, the number of parties, even if there are only five voters.
\end{restatable}

Each of these results uses a reduction from \textsc{Multicolored Clique}, but the constructions become gradually more complicated; the proof of Theorem~\ref{thm:5voters-t} necessitates also the notion of flat elections (see
\ifshort
the full version~\new{\cite{SC24-fullversion}}). 
\else 
Appendix~\ref{sec:proof-of-5voterst}).
\fi

\smallskip
We remark that \PP\  is in~$\XP$ when parameterized by~$t$, assuming that winner determination can be performed in polynomial time: there are at most $\maxsize^t$ possibilities  for how parties can choose their nominated candidates, so we can check whether the distinguished party wins in at least one  election resulting from some nomination strategy in $\maxsize^t n^{O(1)}$ time (see e.g.,~\cite{schlotter2024}).

\section{Maximin voting rule}
\label{sec:Maximin}

Turning to the Maximin voting rule, we investigate how the complexity of \PP\ for Maximin depends on the number of voters (Section~\ref{sec:MM-voters}) and on the number of parties (Section~\ref{sec:MM-parties}). 

\subsection{Few Voters}
\label{sec:MM-voters}

We start by extending the tractability result of Theorem~\ref{thm:Llull-2voters},  dealing with the Llull voting rule with two voters, to the Maximin voting rule with two or three voters.

For two voters, tractability again relies on Observation~\ref{obs:2voters_transitivity} stating the transitivity of the ``defeat'' relation. For three voters we say that candidate~$a$ \emph{strongly defeats} a candidate~$b$, if all three voters prefer~$a$ to~$b$. It is easy to see that the ``strong defeat'' relation is also transitive.  %
This implies that some candidate is the unique winner in a Maximin election if and only if it defeats every other candidate.

\begin{theorem}\label{thm:Maximin_V2}
   \PP{} for the Maximin voting rule is poly\-no\-mial-time solvable if there are only two voters. 
   \end{theorem}
\begin{proof}
The theorem hinges on the fact that a nominee $p$ is a unique winner in a Maximin election~$\E$ if and only if $p$ defeats every other nominee. 
To see this, first realize that if $p$ defeats all nominees then ${\MM_\E(p)=2}$ and we have ${\MM_\E(c)=0}$ for every other nominee~$c$, so $p$ is the unique winner.

Now assume that $p$ is the unique winner of a reduced election~$\EE$. Clearly,
 $p$ cannot be defeated by any nominee, as that would yield ${\MM_\E(p)=0}$. Neither is $\MM_\E(p)=1$ possible, as in this case every other nominee~$c$ must have Maximin-score $\MM_\E(c)=0$, i.e., has to be defeated by at least one other nominee.  However, by a similar argument as in the proof of Theorem~\ref{thm:Llull-2voters}, this quickly leads to a contradiction, because the ``defeat'' relation cannot contain cycles.
 
 Therefore, only $\MM_\E(p)=2$ is possible, and thus $p$ defeats all nominees.
Hence the same quadratic-time algorithm as in Theorem~\ref{thm:Llull-2voters} solves the \PP{} problem also for the Maximin voting rule in the case of two voters.
\end{proof}

\begin{theorem}\label{thm:Maximin_V3}
   \PP{} for the Maximin voting rule is poly\-no\-mial-time solvable if there are only three voters. 
   \end{theorem}
\begin{proof}
Again, we show that $p$ is a unique winner in some election~$\EE$ if and only if $p$ defeats every other nominee in~$\EE$.
To see this, first realize that if $p$ defeats all nominees, then $\MM_\E(p)\ge 2$ and $\MM_\E(c)\le 1$ for every other nominee~$c$, so $p$ is the unique winner.

Now assume that $p$ is the unique winner in some election~$\EE$.  Clearly, $\MM_\E(p)=0$ is impossible. If $\MM_\E(p)=1$, then every other nominee~$c$ must have $\MM_\E(c)=0$. This means that every nominee other than~$p$ is strongly defeated by another nominee. However, this is not possible, as there can be no cycles in the ``strong defeat'' relation due to its transitivity. 
Therefore, $\MM_\E(p) \geq 2$ must hold, 
and thus $p$ can become the unique winner in an election resulting from some nominations if and only if each party has a candidate that is defeated by~$p$. From this, the polynomial-time solvability of the problem follows easily.
\end{proof}

The following  theorem shows that for $n \geq 4$ voters, \PP{} for Maximin is hard even when all parties have size at most~$2$. The proof of Theorem~\ref{thm:Maximin_hard} deals with the case of even and odd number of voters separately, providing two reductions from \textsc{3-SAT}.
Theorems~\ref{thm:Maximin_V2}, \ref{thm:Maximin_V3}, and \ref{thm:Maximin_hard} together prove Theorem~\ref{thm:Maximin-dichotomy}.

\begin{restatable}[\linkproof{sec:proof-of-Maximinhard}]{theorem}{thmMaximinhard}
    \label{thm:Maximin_hard}
    \PP\ for Maximin is $\NP$-complete even for instances where the number of voters is a fixed constant $n \geq 4$, and the maximum party size is~$\maxsize=2$.
\end{restatable}

\subsection{Few Parties}
\label{sec:MM-parties}

Contrasting Theorem~\ref{thm:Copeland-easy-alpha},  %
we show that if the number of parties is small, then we can efficiently solve \PP\ for Maximin. More precisely, we provide an FPT algorithm for this problem with parameter~$t$, the number of parties.
This subsection is dedicated to proving the following result.

\begin{restatable}[\linkproof{app:proof-of-fptcorrect}]{theorem}{thmMMfpt}
\label{thm:MM-fpt}    
There exists an algorithm that solves \PP\ for Maximin and runs in FPT time with parameter~$t$.
\end{restatable}

Let our input instance~$I$ of \PP\ be an election $\EE_0=(V,C,\{\succ_v\}_{v \in V})$ whose candidate set~$C$ is partitioned into a family~$\P$ of parties containing a distinguished party $P^\star \in \P$. 
Our algorithm \AlgMM\ first makes certain guesses about the properties of a hypothetical solution to~$I$, i.e., a nomination strategy that allows $P^\star$ to become the unique winner in the resulting reduced election~$\EE$. 
Then, after some preprocessing steps, 
we reduce our problem to the following directed variant of the \textsc{Partitioned Subgraph Isomorphism} problem~\cite{alon-yuster-zwick-colorcoding,marx-2010-tw}.

\begin{center}
\fbox{ 
\parbox{0.95\columnwidth}{
Problem {\bf \LabSubDigraph:} \\
{\bf Input:} Digraphs $D$ and~$H$ with labelling $\gamma:V(H) \rightarrow V(D)$. \\
{\bf Question:}  Is there a subdigraph~$\wt{H}$ of~$H$ that is isomorphic to~$D$, and an isomorphism~$f:V(D) \rightarrow V(\wt{H})$ that maps each vertex~$v$ of~$D$ to a vertex of~$\wt{H}$ with label~$v$, i.e., satisfies $\gamma(f(v))=v$?
}}
\end{center}

Given an instance of \LabSubDigraph, we may refer to~$D$ and~$H$ as the \emph{pattern} and the \emph{host} graphs, respectively.
We  say that $\wt{H}$ is \emph{$\gamma$-isomorphic} to~$D$ if it satisfies the requirements given in the problem definition.

It is easy to see that \LabSubDigraph\ is $\NP$-complete, e.g., by a simple reduction from \textsc{Multicolored Clique}; see the results by Marx~\cite{marx-2010-tw} for much stronger lower bounds for the undirected version. 
However, we will only need to solve \LabSubDigraph\ in the easy special case when all vertices of the pattern graph have indegree at most~$1$.

\begin{lemma}
\label{lem:labelled-subdigraph}
\LabSubDigraph\ can be solved in $O(|V(H)|^2)$ time if 
the pattern graph~$D$ has maximum indegree~1.
\end{lemma}
\begin{proof}%
For a vertex~$v$ in a digraph~$G$, let $N_G^-(v)$ and~$N_G^+(v)$ denote $v$'s in- and outneighbors in~$G$, respectively.
    We will also use the notation $\Gamma_v=\{x \in V(H):\gamma(x)=v\}$ for the set of vertices in~$H$ with label~$v$ for some $v \in V(D)$.

    We introduce two simple rules that reduce the size of the input instance without changing its solvability.
    The first rule deals with vertices in the pattern graph that have indegree~$1$ and outdegree~$0$.
    
\begin{mylatrule}
\label{latrule:leaf}
    Let $(D,H,\gamma)$ be an instance of \LabSubDigraph\ containing a vertex~$v \in V(D)$ with $N_D^+(v)=\emptyset$ and $N_D^-(v)=\{u\}$.
    Then 
     delete $v$ from~$D$, and
     delete all vertices of~$\Gamma_u$ without
     out-neighbors in~$\Gamma_v$, as well as $\Gamma_v$ itself from~$H$.
\end{mylatrule}

The second rule deals with vertices in the pattern graph that have both in- and outdegree~1. 

\begin{mylatrule}
\label{latrule:cycle}
    Let $(D,H,\gamma)$ be an instance of \LabSubDigraph\ containing a vertex~$v \in V(D)$ with $N_D^+(v)=\{w\}$ and $N_D^-(v)=\{u\}$ such that $(u,w)$ is not an arc in~$D$.\footnote{In fact, it is not hard to modify Rule~\ref{latrule:cycle}  so that it becomes applicable even if $(u,w)$ is an arc of~$D$; however, the current weaker version of the rule is sufficient for us.} 
    First delete $v$ from~$D$ and add the arc~$(u,w)$ to~$D$. 
    Second, 
        delete $\Gamma_v$ from~$H$, and replace the arcs of~$H$ contained in~$\Gamma_u \times \Gamma_w$ with the arc set \[A_{uw}=\{(x,y):x \in \Gamma_u, y \in \Gamma_w,   N_H^+(x) \cap N_H^-(y)  \cap \Gamma_v \neq \emptyset\}.\]%
\end{mylatrule}%
\begin{restatable}[\linkproof{app:proof-of-clmlatrules}]{claim}{clmlatrules}
\label{clm:latrules}
    Applying Rule~\ref{latrule:leaf} or~\ref{latrule:cycle} yields an equivalent instance of \LabSubDigraph.
\end{restatable}  

Applying Rules~\ref{latrule:leaf} and~\ref{latrule:cycle} preserves the property that all vertices in the pattern graph have indegree at most~1. 
After applying Rule~\ref{latrule:leaf} exhaustively, we obtain an instance where all vertices of the pattern graph have in- and outdegree at most one, i.e., the pattern graph is  
a disjoint union of directed cycles, paths, and isolated vertices. In fact, since Rule~\ref{latrule:leaf} is applicable whenever the pattern graph has a connected component that is a directed path with at least two vertices, we know that after the exhaustive application of Rule~\ref{latrule:leaf} we arrive at a pattern graph that is %
a disjoint union of directed cycles and isolated vertices. 
Applying then Rule~\ref{latrule:cycle} exhaustively we arrive at an instance~$I^\star$ whose pattern graph $D^\star$ consists solely of isolated vertices, possibly with loops.
Solving such an instance~$I^\star$ is easy: $I^\star$ is a ``yes''-instance if and only if the host graph contains a vertex~$f(v)$ with label~$v$ for each $v \in V(D^\star)$, with $f(v)$ having a loop whenever $v$ has an incident loop in~$D^\star$. 

Notice that applying either of the two rules consists of the deletion of vertices and, possibly, the addition of arcs to the host graph. Starting from an instance $(D,H,\gamma)$, the total time spent on the former is $O(|V(H)|+|V(D)|)=O(|V(H)|)$, whereas the total time spent on the latter is at most~$O(|V(H)|^2)$, because no arc is added more than once to~$H$. Hence, the total running time is $O(V|H|)^2$.
\end{proof}

We are now ready to describe the steps of \AlgMM\ when run on the instance~$(\EE_0,\P,P^\star)$; see 
\ifshort the full version \K{\cite{SC24-fullversion}}
\else Appendix~\ref{app:proof-of-fptcorrect}
\fi for its correctness.

\begin{description}
    \item[Step~1.] %
    Guess the candidate~$p$ nominated by~$P^\star$ in  the reduced election~$\EE$, as well as its Maximin-score $s^\star=\MM_\EE(p)$ in~$\EE$.
    \item[Step~2.] 
    For each party $P \in \P \setminus \{P^\star\}$, guess a party $P' \in \P \setminus \{P\}$ for which the nominees $c$ and $c'$ of $P$ and $P'$ in~$\EE$, respectively, satisfy $N_{\E}(c,c') < s^\star$. Let $\delta(P)$ denote the guessed party.
    \item[Step 3.] 
    Delete every candidate~$c \in C$ for which $N_{\E_0}(p,c)<s^\star$. 
    \item[Step 4.] 
    For each party~$P \in \P$ such that 
    $\delta(P)=P^\star$, 
    delete all candidates~$c \in P$ for which $N_{\E_0}(c,p)\geq s^\star$.    
    \item[Step 5.]    
    Let $X$ be the set of candidates deleted in Steps~3 and~4.
    If there is a party $P \in \P \setminus \{P^\star\}$ with $P \subseteq X$, then return ``no.''     
    \item[Step 6.]
    Construct a digraph~$D$ whose vertex set is $\P \setminus \{P^\star\}$ and contains an arc~$(P',P)$ if and only if $P'=\delta(P)$; hence, each vertex in~$D$ has at most one incoming arc.\\
    Construct also a digraph~$H$ over $C \setminus X \setminus P^\star$ in which $(c',c)$ is an arc if and only if 
    $N_{\E_0}(c,c')<s^\star$.
    We set the label~$\gamma(c)$ of each candidate~$c$ to be the party containing~$c$.
    \item[Step 7.] Solve \LabSubDigraph\ on instance $J=(D,H,\gamma)$ using the algorithm of Lemma~\ref{lem:labelled-subdigraph}, and return ``yes'' if and only if $H$ admits a  subdigraph $\gamma$-isomorphic to~$D$. Otherwise return ``no.''
\end{description}

\balance

\section{Conclusions and future research}

We provided a detailed multivariate complexity analysis of the \PP{} problem in the framework of candidate nomination by parties for several Condorcet-consistent rules; see Table~\ref{tab:summary} for a summary. 
Our results show a clear difference between Copeland$^\alpha$ for $\alpha \in [0,1]$ and Maximin: although both remain $\NP$-hard even for a constant number of voters, \PP\ for Maximin becomes tractable (in the parameterized sense) in the realistic scenario where the number~$t$ of parties is small, while Copeland$^\alpha$ remains intractable even then.  %
An intriguing question we left open is whether \PP\ for Copeland with two or three voters becomes FPT when parameterized by~$t$.

For another promising research direction, recall that our algorithms for two voters relied on the transitivity of the ``defeat'' relation. Interestingly, the defeat relation is transitive for any number of voters if preferences are single-peaked. 
 Faliszewski et al.~\cite{faliszewski2016} proved that \PP{} for Plurality remains $\mathsf{NP}$-complete for such preferences. Misra~\cite{misra2019parameterized} strengthened this result by showing $\NP$-hardness for  1D-Euclidean profiles that are both single-peaked and single-crossing, even with maximum party size~$2$. What is the situation for voting rules other than Plurality?

 The related {\sc Necessary President} problem, asking if some candidate of a given party can become the winner regardless of nominations from other parties, was shown to be $\mathsf{coNP}$-complete for Plurality by Faliszewski et al.~\cite{faliszewski2016}, even with maximum party size two. Cechl{\'a}rov{\'a} et al.~\cite{cechlarova2023candidates} added the analogous results for $\ell$-Approval, $\ell$-Veto, and Plurality with run-off, and gave integer programs for  {\sc Necessary President} for further voting rules including Copeland, Llull, and Maximin. 
 As far as we know, the parameterized complexity of this problem 
 has not been considered yet.

\begin{acks}
Ildik\'o Schlotter is supported by the Hungarian Academy of Sciences under its Momentum Programme (LP2021-2) and its J\'anos Bolyai Research Scholarship.
Katar\'ina Cechl\'arov\'a is supported by VEGA 1/0585/24 and
APVV-21-0369.
\end{acks}

\bibliographystyle{ACM-Reference-Format} 
\bibliography{citeme}

\ifshort
\else
\clearpage
\begin{appendices}
\section{Additional related work}
\label{app:relatedwork}

The topic of this paper---the problem of nominating the candidates for an election by parties---is relatively new in the computer science literature, as it was initiated by Faliszewski et al.~\cite{faliszewski2016} only in 2017. However, the study of the computational complexity of problems connected with elections where the set of candidates is not fixed appeared already in the seminal paper by Bartholdi et al.~\cite{bartholdi1992} in 1992, in the form of control by adding or deleting candidates.  
Of the extensive literature on electoral control, we shall mention in Section~\ref{subsec:control} only results that are closest to the  current paper, i.e., those that consider Condorcet-consistent voting rules, in particular, Copeland$^\alpha$ and Maximin. 
For a broader overview of electoral control see the survey by Faliszewski and Rothe~\citeapp{Chapter7}.
 Additionally, in Section \ref{sec:parties} we present a review of alternative approaches to the study of elections involving parties that may contain either candidates or voters.

\subsection{Control of Elections}\label{subsec:control}

In the model introduced by Bartholdi et al.~\cite{bartholdi1992} for electoral control, there is a central authority, the \emph{controller}, that aims to achieve  a given goal by applying certain control actions. 
The possible control actions related to changing the set of candidates  include
adding, deleting, replacing, or partitioning candidates. The aim of the controller may be constructive, i.e., to make a distinguished candidate the winner of the election, or destructive, when it aims to prevent a given candidate from winning the election. 
In the combinatorial variant of control, adding or deleting a
candidate automatically requires adding or deleting a whole group of candidates. 

Control by partitioning the candidates assumes a two-stage %
election process where the candidate set~$C$ is partitioned into two disjoint subsets~$C_1$ and~$C_2$. In a sequential election, the entire electorate first votes on the candidates from $C_1$, while the second stage of the election is conducted between the winner of the first election and the candidates of~$C_2$. 
Alternatively, in a run-off election the first stage consists of two separate elections on the candidate sets~$C_1$ and~$C_2$, and the second stage of the election is conducted over the set of  winners from~$C_1$ and~$C_2$. 
In such two-stage elections, the controller may aim to achieve its goal by suitably choosing the partitioning of the candidate set.

Bartholdi et al.~\cite{bartholdi1992} dealt with constructive control and observed that Condorcet voting is immune to control by adding candidates but is computationally vulnerable to control by deleting candidates and by partitioning of candidates for both sequential and run-off elections. The last results have been strengthened by Erd\'elyi et al.~\citeapp{erdelyi2015more} who showed that polynomial-time solvability of the run-off version is preserved even in the case when the cardinality of the two candidate partition sets should differ by at most one. 

 Betzler and Uhlmann~\citeapp{betzlerUhlmann2009} studied constructive control for the Llull and Copeland voting rules. They showed that  control by deleting candidates for Llull is FPT with respect to the parameter ``number of candidates defeating the distinguished candidate'' and  control by deleting candidates for Copeland is $\mathsf{NP}$-hard  even when for every candidate the number of candidates that are not tied
with it is at most three. 
If the voting is tie-free, then for any $\alpha\in[0,1]$ control by deleting and adding candidates for Copeland$^\alpha$ election is $\mathsf{W}$[2]-complete with respect to the parameter ``number of deleted candidates'' 
and $\vphantom{L^{L^{\hat{a}}}}$ ``number of added candidates'', respectively.
For Copeland elections, control by deleting candidates  is $\mathsf{NP}$-complete for six voters and control by adding candidates is $\mathsf{NP}$-complete   if there are eight voters. For Llull, control by deleting candidates is $\mathsf{NP}$-complete for ten voters and control by adding candidates is $\mathsf{NP}$-complete  for eight voters.

A follow-up study is provided by  Faliszewski et al.~\cite{faliszewski2009llull}. The authors showed that  Copeland$^\alpha$  for any  $\alpha\in[0,1]$ is vulnerable to destructive control via adding, deleting, or partitioning candidates. Constructive control of all these types  is  computationally resistant for $\alpha\in [0,1]$ except when the number of added candidates is unlimited  and  $\alpha\in \{0,1\}$.
These results have been fine-tuned by Chen et al.~\citeapp{chen2017elections}. The authors show that for Copeland$^\alpha$ the constructive control by adding as well as deleting candidates is $\mathsf{NP}$-hard, the former if there are 20 voters and the latter for 26 voters.

Gurski and Ross~\citeapp{gurski2014} considered approximating the number of candidates necessary to add or to delete in order to achieve the constructive or the destructive aim of the control. They showed that there is no approximation algorithm with absolute performance guarantee for optimal constructive control by
deleting candidates in Copeland and by adding candidates in Llull voting schemes, unless $\mathsf{P} =\mathsf{NP}$. Moreover, there is no EPTAS  for optimal constructive control by adding and deleting candidates in
Copeland and Llull voting schemes, unless $\mathsf{W}[2] =\mathsf{FPT}$. Furthermore, Gurski and Ross proposed binary linear programs for solving the studied control problems and tested their performance on randomly generated elections.

For the Maximin voting rule, Liu and Zhu~\citeapp{liu2010parameterized}  proved that constructive control by adding candidates is $\mathsf{W}[2]$-hard with respect to the parameter ``number of added candidates''. 
Chen et al.~\citeapp{chen2017elections} showed that  constructive control by adding candidates is $\mathsf{NP}$-hard even if there are only 10 voters. 
Maushagen and Rothe~\citeapp{maushagen2020last} studied
constructive control by partitioning candidates in Maximin elections, and proved it to be $\mathsf{NP}$-complete in both the sequential and the run-off versions.

Combinatorial control was studied in Chen et al.~\citeapp{chen2017elections}. They
   proved that combinatorial constructive control by deleting candidates is $\mathsf{NP}$-hard for  Copeland$^\alpha$ and Maximin, even for elections with only a single voter, using the fact that these voting rules satisfy the unanimity property (which means that if a candidate~$c$ is ranked first by all the voters, then $c$ is the unique winner).
  For destructive control by deleting candidates they proved
  $\mathsf{NP}$-hardness for the  Maximin voting rule even in case of five voters, and for  Copeland even in case of three voters.  As far as constructive control by adding candidates is concerned, it is $\mathsf{NP}$-hard even for elections with three voters for Copeland$^\alpha$ with any $\alpha\in[0,1]$ and for 
Maximin for elections with only six voters. By contrast, combinatorial destructive control by adding candidates can be solved in polynomial time.%

Erd\'elyi et al.~\citeapp{erdelyi2021} gave a detailed overview of the results on computational complexity of election control by adding or deleting candidates. They complemented the known results by proving that the constructive control by replacing candidates is  $\mathsf{NP}$-complete for both Copeland$^\alpha$ and Maximin voting rules, while the destructive control is in $\mathsf{P}$ in both cases.

\subsection{Elections Involving Parties}\label{sec:parties}

Harrenstein et al.~\citeapp{harrenstein2021hotelling} present
a model for the strategic selection of party nominees.
Each party is fully described by a set of points on
the line (political spectrum) that correspond to the
positions of the candidates the party has to choose its nominee from. Each nominee then attracts the closest voters compared to all other representatives. The paper shows that a Nash equilibrium
is not guaranteed to exist even in a two-party game. Finding a Nash equilibrium is $\mathsf{NP}$-complete for the
general case, but for  two competing parties can be computed in linear time. 

In a different model, Harrenstein and Turrini~\citeapp{harrenstein2022computing} consider district-based elections.  In each district, voters rank the nominated candidates and elect the plurality winners. So  parties have to strategically place their candidates in districts so as to maximize the number of their nominees that get elected.  The authors show  that deciding the existence of pure
Nash equilibria for these games is $\mathsf{NP}$-complete if party size is bounded by a constant and  $\Sigma^P_2$-complete for the general case.

Another model of control  with elections involving parties was introduced by Perek et al.~\citeapp{perek2013}. Unlike in the current paper, the authors assume that voters, not candidates are partitioned into parties and voters in the same party vote in the same way. It is assumed that the most preferred candidate of a given leading party is the winner of the election. The authors study  the computational complexity of the problems to determine the maximum number of voters who can leave the leading party without changing the winner and the minimal number of voters that must leave the leading party to ensure that the winner will be changed. Perek at al.~\citeapp{perek2013} and a follow-up paper by Guo et al.~\citeapp{guo_yang2015} classify these problems into polynomial or  $\mathsf{NP}$-complete for several different voting rules.

\section{Additional Notation}
\label{app:graphs}

Here we provide some additional definitions for concepts appearing in the paper.

\medskip
\noindent
{\bf Undirected graphs.}
An undirected graph $G$ is a pair $G=(U,E)$ where $U$ is the set of \emph{vertices} and $E$ the set of \emph{edges} in~$G$, each edge connecting two vertices. If an edge~$e$ connects $u$ and~$v$, then we write $e=uv$, and we say that $u$ and $v$ are the \emph{endpoints} of~$e$. Moreover, $u$ and $v$ are \emph{adjacent to} each other, and each of them is \emph{incident to}~$e$.  
The \emph{neighborhood} of a vertex~$u\in U$ in~$G$ is denoted by $N_G(u)=\{v:uv \in E\}$, and we call $|N_G(u)|$ the degree of~$u$. 
If the set of vertices $U$ of $G$ can be partitioned into two disjoint sets $A$ and $B$ such that each edge connects a vertex in~$A$ with a vertex in~$B$, then $G$ is  \emph{bipartite}, and to stress its bipartition we shall write $G=(A\uplus B, E)$. 

Given a subset of vertices $K\subset U$, we denote by $G[K]$ the subgraph of $G$ \emph{induced} by~$K$, i.e., the graph whose vertex set is~$K$ and whose edges are those edges from $E$ that connect two vertices in~$K$. 

A \emph{clique} is a set of vertices that are pairwise adjacent, and an \emph{independent set} is a set of vertices so that they are pairwise non-adjacent. A set of edges is \emph{independent}, if no two of them shares an endpoint.

\smallskip
\noindent
{\bf Directed graphs.}
A directed graph~$D$ (or \emph{digraph}) is a pair~$(V,A)$ with vertex set~$V$ and arc set~$A$, where $A\subseteq V \times V$; we will also write $V(D)$ to denote the set of vertices in~$D$.
Each arc $a =(u,v) \in A$ \emph{points} from its \emph{head} vertex~$u \in V$ to its \emph{tail} vertex~$v \in V$.  
A \emph{loop} is an arc of the form~$(v,v)$ for some vertex~$v \in V$.
We write $N_D^-(v)=\{u:(u,v) \in A\}$ and~$N_D^+(v)=\{u:(v,u) \in A\}$ to denote the set of \emph{in-} and \emph{out-neighbors} of a vertex~$v \in V$ in~$D$; when the digraph~$D$ is clear from the context, we may omit the subscript.

A \emph{tournament} is a directed graph~$D=(V,A)$ without loops where for each distinct vertices $u,v \in V$, exactly one of~$(u,v)$ and~$(v,u)$ is an arc contained in~$A$.

Digraphs~$D=(V,A)$ and~$D'=(V',A')$ are \emph{isomorphic}, if there exists a mapping $f:V \rightarrow V'$ such that for each $u,v \in V$, it holds that $(u,v) \in A$ if and only if $(f(u),f(v) \in A')$;
the function~$f$ is then called an \emph{isomorphism}.

\smallskip
\noindent
{\bf Parameterized complexity.}
Parameterized complexity was introduced by Downey and Fellows~\cite{downey1999parameterized} as a tool to deal with hard problems. 
Each instance of a parameterized problem $Q$ is a pair $(I,k)$ consisting of an input $I$ and a parameter $k$, which is usually an integer. A parameterized problem is \emph{fixed-parameter tractable} (FPT) if there is an algorithm which correctly determines for each instance $(I,k)$ of the problem whether $(I,k)$ is a ``yes''-instance in running time $f(k)\cdot|I|^{O(1)}$, where $f$ is a computable function.

Analogously to $\NP$-hardness, parameterized complexity theory offers the notion of \emph{$\W1$-hardness};  
proving that a parameterized problem~$Q$ is $\mathsf{W}[1]$-hard provides strong evidence that we cannot expect an FPT algorithm for~$Q$. 
$\W1$-hardness can be established via  \emph{parameterized reductions}: given two parameterized problems $Q$ and~$Q'$, a parameterized reduction from~$Q$ to~$Q'$ is an algorithm that runs in FPT time and transforms an instance $(I,k)$ of~$Q$ into an equivalent instance $(I',k')$ of $Q'$ such that $k' \leq g(k)$ for some computable function $g$. 
If $Q$ is $\mathsf{W}[1]$-hard, then such a reduction from~$Q$ to~$Q'$ implies that $Q'$ is $\mathsf{W}[1]$-hard as well.

If a parameterized problem~$Q$  is $\NP$-hard for some fixed constant value of the parameter, then $Q$ is said to be para-$\NP$-hard with respect to this parameter.
By contrast, if $Q$ can be solved in polynomial time for all constant values of the parameter, then we say that $Q$ is in the class $\mathsf{XP}$. Clearly, $\mathsf{FPT \subseteq XP}$.

\medskip
\noindent
{\bf Additional notation for elections.}
Given a reduced election~$\EE$, we define the Copeland$^\alpha$-score of a party~$P$ as the  Copeland$^\alpha$-score of its nominated candidate in~$\EE$, and denote it as $\Cpl(P)$.
We define the Maximin-score of~$P$ analogously, and denote it as~$\MM(P)$.

Sometimes, when defining preferences of voters, we assume that there is an arbitrarily fixed total order~$\theta$ given over the set of all candidates~$C$. When we write~$\overrightarrow{A}$ for some $A\subseteq C$ we mean the sequence of candidates in~$A$ listed according to~$\theta$; notation~$\overleftarrow{A}$ denotes the sequence of these candidates in the opposite order.
\section{Missing proofs from Sections~\ref{sec:2voters} and~\ref{sec:3voters}}.

\subsection{Proof of Theorem~\ref{thm:Llull-2voters}}
\label{app:proof-of-Llulltwovoters}

\thmLlulltwovoters*

\begin{proof}
Using Observation~\ref{obs:2voters_transitivity}, it is not hard to prove the following claim:  some candidate~$p$ contained in our distinguished party can be a unique winner in an election~$\EE$ resulting from some nominations if and only if $p$ defeats every other nominee.

First, if $p$ defeats all other nominees in~$\EE$, then its Llull-score is $t-1$, the maximum score possible, while each candidate~$c \in C \setminus \{p\}$ is defeated by~$p$ and hence obtains a Llull-score at most~$t-2$. Thus, $p$ is indeed the unique winner of~$\EE$.

Second, suppose that there exists some candidate~$c_0$ not defeated by~$p$; we show that $p$ cannot be the unique winner in~$\EE$. Suppose for the sake of contradiction that $p$ has higher Llull-score than every other nominated candidate. 
Then each nominee in~$C \setminus \{p\}$ has to be defeated by some other nominee, as otherwise its Llull-score will be $t-1$, preventing $p$ from being the unique winner.
In particular, $c_0$ must be defeated by some nominee~$c_1$. Similarly, $c_1$ must also be defeated by some nominee~$c_2$. Again, $c_2$ in turn must be defeated by some nominee~$c_3$.
Since we can always repeat this argument and the number of nominees is finite, we must arrive  at a cycle in the ``defeat'' relation, contradicting Observation~\ref{obs:2voters_transitivity}. 

The claim we have just proved offers a simple quadratic-time algorithm to solve an instance of \PP{} for Llull voting with two voters: we simply need to check for each candidate $p$ in the distinguished party~$P$ whether each other party~$P'$ contains at least one candidate~$c_{P'}$ that is defeated by $p$ (i.e., is ranked worse than~$p$ by both voters). If such a candidate exists for all parties, then $p$ can be the unique winner by nominating these candidates, and thus we output ``Yes''; otherwise, we know that $p$ cannot be the unique winner. If no candidate in~$P$ can be the unique winner, then we output ``No''.
\end{proof}

\subsection{Proof of Theorem~\ref{thm:Copeland-v=s=2}}
\label{sec:proof-of-CopelandtwoNPc}

\thmCopelandtwovoters*
\begin{proof}
We reduce from \textsc{3-Coloring}.
Let $G=(U,E)$ be our input graph with vertex set $U=\{u_1,\dots,u_n\}$. We construct an instance of \PP\ with voter set~$\{v,v'\}$ and~$\maxsize=2$ as follows.

We let $P=\{p\}$ be our distinguished party, and we create candidates $p'_1,p'_2,p'_3,d,d'$ as well as candidates~$q_i$ for $i \in [n]$; all of these $n+5$ candidates form singleton parties. For each color~$c \in [3]$, we further create a party~$P_{u_i}^c=\{u_i^c,\overline{u}_i^c\}$ for each~$i \in [n]$, and a party $P_e^c=\{e_i^c,e_j^c\}$ for each edge $e=u_i u_j \in E$. Then the number of parties is $t=4n+3|E|+6$, and each party has size at most two.

To define the preferences of~$v$ and~$v'$, we first define certain series of candidates that we call \emph{blocks}. For each candidate $u_i^c$ representing the option of coloring vertex $u_i \in U$ with color~$c \in [3]$, we define the set 
\[    A(u_i^c)=\left\{e^c_i:e \in E \textrm{ is incident to }u_i \right\}.\]
Note that if $i \neq j$, then the sets~$A(u_i^c)$ and $A(u_j^c)$ are disjoint, and if $u_i u_j=e \in E$, then both of these sets contain exactly one candidate from the party~$P_e^c$. Moreover, we also have $\bigcup_{i \in [n]} A(u_i^c)= \bigcup_{e \in E} P_e^c$ for each $c\in [3]$. Recall the notation introduced in Appendix~\ref{app:graphs} that we write $\ora{X}$ for the ordering of a set~$X$ of candidates according to some arbitrary fixed ordering over all candidates, and we write $\ola{X}$  for its reverse.
We can now define the following blocks, $X^c$ and~$\wt{X}^c$, for each $c \in [3]$:
\begin{align}
\label{eq:Xblocks-def}
\begin{split}
X^c &= u_1^c, \ora{A(u_1^c)}, u_2^c, \ora{A(u_2^c)}, \dots, u_n^c, \ora{A(u_n^c)}; \\
\wt{X}^c &= u_n^c, \ola{A(u_n^c)}, u_{n-1}^c, \ola{A(u_{n-1}^c)}, \dots, u_1^c, \ola{A(u_1^c)}.
\end{split}
\end{align}
We further introduce the following blocks:
\begin{align*}
    Y &= 
    q_1,\ol{u}_1^1,\ol{u}_1^2,\ol{u}_1^3,\, 
    q_2,\ol{u}_2^1,\ol{u}_2^2,\ol{u}_2^3,\dots,
    q_n,\ol{u}_n^1,\ol{u}_n^2,\ol{u}_n^3; \\
    \wt{Y} &= 
    q_n,\ol{u}_n^1,\ol{u}_n^2,\ol{u}_n^3,
    q_{n-1},\ol{u}_{n-1}^1,\ol{u}_{n-1}^2,\ol{u}_{n-1}^3, \dots,
    q_1,\ol{u}_1^1,\ol{u}_1^2,\ol{u}_1^3.    
\end{align*}

\noindent Now we can define the preferences of voters~$v$ and~$v'$ as follows:
\begin{align*}
    v &: p,p'_1,p'_2,p'_3,Y,X^1,X^2,X^3,d,d'; \\
    v' &: \wt{X}^3,\wt{X}^2,\wt{X}^1,d,d',\wt{Y},p,p'_1,p'_2,p'_3.
\end{align*}

First notice that $\Cpl(p)=3+\alpha(t-4)$ in
any %
reduced election~$\EE$, because $p$ defeats~$p'_1,p'_2,$ and~$p'_3$, and is tied with every other nominee. 
Define the candidate set~$\ol{U}_i:=\{\ol{u}_i^1,\ol{u}_i^2,\ol{u}_i^3\}$ for each $i \in [n]$, and the set~$A^c$ of candidates that are present in~$X^c$ (or equivalently, in~$\wt{X}^c$), i.e., $A^c:=\bigcup_{e \in E} P_e^c \cup \{u_i^c:i \in [n]\}$ for each $c \in [3]$.

Observe the following facts: 
\begin{itemize}
    \item Nominees in~$\bigcup_{i \in [n]} \ol{U}_i \cup \{p'_1,p'_2,p'_3,d,d'\}$ each defeat at most two nominees, and have Copeland$^\alpha$-score less than $\Cpl(p)$. 
    \item For each $i \in [n]$, nominee~$q_i$ is tied with every other nominee except for those in $\ol{U}_i$, while it defeats all nominees in~$\ol{U}_i$. Therefore, $\Cpl(q_i)<\Cpl(p)$ if and only if at most two candidates from~$\ol{U}_i$ are nominated, which in turn means that at least one candidate among $\{u_i^1,u_i^2,u_i^3\}$ is nominated.     
    \item For each $c \in [3]$, all nominees in~$A^c$ 
    defeat both~$d$ and~$d'$, and are tied with every other nominee not in~$A^c$. Hence, it holds that all nominees in~$A^c$ have Copeland$^\alpha$-score less than~$\Cpl(p)$ if and only if no nominee in~$A^c$ defeats another nominee in~$A^c$.
\end{itemize} 

This shows that $p$ is the unique winner in the reduced election~$\EE$ if and only if 
no nominee in~$A^c$ defeats another nominee in~$A^c$ for each $c \in [3]$. 
We claim that there exists such an election~$\EE$ resulting from some nominations if and only if
$G$ admits a proper 3-coloring.

\medskip 
\noindent
{\bf Direction ``$\Longrightarrow$'':}
Suppose first that $p$ is the unique winner in~$\EE$. By the above observations, for each $i \in [n]$, at least one candidate among $\{u_i^1,u_i^2,u_i^3\}$ must be nominated; let us color vertex~$u_i$ with some (arbitrarily fixed) color~$c \in [3]$ for which $u_i^c$ is nominated in~$\EE$.
Using that no nominee in~$A^c$ can defeat another nominee in~$A^c$, we are going to prove that the constructed coloring is proper. Assume for contradiction that $u_i$ and $u_j$ are both colored with some color~$c \in [3]$,  but there is an edge $e=u_iu_j \in E$. Consider the party $P_e^c=\{e^c_i,e^c_j\}$. If $P_e^c$ nominates $e_i^c$, then both~$v$ and~$v'$ prefer $u_i^c$ to~$e_i^c$, and thus party~$P_{u_i}^c$ defeats~$P_e^c$, a contradiction to our assumption on~$\EE$. 
If $P_e^c$ nominates~$e_j^c$, then a symmetric argument shows that $P^c_{u_j}$ defeats~$P_e^c$; a contradiction again.  Hence, our coloring is indeed proper.

\medskip 
\noindent
{\bf Direction ``$\Longleftarrow$'':}
Suppose now that $\chi:U \rightarrow [3]$ is a proper 3-coloring of~$G$. Let us nominate~$u_i^{\chi(u_i)}$ and also the candidates in~$\{ \ol{u}_i^{c'}:c' \in [3],c' \neq \chi(u_i)\}$  for each $i \in [n]$. Notice that in this way, at most two candidates from~$\ol{U}_i$ are nominated for each~$i \in [n]$.
Furthermore, for each $e \in E$ and $c \in [3]$, let us nominate the candidate~$e_i^c$ where $u_i$ is an (arbitrarily fixed) endpoint of~$e$ that does \emph{not} have color~$c$; crucially, $e$ must admit such an endpoint, because $\chi$ is a proper coloring, so $\chi(u_i)=\chi(u_j)=c$ is not possible.
 It remains to show that for each $c \in [3]$, no candidate from~$A^c$ defeats another in~$A^c$.

We say that candidates of the form~$u_i^c$ or~$e_i^c$ are \emph{associated} with color~$c$. 
First, note that two candidates associated with different colors are ranked differently by~$v$ and~$v'$, due to the ordering of the different blocks in the preference lists of~$v$ and~$v'$. 
Hence, it suffices to consider nominated candidates that are associated with the same color, say~$c \in [3]$.

Note that each two candidates~$u_i^c$ and $u_j^c$ for some $i,j \in [n]$ are ranked differently by~$v$ and~$v'$, because $v$ and $v'$ order the candidates $u_1^c,\dots,u_n^c$ in exactly the opposite order.
Let us now show similarly that each two candidates~$e_i^c$ and $\hat{e}_j^c$ for some $i,j \in [n]$ are ranked differently by~$v$ and~$v'$. If $i < j$, then $v$ prefers all candidates in $A(u_i^c)$ to all candidates in~$A(u_j^c)$, while $v'$ prefers candidates in $A(u_j^c)$ to candidates in~$A(u_i^c)$; 
the case $i>j$ is analogous. If $i=j$, then $e_i^c$ and $\hat{e}_j^c$ are both contained in~$A(u_i^c)$, and our claim follows from the fact that $v$ and $v'$ order the candidate set~$A(u_i^c)$ in exactly the opposite order.

Hence, it remains to compare two candidates $u_i^c$ and~$e_j^c$ for some $i,j \in [n]$ and $e \in E$. If $i < j$, then $v$ prefers~$u_i^c$ to~$e_j^c \in A(u_j^c)$, but $v'$ prefers~$e_j^c$ to~$u_i^c$. The case $i>j$ is analogous. 
Therefore, the only possibility when $v$ and~$v'$ could order~$u_i^c$ and~$e_j^c$ in the same way is when $i=j$. However, then it is not possible that both $u_i^c$ and $e_j^c=e_i^c$ are nominates, since we only nominate $e_i^c$ if $\chi(u_i) \neq c$, in which case $u_i^c$ is \emph{not} nominated. This proves that there are no two nominees that are ranked in the same order by~$v$ and~$v'$. Therefore, no candidate in~$A^c$ defeats another candidate in~$A^c$, as required, finishing the proof.
\end{proof}

\subsection{Proof of Theorem~\ref{thm:3voters-NPh}}
\label{sec:proof-of-3votersNPh}

We start with proving Theorem~\ref{thm:MMC-deg3} in Appendix~\ref{app:MC-variant}, then prove the correctness of the reduction presented in Section~\ref{sec:3votersreduction} in Appendix~\ref{app:proof-of-correctness-3voters}.

\subsubsection{Proof of Theorem~\ref{thm:MMC-deg3}}
\label{app:MC-variant}
    
    \mmcvariant*
    
    \begin{proof}
    Let $G=(S \uplus C, R,E)$ be our input instance with~$S$, $C$, and~$R$ being the set of singles, couples, and rooms, respectively.
    It is easy to see that by adding the necessary number, namely $2|R|-(|S|+2|C|)$, of dummy singles and making them adjacent to all rooms in~$G$ we obtain an equivalent instance where a complete matching necessarily assigns two people (i.e., two singles or one couple) to each room. Hence, we may assume that $|R|=|S|/2+|C|$. We may further assume that $G$ contains no isolated vertices.
    
    Next, we are going to construct an equivalent instance where each vertex in~$G$ has degree at most~$3$.
    We will achieve this through a series of operations,
    each one modifying the current instance and constructing an equivalent instance.

    Rule~\ref{rule:separate} decreases the number of rooms that are adjacent to at least one single and to at least two couples.
        \begin{myrule}
        \label{rule:separate}
        If $r \in R$ is adjacent to at least two couples and at least one single, then do the following: 
        add a new room~$r'$ and a new couple~$c'$, 
        and modify the set of edges by making~$c'$ adjacent to both~$r$ and~$r'$, and then replacing each edge~$rs$ connecting $r$ to a single~$s \in S$ with an edge~$r's$. See Figure~\ref{fig:separate} for an illustration.
        \end{myrule}
    
    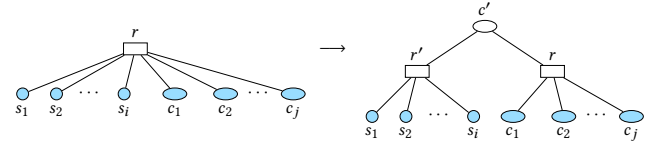
\begin{figure}[ht]
    \centering
    \resizebox{1\columnwidth}{!}{       
    \begin{tikzpicture}[scale=1]
    
    \node[room,label={[label distance=2pt]90:{$ r$}}] (r) at (0,0.4) {};
    \node[single,fill=bordercol,label={[label distance=2pt]-90:{$s_1$}}] (s1) at (-2.0,-0.4) {};
    \node[single,fill=bordercol,label={[label distance=2pt]-90:{$s_2$}}] (s2) at (-1.4,-0.4) {};
    \node at (-0.8,-0.4) {$\cdots$};
    \node[single,fill=bordercol,label={[label distance=2pt]-90:{$s_i$}}] (s3) at (-0.2,-0.4) {};
    \node[couple,fill=bordercol,label={[label distance=2pt]-90:{$c_1$}}] (c1) at (0.7,-0.4) {};  	
    \node[couple,fill=bordercol,label={[label distance=2pt]-90:{$c_2$}}] (c2) at (1.6,-0.4) {};  	
    \node at (2.2,-0.4) {$\cdots$};
    \node[couple,fill=bordercol,label={[label distance=2pt]-90:{$c_j$}}] (c3) at (2.8,-0.4) {};  	
    
    \draw (r)--(s1);
    \draw (r)--(s2);
    \draw (r)--(s3);
    \draw (r)--(c1);
    \draw (r)--(c2);
    \draw (r)--(c3);
    
    \node at (3.5,0.4) {$\longrightarrow$};
    
    \node[single,fill=bordercol,label={[label distance=2pt]-90:{$s_1$}}] (s'1) at (4.2,-0.8) {};
    \node[single,fill=bordercol,label={[label distance=2pt]-90:{$s_2$}}] (s'2) at (4.8,-0.8) {};
    \node at (5.4,-0.8) {$\cdots$};
    \node[single,fill=bordercol,label={[label distance=2pt]-90:{$s_i$}}] (s'3) at (6,-0.8) {};
    \node[room,label={[label distance=2pt]90:{$ r'$}}] (r') at (5.0,0) {};
    \node[couple,fill=bordercol,label={[label distance=2pt]-90:{$c_1$}}] (c'1) at (6.7,-0.8) {};  	
    \node[couple,fill=bordercol,label={[label distance=2pt]-90:{$c_2$}}] (c'2) at (7.6,-0.8) {};  	
    \node at (8.2,-0.8) {$\cdots$};
    \node[couple,fill=bordercol,label={[label distance=2pt]-90:{$c_j$}}] (c'3) at (8.8,-0.8) {};  	
    \node[room,label={[label distance=2pt]90:{$ r$}}] (rb) at (7.4,0) {};
    \node[couple,label={[label distance=2pt]90:{$ c'$}}] (cb) at (6.2,0.8) {};
    
    \draw (r')--(s'1);
    \draw (r')--(s'2);
    \draw (r')--(s'3);
    \draw (rb)--(c'1);
    \draw (rb)--(c'2);
    \draw (rb)--(c'3);
    \draw (rb)--(cb)--(r');
    
    \end{tikzpicture}
    }
    \caption{Illustration for Rule~\ref{rule:separate}.
     Singles, couples, and rooms are depicted as circles, ellipses, and rectangles, respectively. We highlighted in light blue those nodes that might be connected to parts of the graph not depicted in the figure. }
    \label{fig:separate}
    \end{figure}
    
        \begin{claim}
            \label{clm:rule-sep-safe}
            Applying Rule~\ref{rule:separate} results in an equivalent instance.
        \end{claim}
        \begin{claimproof}
            Let $G=((S \cup C) \uplus R,E)$ be the instance before applying Rule~\ref{rule:separate} for some room~$r \in R$, and let the resulting instance be $G'=(S \cup C \cup \{c'\}) \uplus R,E')$. 
            First, observe that given a complete matching~$M$ for~$G$, we can construct a complete matching $M'$ for~$G'$ as follows. 
            If $M$ contains an edge~$rc$ for some couple~$C$, then we set $M'=M \cup \{r'c'\}$. 
            Otherwise, there exist singles~$s$ and~$s'$ such that $\{sr,s'r\} \subseteq M$; we then set $M'=M \setminus \{sr,s'r\} \cup \{sr',s'r',rc'\}$. It is straightforward to check that $M'$ is a complete matching in both cases.
    
            For the other direction, suppose now that $G'$ admits a complete matching~$M'$. Then either~$c'r' \in M$ or $c'r \in M$. In the former case, the matching $M' \setminus \{c'r'\}$ is a complete matching in~$G$.  
            In the latter case, $r'$ must be adjacent to two singles, say $s$ and~$s'$ in~$M'$; hence, 
            $M=M' \setminus \{c'r,sr',s'r'\} \cup \{sr,s'r\}$ is a complete matching in~$G$.
        \end{claimproof}
        \smallskip

    Rule~\ref{rule:mixed-room} deals with rooms that are adjacent to exactly one single and one couple.
        \begin{myrule}
        \label{rule:mixed-room}
        If $r \in R$ is adjacent to exactly one single~$s$ and exactly one couple~$c$, then delete the edge~$rs$.
        \end{myrule}
    
    \begin{claim}
        \label{clm:mixed-room-safe}
        Rule~\ref{rule:mixed-room} constructs an equivalent instance.
    \end{claim}
    \begin{claimproof}
        Since a complete matching needs to match exactly two persons (a couple or two singles) to each room, it is clear that no such matching can include the edge~$rs$, and hence, we can safely delete it.
    \end{claimproof}
        Rule~\ref{rule:reduce-single-degree} reduces the degree of a single who is adjacent to more than three rooms. 
        \begin{myrule}
        \label{rule:reduce-single-degree}
            If $s \in S$ is adjacent to at least four rooms~$r_1,r_2,r_3,$ and~$r_4$, then do the following: 
            add new singles $s_2,s_3,s_4$, and~$s^\star$, 
            along with new rooms~$r'_1$ and~$r'_2$, and modify the set of edges by deleting the edges $sr_2,sr_3,$ and~$sr_4$, 
            and adding the edge set \[\{s_2r_2, s_3 r_3, s_4r_4, sr'_1,s_2 r'_1,s_3 r'_2, s_4 r'_2,s^\star r'_1,s^\star r'_2\}.\] See Figure~\ref{fig:reduce-single-degree} for an illustration.
        \end{myrule}
    
    \begin{figure}[ht]
    \centering
    \resizebox{1\columnwidth}{!}{       
    \begin{tikzpicture}[scale=1]
    
    \node[single,fill=bordercol,label={[label distance=2pt]90:{$ s$}}] (s) at (0,0.0) {};
    \node[room,fill=bordercol,label={[label distance=2pt]270:{$ r_1$}}] (r1) at (-1.2,-0.8) {};
    \node[room,fill=bordercol,label={[label distance=2pt]270:{$ r_2$}}] (r2) at (-0.4,-0.8) {};
    \node[room,fill=bordercol,label={[label distance=2pt]270:{$ r_3$}}] (r3) at (0.4,-0.8) {};
    \node[room,fill=bordercol,label={[label distance=2pt]270:{$ r_4$}}] (r4) at (1.2,-0.8) {};
    
    \node (h1) at (-0.7,0.0) {};
    \node (h2) at (-0.65,0.12) {};
    \node (h3) at (-0.65,-0.12) {};
    
    \draw (s)--(r1);
    \draw (s)--(r2);
    \draw (s)--(r3);
    \draw (s)--(r4);
    \draw (s)--(h1);
    \draw (s)--(h2);
    \draw (s)--(h3);
    
    \node at (2,0.4) {$\longrightarrow$};
    
    \node[room,fill=bordercol,label={[label distance=2pt]270:{$r_1$}}] (rb1) at (3.4,-0.8) {};
    \node[room,fill=bordercol,label={[label distance=2pt]270:{$r_2$}}] (rb2) at (4.2,-0.8) {};
    \node[room,fill=bordercol,label={[label distance=2pt]270:{$r_3$}}] (rb3) at (5.8,-0.8) {};
    \node[room,fill=bordercol,label={[label distance=2pt]270:{$r_4$}}] (rb4) at (6.6,-0.8) {};
    
    \node[single,fill=bordercol,label={[label distance=2pt]110:{$s$}}] (sb) at (3.4,0) {};
    \node[single,label={[label distance=2pt]70:{$s_2$}}] (s2) at (4.2,0) {};
    \node[room,label={[label distance=2pt]160:{$r'_1$}}] (r'1) at (3.8,0.8) {};
    
    \node[single,label={[label distance=2pt]110:{$s_3$}}] (s3) at (5.8,0) {};
    \node[single,label={[label distance=2pt]70:{$s_4$}}] (s4) at (6.6,0) {};
    \node[room,label={[label distance=1.5pt]20:{$r'_2$}}] (r'2) at (6.2,0.8) {};
    
    \node[single,label={[label distance=2pt]90:{$s^\star$}}] (sstar) at (5,1.2) {};
    
    \node (hb1) at (2.7,0) {};
    \node (hb2) at (2.75,0.12) {};
    \node (hb3) at (2.75,-0.12) {};
    
    \draw (r'1)--(sb)--(rb1);
    \draw[double,double distance=1.5pt] (sstar)--(r'1.base)--(s2);
    \draw (s2)--(rb2);
    \draw (sstar)--(r'2);
    \draw (s3)--(rb3);
    \draw (s4)--(rb4);
    \draw[double,double distance=1.5pt] (s4)--(r'2.base)--(s3);
    \draw (sb)--(hb1);
    \draw (sb)--(hb2);
    \draw (sb)--(hb3);
    
    \node[room,fill=white] (r'1) at (3.8,0.8) {};
    \node[room,fill=white] (r'2) at (6.2,0.8) {};
    
    \end{tikzpicture}
    }
    \caption{Illustration for Rule~\ref{rule:reduce-single-degree}, using the notation of Figure~\ref{fig:separate}.
    Double lines depict the edge set~$M_0$ defined in the proof of Claim~\ref{clm:rule-single-degree-safe}.}
    \label{fig:reduce-single-degree}
    \end{figure}
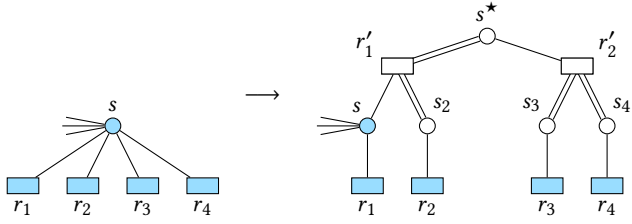

        \begin{claim}
        \label{clm:rule-single-degree-safe}
            Applying Rule~\ref{rule:reduce-single-degree} results in an equivalent instance.
        \end{claim}
        \begin{claimproof}
            Let $G=((S \cup C) \uplus R,E)$ be the instance before applying Rule~\ref{rule:reduce-single-degree} for some single~$s \in S$, and let   the resulting instance be $G'=((S \cup \{s_2,s_3,s_4,s^\star\} \cup C) \uplus (R \cup \{r'_1,r'_2\}),E')$. 
            
            First assume that $G$ admits a complete matching~$M$. 
            Define the matching~$M_0=\{s_2 r'_1,s^\star r'_1, s_3 r'_2,s_4 r'_2\}$.
            If $M$ does not contain any of the edges~$s r_i$, $i \in \{2,3,4\}$, then $M'=M \cup M_0$ is clearly a complete matching for~$G'$. 
            If $s r_2 \in M$, then define $P$ as the edge set of the path $(r_2, s_2, r'_1, s)$ in~$G$, and if $sr_i \in M$ for some~$i \in \{3,4\}$, then let $P$ be the edge set of the path~$(r_i, s_i, r'_2, s^\star, r'_1, s)$ in~$G'$. In either case, it is not hard to verify that $M'=(M \cup M_0) \triangle P$ is a complete matching in~$G'$.
            
            Assume now that $G'$ admits a complete matching~$M$. Then $s^\star$ is matched either to~$r'_1$ or to~$r'_2$, and these two rooms have to accommodate exactly four among the singles $s^\star,s_2,s_3,s_4,$ and~$s$. Hence, at most one single among $s_2,s_3,s_4$ and~$s$ is matched to a room in~$\{r_1,r_2,r_3,r_4\}$; if there is indeed one such room $r_i$, then let $F=\{sr_i\}$, otherwise let $F=\emptyset$. Then deleting all edges of~$M$  not present in~$G'$ and then adding~$F$ yields a complete matching in~$G$.
        \end{claimproof}
    
        Rule~\ref{rule:reduce-couple-degree} reduces the degree of a couple who is adjacent to more than three rooms. 
        \begin{myrule}
        \label{rule:reduce-couple-degree}
            If $c \in C$ is adjacent to rooms~$r_1,\dots,r_i$ for some $i \geq 4$, then do the following: 
            replace~$c$ with new couples $c_1,\dots,c_i$, 
            add new rooms~$r'_1,\dots,r'_{i-1}$, and modify the set of edges by deleting all edges incident to~$c$, and adding the edge set~$\{c_j r_j \colon j \in [i]\} \cup \{c_j r'_j,c_{j+1} r'_j:j \in [i-1]\}$. See Figure~\ref{fig:reduce-couple-degree} for an illustration.
        \end{myrule}
    
    \begin{figure}[ht]
    \centering
    \resizebox{1\columnwidth}{!}{       
    \begin{tikzpicture}[scale=1]
    
    \node[couple,label={[label distance=2pt]90:{$c$}}] (c) at (0,0.0) {};
    \node[room,fill=bordercol,label={[label distance=2pt]270:{$r_1$}}] (r1) at (-1.2,-0.8) {};
    \node[room,fill=bordercol,label={[label distance=2pt]270:{$r_2$}}] (r2) at (-0.4,-0.8) {};
    \node at (0.4,-0.8) {$\cdots$};
    \node[room,fill=bordercol,label={[label distance=2pt]270:{$r_i$}}] (ri) at (1.2,-0.8) {};
    \draw (c)--(r1);
    \draw (c)--(r2);
    \draw (c)--(ri);
    
    \node at (2.2,0) {$\longrightarrow$};
    
    \node[room,fill=bordercol,label={[label distance=2pt]270:{$r_1$}}] (rb1) at (3.4,-0.8) {};
    \node[room,fill=bordercol,label={[label distance=2pt]270:{$r_2$}}] (rb2) at (4.8,-0.8) {};
    \node at (5.6,-0.5) {$\cdots$};
    \node[room,fill=bordercol,label={[label distance=2pt]270:{$r_{i-1}$}}] (rbi1) at (6.4,-0.8) {};
    \node[room,fill=bordercol,label={[label distance=2pt]270:{$r_i$}}] (rbi) at (7.8,-0.8) {};
    
    \node[couple,label={[label distance=0pt]-5:{$c_1$}}] (c1) at (3.4,0) {};
    \node[couple,label={[label distance=0pt]-5:{$c_2$}}] (c2) at (4.8,0) {};
    \node[couple,label={[label distance=0pt]-5:{$c_{i-1}$}}] (ci1) at (6.4,0) {};
    \node[couple,label={[label distance=0pt]-5:{$c_i$}}] (ci) at (7.8,0) {};
    \draw (c1)--(rb1);
    \draw[double, double distance=1.5pt] (c2)--(rb2);
    \draw (ci)--(rbi);
    \draw (ci1)--(rbi1);
    \node[room,label={[label distance=2pt]90:{$r'_1$}}] (r'1) at (4.1,0.8) {};
    \node[room,label={[label distance=2pt]90:{$r'_2$}}] (r'2) at (5.5,0.8) {};
    \node[room,label={[label distance=2pt]90:{$r'_{i-1}$}}] (r'i) at (7.1,0.8) {};
    
    \draw[double, double distance=1.5pt] (r'1.base)--(c1);
    \draw (r'1)--(c2)--(r'2);
    \draw[double, double distance=1.5pt] (ci1) to node[pos=0.4,fill=white,yshift=2.8pt] {$\footnotesize \ddots$} (r'2.base);
    \draw (ci1)--(r'i);
    \draw[double, double distance=1.5pt] (r'i.base)--(ci);
    
    \node[room,fill=white] at (4.1,0.8) {};
    \node[room,fill=white] at (5.5,0.8) {};
    \node[room,fill=white] at (7.1,0.8) {};
    \end{tikzpicture}
    }
    \caption{Illustration for Rule~\ref{rule:reduce-couple-degree} using the notation of Figure~\ref{fig:separate}. Double lines depict the edge set~$M_j$ for $j=2$ defined in the proof of Claim~\ref{clm:rule-couple-degree-safe}.}
    \label{fig:reduce-couple-degree}
    \end{figure}
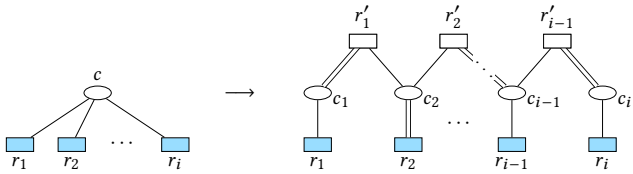
        
        \begin{claim}
        \label{clm:rule-couple-degree-safe}
            Applying Rule~\ref{rule:reduce-couple-degree} results in an equivalent instance.
        \end{claim}
        \begin{claimproof}
            Let $G=((S \cup C) \uplus R,E)$ be the instance before applying Rule~\ref{rule:reduce-couple-degree} for some couple~$c \in C$, and let the resulting instance be $G'=((S \cup C') \uplus (R \cup \{r'_1,\dots,r'_{i-1}\}),E')$  for $C'=C \setminus \{c\} \cup \{c_1,\dots,c_i\}$. 
    
            First assume that $G$ admits a complete matching~$M$. Then $cr_j \in M$ for some $j \in [i]$. Let 
            \begin{equation}
            \label{eqn:form-of-matching-R3}
            M_j= \{c_jr_j\} \cup \{c_h r'_h:1 \leq h <j\} \cup  \{c_h r'_{h-1}: j<h \leq i\};
            \end{equation}
            then $M'=M \setminus \{cr_j\}  \cup M_j$ is a complete matching in~$G'$. 
            
            Conversely, if $M'$ is a complete matching in~$G'$, then it is not hard to see that since the rooms~$r'_1,\dots,r'_{i-1}$ must accommodate all but one couple among $c_1,\dots,c_i$, it follows that $M'$ contains a matching~$M_j$ of the form~(\ref{eqn:form-of-matching-R3}).
            In this case, $M=M' \setminus M_j \cup \{cr_j\}$ is a complete matching in~$M$.
        \end{claimproof}
    
    Rule~\ref{rule:reduce-room-degree-couples} is obtained from Rule~\ref{rule:reduce-couple-degree} by switching the roles of couples and rooms.
        \begin{myrule}
        \label{rule:reduce-room-degree-couples}
            If $r \in R$ is adjacent to couples~$c_1,\dots,c_i$ for some $i \geq 4$ and no singles, then do the following: 
            replace~$r$ with new rooms $r_1,\dots,r_i$, 
            add new couples~$c'_1,\dots,c'_{i-1}$, and modify the set of edges by deleting all edges incident to~$r$, and adding the edge set~$\{c_j r_j \colon j \in [i]\} \cup \{r_j c'_j,r_{j+1} c'_j:j \in [i-1]\}$.
        \end{myrule}
    
        By symmetry,  Claim~\ref{clm:rule-couple-degree-safe} immediately implies the following.
        \begin{claim}
        \label{clm:rule-room-degree-couples-safe}
            Applying Rule~\ref{rule:reduce-room-degree-couples} results in an equivalent instance.
        \end{claim}
    
        Notice that after applying Rule~\ref{rule:separate} exhaustively, each room that is adjacent to at least one single can only be adjacent to at most one couple. 
        Rule~\ref{rule:reduce-room-degree-singles} deals with such rooms, decreasing their degree in~$G$ whenever it is more than three.
    \begin{myrule}
        \label{rule:reduce-room-degree-singles}
        If $r \in R$ has degree more than three in~$G$, and is adjacent to distinct singles $s_1$ and~$s_2$ and at most one couple, then do the following: add a set~$\wt{S}$ of new singles, a set~$\wt{C}$ of new couples, and a set~$\wt{R}$ of new rooms where
        \begin{align*}
            \wt{S} &= \{ s'_1,s'_2,\hat{s}_1,\hat{s_2}, s_{1 \vee 2},\hat{s}_{1 \vee 2}\}, \\
            \wt{C} &= \{ c',c_{1 \wedge 2},c_3,c_4,c_5\}, \\
            \wt{R} &= \{ r',r_1,r_2,r_{1 \wedge 2},r_{1 \vee 2},r_3,r_4,r_5\},        
        \end{align*}
        and then delete the edges~$r s_1, rs_2$, add the edge set~$\wt{E}$ where
        \begin{align*}
            \wt{E} = & \,\, \{ s_i r_i ,s'_i r_i ,\hat{s}_i r_i, s'_i r_{i+2}, \hat{s}_{1 \vee 2} r_{i+2},
            c_{i+2} r_{i+2}, c_{i+2} r_5:i \in [2]\}
             \\
            & \,\,  \cup \{  
            c_5 r_5, s_{1 \vee 2} r_{1 \vee 2} , 
            \hat{s}_{1 \vee 2} r_{1 \vee 2},
            s_{1 \vee 2} r, 
            c'r, c' r', 
            c_{1 \wedge 2} r', c_{1 \wedge 2} r_{1 \wedge 2} \}
             \\
            & \,\,  \cup \{  
                    s'_1 r_{1 \wedge 2} , s'_2 r_{1 \wedge 2}
                    \},            
        \end{align*}    
        and finally, if there is an edge between~$r$ and a couple~$c$ in~$G$, then replace~$c r$ with~$c r'$.  See Figure~\ref{fig:reduce-room-degree-single} for an illustration.
        \end{myrule}
    
    \begin{figure}[ht]
    \centering
    \resizebox{1\columnwidth}{!}{       
    \begin{tikzpicture}[scale=1]
    \node[room,fill=bordercol,label={[label distance=2pt]90:{$r$}}] (ra) at (-3.4,0.0) {};
    \draw (ra)--(-2.7,0.15);
    \draw (ra)--(-2.6,0.03);
    \draw (ra)--(-2.7,-0.15);
    \node[single,fill=bordercol,label={[label distance=2pt]0:{$s_1$}}] (sa1) at (-2.4,0.7) {};
    \node[single,fill=bordercol,label={[label distance=2pt]0:{$s_2$}}] (sa2) at (-2.4,1.5) {};
    \node[couple,fill=bordercol,label={[label distance=2pt]0:{$c$}}] (ca) at (-2.4,-1.0) {};
    \draw (sa2)--(ra)--(sa1);
    \draw (ra)--(ca);
    \node at (-1.2,0) {$\longrightarrow$};
    
    \node[room,fill=bordercol   ,label={[label distance=2pt]90:{$r$}}] (r) at (0,0.0) {};
    \node[single,label={[label distance=2pt]130:{$s_{1 \vee 2}$}}] (s1v2) at (0.6,1.2) {};
    \node[room,label={[label distance=2pt]95:{$r_{1 \vee 2}$}}] (r1v2) at (1.6,1.2) {};
    \node[couple,label={[label distance=2pt]90:{$c_5$}}] (c5) at (2.6,1.2) {};
    \node[room,label={[label distance=2pt]90:{$r_5$}}] (r5) at (3.6,1.2) {};
    \node[couple,label={[label distance=2pt]90:{$c_4$}}] (c4) at (4.6,0.6) {};
    \node[couple,label={[label distance=2pt]90:{$c_3$}}] (c3) at (4.6,1.8) {};
    \node[room,label={[label distance=2pt]90:{$r_4$}}] (r4) at (5.6,0.6) {};
    \node[room,label={[label distance=2pt]90:{$r_3$}}] (r3) at (5.6,1.8) {};
    \node[single,label={[label distance=2pt]0:{$\hat{s}_{1 \vee 2}$}}] (hats1v2) at (6.6,1.2) {};
    \node[couple,label={[label distance=2pt]-130:{$c'$}}] (c') at (0.6,-1.2) {};
    \node[room,label={[label distance=2pt]-100:{$r'$}}] (r') at (1.6,-1.2) {};
    \node[couple,fill=bordercol,label={[label distance=2pt]0:{$c$}}] (c) at (1.6,-2.2) {};
    
    \node[couple,label={[label distance=2pt]-90:{$c_{1 \wedge 2}$}}] (c1w2) at (2.6,-1.2) {};
    \node[room,label={[label distance=2pt]-90:{$r_{1 \wedge 2}$}}] (r1w2) at (3.6,-1.2) {};
    \node[single,label={[label distance=0pt]-90:{$s'_1$}}] (s'1) at (4.6,-0.6) {};
    \node[single,label={[label distance=0pt]-90:{$s'_2$}}] (s'2) at (4.6,-1.8) {};
    \node[room,label={[label distance=2pt]-90:{$r_1$}}] (r1) at (5.6,-0.6) {};
    \node[room,label={[label distance=2pt]-90:{$r_2$}}] (r2) at (5.6,-1.8) {};
    
    \node[single,fill=bordercol,label={[label distance=2pt]0:{$s_1$}}] (s1) at (6.6,-0.3) {};
    \node[single,label={[label distance=2pt]0:{$\hat{s}_1$}}] (hats1) at (6.6,-0.9) {};
    \node[single,fill=bordercol,label={[label distance=2pt]0:{$s_2$}}] (s2) at (6.6,-1.5) {};
    \node[single,label={[label distance=2pt]0:{$\hat{s}_2$}}] (hats2) at (6.6,-2.1) {};
    
    \draw (r)--(s1v2);
    \draw (r1v2)--(c5);
    \draw[double, double distance=1.5pt] (r1v2)--(s1v2);
    \draw (r5)--(c3);
    \draw (r3)--(hats1v2);
    \draw (r5)--(c4);
    \draw[double, double distance=1.5pt] (c3)--(r3);
    \draw[double, double distance=1.5pt] (c4)--(r4);
    \draw[double, double distance=1.5pt] (c5)--(r5);
    \draw (r4)--(hats1v2);
    \draw (r)--(c');
    \draw[double, double distance=1.5pt] (r')--(c');
    \draw (r')--(c1w2);
    \draw[double, double distance=1.5pt] (c1w2)--(r1w2);
    \draw (r1w2)--(s'1);
    \draw (r1w2)--(s'2);
    \draw (r')--(c);
    \draw[double, double distance=1.5pt] (s'1)--(r1)--(hats1);
    \draw[double, double distance=1.5pt] (s'2)--(r2)--(hats2);
    \draw (s1)--(r1);
    \draw (s2)--(r2);
    \draw (s'1) [out=45,in=-135]to (r3);
    \draw (s'2) [out=45,in=-135]to (r4);
    \draw[double, double distance=1.5pt] (r1v2) [out=400,in=90]to (hats1v2);
    \draw (r)--(0.7,0.15);
    \draw (r)--(0.8,0.03);
    \draw (r)--(0.7,-0.15);
    
    \end{tikzpicture}
    }
    \caption{Illustration for Rule~\ref{rule:reduce-room-degree-singles} using the notation of Figure~\ref{fig:separate}.}
    \label{fig:reduce-room-degree-single}
    \end{figure}

        \begin{claim}
            \label{clm:rule-room-degree-singles-safe}
            Applying Rule~\ref{rule:reduce-room-degree-singles} results in an equivalent instance.
        \end{claim}
        \begin{claimproof}
            Let $G=((S \cup C) \uplus R,E)$ be the instance before applying Rule~\ref{rule:reduce-room-degree-singles} for some room~$r \in R$, and let~$G'=((S' \cup C') \uplus R',E')$ denote the resulting instance. 
            
            First, we prove that given a complete matching~$M$ for~$G$, we can construct a complete matching $M'$ for~$G'$ as follows. 
            Let $M_0$ be defined as the edge set
            \begin{align*}
            M_0 & =\{  s'_i r_i ,\hat{s}_i r_i,
            c_{i+2} r_{i+2}:i \in [2]\} \\
            & \phantom{=(} \cup 
            \{ c_5 r_5, s_{1 \vee 2} r_{1 \vee 2} , 
            \hat{s}_{1 \vee 2} r_{1 \vee 2}, c' r',  c_{1 \wedge 2} r_{1 \wedge 2}\};
            \end{align*}
            see again Figure~\ref{fig:reduce-room-degree-single}.
            Observe that $M_0 \subseteq \wt{E}$ is a complete matching for the subgraph of~$G'$ induced by~$\wt{S} \cup \wt{C} \cup \wt{R}$. 
            Therefore, if the matching~$M$ does not contain any of the edges~$E \setminus E'$, which Rule~\ref{rule:reduce-room-degree-singles} removed from~$G$, then 
            $M \cup M_0$ is a complete matching for~$G'$.
    
            Otherwise, we distinguish between three cases, depending on how the room~$r$ is matched in~$M$.
            In each of these case, we will define an edge set~$M'$.
            First, if $cr \in M$ for some couple~$c$ (recall that there is at most one couple adjacent to~$r$ in~$G$), then we set
            \[M'=M \setminus \{cr\} \cup \{c r', c'r\} \cup (M_0 \setminus \{c'r'\}).
            \]
            Second, if $|\{s_1 r,s_2 r\} \cap M|=1$,  then assume that $s_1 r \in M$; the case $s_2 r \in M$ is symmetric. In this case, we let
            \begin{align*}
            M' &= (M \cup M_0 \setminus 
            \{ s_1 r, s'_1 r_1, c_3 r_3, c_5 r_5, 
            s_{1 \vee 2} r_{1 \vee 2}, 
            \hat{s}_{1 \vee 2} r_{1 \vee 2}
            \} \\
            & \phantom{=a}
            \cup 
            \{ s_1 r_1, s'_1 r_3, \hat{s}_{1 \vee 2} r_3, 
            c_3 r_5, c_5 r_{1 \vee 2}, s_{1 \vee 2} r.
            \}
            \end{align*}
            Third, if $\{s_1 r,s_2 r\} \subseteq M$, then 
            we can define 
            \begin{align*} 
            M'&=M \cup M_0 \setminus 
            \{ s_1 r, s_2 r, s'_1 r_1, s'_2 r_2, 
            c_{1 \wedge 2} r_{1 \wedge 2}, c'r' \}
            \\
            & \phantom{=a}
            \cup 
            \{ s_1 r_1, s_2 r_2, s'_1 r_{1 \wedge 2}, 
            s'_2 r_{1 \wedge 2}, c_{1 \wedge 2} r', c'r\}
            \end{align*}
            In each of the three cases, it is straightforward to verify that $M'$ is a complete matching for~$G'$.
    
            \smallskip
        For the other direction, suppose now that $M'$ is a complete matching for~$G'$. Again, we will define an edge set~$M$ depending on the properties of~$M'$. 
        
        First assume that $r c' \in M'$. 
        If it further holds that $r' c_{1 \wedge 2} \in M'$, then room~$r_{1 \wedge 2}$ can only be matched to the singles~$s'_1$ and~$s'_2$ by~$M'$, that is, 
        ${\{ s'_1 r_{1 \wedge 2}, s'_2 r_{1 \wedge 2}\} \subseteq M'}$, because $M'$ is a complete matching. This also implies $\{s_1 r_1,s_2 r_2 \} \subseteq M'$. Note also that no vertex in~$S \cup C \cup R$ other than the vertices~$r,s_1,$ and $s_2$ can be matched to a vertex in~$\wt{S} \cup \wt{C} \cup \wt{R}$ (recall that $r' c_{1 \wedge 2} \in M'$, so $r'$ cannot be matched to any couple in~$C$). 
        However, then 
        $M = M' \setminus \wt{E} \cup \{ s_1 r, s_2 r\}$ is a complete matching for~$G$.
        
        Otherwise, if $rc' \in M'$ but $r' c_{1 \wedge 2} \notin M'$, then $r'$ must be adjacent to some couple~$c \in C$ within~$M'$, and the edge~$cr'$ must have been added to~$E'$ by Rule~\ref{rule:reduce-room-degree-singles} because it found $rc \in E$. In this case, we must have 
        $c_{1 \wedge 2} r_{1 \wedge 2} \in M'$, as otherwise $c_{1 \wedge 2 }$ would be unmatched. 
        Since $s_{1 \vee 2}$ is not matched to~$r$, we also know that $s_{1 \vee 2} r_{1 \vee 2} \in M'$, which in turn implies $\hat{s}_{1 \vee 2} r_{1 \vee 2} \in M'$. This yields also $\{c_3 r_3, c_4 r_4, c_5 r_5\} \subseteq M'$. Now, as $s'_1$ and $s'_2$ both must be covered by~$M'$ (and they cannot be matched to~$r_{1 \wedge 2}$), we obtain $s'_1 r_1,s'_2 r_2 \in M$.
        Since $\hat{s}_i$ can only be matched to~$r_i$ for both $i \in [2]$, we obtain that $s_1 r_1, s_2 r_2 \notin M$.
        Therefore, 
        $M=M' \setminus \wt{E} \setminus \{r'c\} \cup \{rc\}$ is a complete matching for~$G$.
    
        It remains to consider the case when $rc' \notin M'$.
        On the one hand, if $r s_{1 \vee 2} \notin M'$, then via a reasoning similar to the previous case, it is not hard to see that $M_0 \subseteq M'$. Thus, setting $M=M \setminus M_0$ is a complete matching for~$G$.
        On the other hand, if $r s_{1 \vee 2} \in M'$, then 
        we get $r_{1 \vee 2} c_5$ as otherwise room~$r_{1 \vee 2}$ is not filled. Hence, $\hat{s}_{1 \vee 2}$ must be matched to $r_3$ or to~$r_4$. Assume that $\hat{s}_{1 \vee 2} r_3 \in M'$; the case $\hat{s}_{1 \vee 2} r_4 \in M'$ is symmetric.
        Then $r_3$ must be also matched to~$s'_1$, which implies $s_1 r_1 \in M$. By contrast, $r_4$ must be matched to~$c_4$. Since $s'_2$ cannot be matched to $r_{1 \wedge 2}$ (because we have $c'r' \in M'$ 
        and $c_{1 \wedge 2} r_{1 \wedge 2} \in M'$ due to $r s_{1 \vee 2} \in M'$), we obtain that $s'_2 r_2 \in M'$, and consequently, $s_2 r_2 \notin M'$.
        This yields that $M=M \setminus \wt{E} \cup \{s_1 r\}$ is a complete matching for~$G$.
         This finishes the proof.
        \end{claimproof}    
    
    After applying Rules~\ref{rule:separate}--\ref{rule:reduce-room-degree-singles} exhaustively, let $G^\star$ denote the resulting graph. Then
    singles in~$G^\star$ have degree at most~$3$, because Rule~\ref{rule:reduce-single-degree} is not applicable. 
    Similarly, couples in~$G^\star$ have degree at most~$3$, because Rule~\ref{rule:reduce-couple-degree} is not applicable. Rooms in~$G^\star$ that are adjacent to at least one single 
    are adjacent to at most one couple, because Rule~\ref{rule:separate} is not applicable, 
    and they are adjacent to exactly two singles, because neither Rule~\ref{rule:mixed-room} nor Rule~\ref{rule:reduce-room-degree-singles}
    is applicable; in particular, they have degree~$3$. 
    Lastly, rooms that are only adjacent to couples also have degree at most~$3$, because Rule~\ref{rule:reduce-room-degree-couples} is not applicable. Hence, $G^\star$ has maximum degree at most~$3$, and fulfills also the requirement of the theorem on rooms adjacent to both singles and couples. 
    Due to Claims~\ref{clm:rule-sep-safe}--\ref{clm:rule-room-degree-singles-safe}, $G^\star$ is equivalent to the original graph~$G$ when interpreted as an instance of \textsc{Maximum Matching with Couples}.
    Note also that $G^\star$ also fulfills $|S^\star|/2+|C^\star|=|R^\star|$ where $S^\star$, $C^\star$, and $R^\star$ denote the number of singles, couples, and rooms in~$G^\star$, respectively; to see this, it suffices to observe that each of the Rules~\ref{rule:separate}--\ref{rule:reduce-room-degree-singles} maintains this invariant.
    
    Let us now show that applying Rules~\ref{rule:separate}--\ref{rule:reduce-room-degree-singles} takes polynomial time.
    To see this, it suffices to observe the following facts. 
    First, no new vertex introduced during the application of a rule has degree more than~$3$, or is adjacent to a single and more than one couple. Hence, none of the rules can be applied to these vertices. 
    Second, when applying one of the rules to a given vertex~$v$, the degree of~$v$ decreases, except in the case when $v$ is a room adjacent only to singles, and we apply Rule~\ref{rule:reduce-room-degree-singles}. However, this latter can happen to each room at most once (as after the application of Rule~\ref{rule:reduce-room-degree-singles}, they are adjacent to a newly introduced couple). Since we can decrease the degree of some vertex of~$G$ at most $2|E|$ times, and we can apply Rule~\ref{rule:reduce-room-degree-singles} to a room adjacent only to singles at most~$|R|$ times, we obtain that the total number of applying Rules~\ref{rule:separate}--\ref{rule:reduce-room-degree-singles} is at most $2|E|+|R|$. Since applying each rule takes linear time, the construction of~$G^\star$ can be completed in polynomial time.
    
    It remains to deal with vertices of~$G^\star$ that have degree less than~$2$. Let us construct a trivial ``no''-instance $G_0$ containing five rooms $r_1,\dots,r_5$ and five couples~$c_1,\dots,c_5$ whose underlying graph is the disjoint union of the complete bipartite graph on couples~$c_1,c_2,c_3$ and rooms~$r_1,r_2$, and the complete bipartite graph on couples~$c_4,c_5$ and rooms~$r_3,r_4,r_5$.
    We apply the following operations.
    \begin{myrule}
    \label{rule:reject-halfempty-room}
        If $r \in R$ is a room adjacent to exactly one single but no couples, then output~$G_0$. 
    \end{myrule}
    \begin{myrule}
    \label{rule:increase-degree1-single}
        If $s \in S$ has degree~$1$, then 
        introduce two new rooms $r$ and~$r'$, two new singles~$s'$ and~$s''$, and a new couple~$c$, 
        then add the edges $sr', cr,s'r,s''r,s'r',s''r'$. 
    \end{myrule}
    \begin{myrule}
    \label{rule:increase-degree1}
        If $v$ is a room adjacent only to one couple (and no singles), or~$v$ is a couple of degree~$1$, then 
        introduce two new rooms, $r$ and~$r'$, and two new couples, $c$ and~$c'$, and add the edges $cr,c'r,cr',c'r'$ plus an additional edge connecting~$v$ with either~$c$ or~$r$, depending on whether $v$ is a room or a couple. 
    \end{myrule}
    \begin{claim}
        Applying any of Rules~\ref{rule:reject-halfempty-room}--\ref{rule:increase-degree1} results in an equivalent instance.
    \end{claim}
    \begin{claimproof}
        First, a room~$r$ adjacent to only one single and no couples can never be matched two people; hence, an instance containing such a room is a ``no''-instance and hence equivalent to~$G_0$, which shows that Rule~\ref{rule:reject-halfempty-room} is correct.
        Second, observe that all complete matchings in an instance resulting from applying Rule~\ref{rule:increase-degree1-single} must contain the edges $cr,s'r',s''r'$ and, hence, cannot match~$s$ along any of the newly introduced edges. From this, the statement for Rule~\ref{rule:increase-degree1-single} follows.
        Similarly, 
        the newly added rooms~$r,r'$ and couples~$c,c'$ in Rule~\ref{rule:increase-degree1} can only be matched to each other (in some way) and not to~$v$ in all complete matchings. From this, the statement for Rule~\ref{rule:increase-degree1} follows.
    \end{claimproof}
    
    Finally, observe that applying first Rule~\ref{rule:reject-halfempty-room} exhaustively to~$G^\star$, then Rule~\ref{rule:increase-degree1-single} exhaustively, and then Rule~\ref{rule:increase-degree1} exhaustively (which takes linear time in total), 
    we ensure that each vertex has degree~$2$ or~$3$ in the resulting graph, 
    while also maintaining the remaining properties required by the theorem. 
    \end{proof}

\subsubsection{Correctness of the reduction presented in Section~\ref{sec:3votersreduction}}
\label{app:proof-of-correctness-3voters}
Before proving the correctness of the reduction presented in Section~\ref{sec:3votersreduction}, we present a simple lemma about flat elections.

\begin{restatable}{lemma}{lemacyclicflat}
\label{lem:acyclic-flat}
    If $\EE$ is a flat election with an odd number of voters whose candidate set is partitioned into pairwise disjoint sets~$C_1$, $C_2$, and~$C_3$ such that 
    \begin{itemize}
        \item each candidate in~$C_1$ defeats all candidates in~$C_2$,
        \item each candidate in~$C_2$ defeats all candidates in~$C_3$, and
        \item each candidate in~$C_3$ defeats all candidates in~$C_1$,
    \end{itemize}
    then $|C_1|=|C_2|=|C_3|$.
\end{restatable}

\begin{proof}
    Note that for each $i \in [3]$, there exists a candidate~$c_i^\star$ in~$C_i$ that defeats at least $\frac{|C_i|-1}{2}$ candidates from~$C_i$, because for each pair of candidates in~$C_i$ one of them defeats the other (as there are an odd number of voters), and thus the total number of defeats within~$C_i$ is $\frac{|C_i|\cdot (|C_i|-1)}{2}$.
    Recall also that by the conditions of the lemma, $c_i^\star$ also defeats all candidates in~$C_{i+1}$ where we set $C_4:=C_1$. 
    
    Since $\EE$ is flat, every candidate defeats exactly $\frac{|C|-1}{2}$ candidates, which implies
    \begin{align}
    \frac{|C_1|-1}{2} + |C_2| &\leq \frac{|C|-1}{2}; \label{eq:C1}\\
    \frac{|C_2|-1}{2} + |C_3| &\leq \frac{|C|-1}{2}; \label{eq:C2}\\
    \frac{|C_3|-1}{2} + |C_1| &\leq \frac{|C|-1}{2}.\label{eq:C3}
    \end{align}
    Summing up inequalities (\ref{eq:C1})--\ref{eq:C3}), we obtain $\frac{3|C|-3}{2}$ on both sides, which implies that all of these inequalities must hold with equality. From this, $|C_1|=|C_2|=|C_3|$ follows by simple calculus.
\end{proof}

We are now ready to prove that the reduction presented in Section~\ref{sec:3votersreduction}, proving Theorem~\ref{thm:3voters-NPh}, is correct. 

\medskip
\noindent
{\bf Direction ``$\Longleftarrow:$''}%
    First, assume that  $\EE_X$ admits nominations resulting in a flat election. We start with the following claim.

\begin{claim}
    \label{clm:flatness}
    If a set of nominations for the relevant election~$\EE_X$ results in a flat election, then
    for each $t \in T$
    exactly three candidates from~$F_t$ are nominated, and the restriction of~$\EE_X$ to these three nominees is flat.
\end{claim}
\begin{claimproof}
    Assume that we are given a set of nominations that results in a flat election~$\EE$.

     Consider some candidate~$z$ of~$\EE_q$ written in the form~(\ref{eqn:candidate-form}) for some~$x \in \{\abase,\bbase,\cbase\}$ and indices~$h_1,\dots,h_{q-1} \in [3]$. %
    For each index ${i \in [q]}$, let $Z_{i,z}$ denote the candidates of~$\EE_q$ belonging to the same group as~$z$ at level~$i$;
    set also $Z_{q,z}=C_{q}$
    and $Z_{0,z}=\{z\}$.
    By the construction of~$\EE_q$ as given in Definition~\ref{def:EEq}, we know
    $|Z_{i,z}|=3^i$ for each $i \in [q]$ and 
    $Z_{0,z} \subset Z_{1,z} \subset \dots \subset Z_{q-1,z} \subset Z_{q,z}$. 
    Moreover, by Observation~\ref{obs:EEq-properties}, the candidate set~$Z_{i,z}$ for some $i \in [q]$ can be partitioned as $Z_{i-1,z} \uplus Z^+_{i-1,z} \uplus Z^-_{i-1,z}$ where 
    \begin{itemize}
    \item candidates in~$Z_{i-1,z}$ defeat all candidates in~$Z^+_{i-1,z}$, 
    \item candidates in~$Z^+_{i-1,z}$ defeat all candidates in~$Z^-_{i-1,z}$, and 
    \item  candidates in~$Z^-_{i-1,z}$ defeat all candidates in~$Z_{i-1,z}$
    \end{itemize}
    in the election~$\EE_q$.
    
    For some $z \in C_q$ and $i \in [q]$, 
    let us define the set~$\wt{F}_{i,z}$ of candidates in the election~$\EE_X$ as $\wt{F}_{i,z}=\bigcup \{F_{\psi(z)}:z \in Z_{i,z}\}$; for $i+1 \in [q]$ we define $\wt{F}^+_{i,z}$ and $\wt{F}^-_{i,z}$ analogously.\footnote{With a slight abuse of the notation, $F_t$ for some $t \in T$ is treated here, as well as at certain places later, as a \emph{set} (instead of a \emph{list}) of candidates from~$X$.}
    Notice that due to %
    Observation~\ref{obs:bijection-property},
    we have that for each $i \in [q]$ and $z \in C_q$
    \begin{itemize}
    \item[(a)] candidates in~$\wt{F}_{i-1,z}$ defeat all candidates in~$\wt{F}^+_{i-1,z}$, 
    \item[(b)] candidates in~$\wt{F}^+_{i-1,z}$ defeat all candidates in~$\wt{F}^-_{i-1,z}$, and 
    \item[(c)]  candidates in~$\wt{F}^-_{i-1,z}$ defeat all candidates in~$\wt{F}_{i-1,z}$
    \end{itemize}
    in the election~$\EE_X$. 
    We are going to prove by induction on~$i$ that 
    for each~$i=q,q-1, \dots, 1,0$ and
    for each $z \in C_q$, 
    \begin{itemize}
        \item[(i)] there are exactly $3 \cdot 3^i$ candidates
    in~$\wt{F}_{i,z}$ nominated in~$\EE$, and
        \item[(ii)] nominees in~$\wt{F}_{i,z}$ defeat exactly half of the nominees not contained in~$\wt{F}_{i,z}$.
    \end{itemize}
    Statement~(i) clearly holds for $i=q$, because $\wt{F}_{q,z}= \bigcup \{F_t \colon t \in T\}=X$ for every $z \in C_q$, and $X$ is the union of $3\rho=3|T|=3 \cdot 3^q$ parties. Statement~(ii) is vacuously true for~$i=q$.

    Assume now that the hypothesis holds for~$i$, and consider the statement for~$i-1$. 
    Fix some~$z \in C_q$.
    If the number of nominees is equal in the three sets~$\wt{F}_{i-1,z}$, $\wt{F}^+_{i-1,z}$, and $\wt{F}^-_{i-1,z}$, then statement~(i) holds due to the inductive hypothesis for~(i), because the number of candidates nominated in~$\EE$ in each  of these three sets will be~$3^i$; statement~(ii) also holds, because each nominee in~$\wt{F}_{i-1,z}$ defeats half of the nominees not in~$\wt{F}_{i,z}$ due to the inductive hypothesis for~(ii), as well as 
    all nominees in~$\wt{F}^+_{i-1,z}$, which is exactly half of the nominees not in~$\wt{F}_{i-1,z}$.
    
    So suppose that one of the sets~$\wt{F}_{i-1,z}$, $\wt{F}^+_{i-1,z}$, and $\wt{F}^-_{i-1,z}$ contains more than~$3^i$ nominees in~$\EE$.
    In such a case, facts~(a)--(c) together with Lemma~\ref{lem:acyclic-flat} applied to the restriction of~$\EE$ to the nominees in~$\wt{F}_{i,z}$ imply that 
    there exists at least one nominee~$c^\star$ in~$\wt{F}_{i,z}$ who defeats more than~$\frac{3 \cdot 3^i-1}{2}$ nominees from~$\wt{F}_{i,z}$ in~$\EE$. 
    Due to our inductive hypothesis on statements~(i) and~(ii), we also know that 
    candidate~$c^\star$ defeats $\frac{3 \cdot 3^q-3 \cdot 3^i}{2}$ nominees that are not in~$\wt{F}_{i,z}$.
    This yields that 
    \[\Cplodd(c^\star) \geq \frac{3\rho-3 \cdot 3^i}{2}+\frac{3 \cdot 3^i+1}{2}>\frac{3\rho-1}{2}\]
    in total, a contradiction to the flatness of~$\EE$.
    This proves the induction. 

    In particular, statement~(i) for $i=0$ yields that for each $z \in C_q$, there are exactly three nominees in the set~$\wt{F}_{0,z}=F_{\psi(z)}$; here we used that $Z_{0,z}=\{z\}$. Since this holds for each $z \in C_q$, and $\psi$ is a bijection, we obtain that $F_t$ for each team~$t \in T$ contains exactly three nominees in~$\EE$.

    Using statement~(ii) for~$i=0$ we also obtain that each nominee in~$F_t$ for some~$t \in T$ defeats half of the nominees not in~$F_t$; hence if any nominee in~$F_t$ would defeat both nominees in~$F_t$ other than itself, then it would obtain a Copeland-score of~$\frac{3|T|-3}{2}+2$, which again contradicts the flatness of~$\EE$.
\end{claimproof}

\smallskip
By Claim~\ref{clm:flatness}, we know that for each team~$t \in T$, the election~$\EE_X$ restricted to the three nominees contained in~$F_t$ is flat, i.e., each of these three nominees defeats exactly one other nominee. 
Recall that in this restricted election, the preferences of~$v$, $v'$, and~$v''$ are exactly the lists given by~$F_t$, $F'_t$, and~$F''_t$, respectively, as defined in~(\ref{eqn:def-Fr-mixed})--(\ref{eqn:def-Fd}).

\begin{claim}
    \label{clm:cyclic}
    If an election with three voters and three candidates is flat, then it contains no two candidates that are ranked the same way by all three voters.
\end{claim}
\begin{claimproof}
    Assume w.l.o.g.\ that candidate~$x$ is preferred to candidate~$y$ by all three voters. Since~$y$ also must defeat a candidate due to the flatness of the election, the remaining candidate~$z$ must follow~$y$ in the preference list of at least two voters. However, then $x$ also defeats~$z$, contradicting the assumption of flatness.
\end{claimproof}

We next argue how nominations can result in a flat election when restricted to the candidate set of certain team lists.
\begin{claim}
\label{clm:agents-nominate-oneroom}
    If a set of nominations for the election~$\EE_X$ results in a flat election, then for each $p \in S \cup C \cup \hat{C}$, exactly one candidate of the form~$p^r$ for some $r \in R$ is nominated.
\end{claim}
\begin{claimproof}
By Claim~\ref{clm:flatness}, we know that exactly three candidates from~$F_p$ are nominated.
First assume that $p$ is a single or couple of degree~$2$ in~$G$, or a copy of such a couple. In this case, there are two parties $P_p^{r_1}$ and~$P_p^{r_2}$ associated with~$p$ for some rooms~$r_1$ and~$r_2$. Moreover, the candidates~$p$ and~$p'$ both form a singleton party, and thus must be nominated. By the definition of~$F_p$ as given in~(\ref{eqn:def-Fp-2choices}), there are two additional  candidates in~$F_r$, namely $\neg p^{r_1}$ and~$\neg p^{r_2}$. Due to Claim~\ref{clm:flatness}, exactly one of these must be nominated, which means that exactly one of the candidates~$p^{r_1}$ and~$p^{r_2}$ is nominated.

Assume now that $p$ is a single or couple of degree~$3$ in~$G$, or a copy of such a couple. In this case, there are three parties $P_p^{r_1}, P_p^{r_2}$ and~$P_p^{r_3}$ associated with~$p$ for some rooms~$r_1,r_2$, and~$r_3$. Moreover, candidates~$p$ and~$p'$ together form a party, and thus exactly one of them is nominated. By the definition of~$F_p$ as given in~(\ref{eqn:def-Fp-3choices}), there are three additional  candidates in~$F_r$, namely $\neg p^{r_1}, \neg p^{r_2}$, and~$\neg p^{r_3}$. Due to Claim~\ref{clm:flatness}, exactly two of these must be nominated, which means that exactly one of the candidates~$p^{r_1},p^{r_2}$, and~$p^{r_3}$ is nominated.
\end{claimproof}

The following claim finishes the proof of this direction.
\begin{claim}
\label{clm:nomination-yields-matching}
    If a set of nominations for the election~$\EE_X$ results in a flat election, 
    then~$M=\{pr:p \in S \cup C,r \in R, p^r$ is nominated $\}$ is a complete matching for~$G$.
\end{claim}
\begin{claimproof}
Claim~\ref{clm:agents-nominate-oneroom} implies that $M$ contains exactly one edge incident to each vertex in~$S \cup C$, so it suffices to prove that
\begin{equation}
\label{eqn:capacity-of-room}
    |M(r) \cap S|+2 |M(r) \cap C| \leq 2
\end{equation}
holds for each room~$r \in R$.

Consider a room~$r \in R$. Recall that due to Claim~\ref{clm:flatness}, there must be exactly three nominees among the candidates of~$F_r$.
Recall that $P_r \subseteq F_r$, and that exactly one of the candidates in~$P_r$ is nominated. 
Let $x$ and~$y$ denote the other two nominees in~$F_r$.  
We will distinguish between three cases. 
 
First, suppose that $r$ is adjacent to one couple~$c$ and two singles, say $s_1$ and~$s_2$, in~$G$;
so that is $F_r$ is defined by~(\ref{eqn:def-Fr-mixed}).
Observe that if $x$ belongs to some party~$P_c$ or~$P_{\hat{c}}$ associated with a couple~$c \in C$ or its copy,   
then $y$ cannot belong to some party~$P_s$ associated with a single~$s \in S$, because then $y$ would be preferred to~$x$ by all three voters, which cannot happen by Claim~\ref{clm:cyclic}.
Therefore, only two cases are possible: 
either $\{x,y\}=\{c^r,\hat{c}^r\}$, which means that $M(r)=\{c\}$, or $\{x,y\}=\{s_1^r,s_2^r\}$, which means that $M(r)=\{s_1,s_2\}$. In either case, (\ref{eqn:capacity-of-room}) holds.

Second, assume that $r$ is only adjacent to singles in~$G$; 
then $F_r$ is defined by~(\ref{eqn:def-Fr-singles}). Due to Claim~\ref{clm:flatness}, we know that exactly two among the candidates $\{s^r:sr \in E\}$ must be nominated besides the nominee of~$P_r$; hence, $M(r)$ contains exactly two singles, 
and thus (\ref{eqn:capacity-of-room}) holds.

Third, assume that $r$ is only adjacent to couples in~$G$;
then $F_r$ is defined by~(\ref{eqn:def-Fr-couples}).
Again, let~$x$ and~$y$ denote the two nominees in~$F_r$ besides the nominee of~$P_r$.  
First note that $x$ and $y$ cannot both be of the form~$c^r$ for some~$c \in C$, as in that case both of them would defeat~$P_r$, which 
cannot happen by Claim~\ref{clm:flatness}. It also cannot happen that both are of the form~$\hat{c}^r$ for some~$c \in C$, as in that case $P_r$ would defeat both of them, again a contradiction to Claim~\ref{clm:flatness}. Hence, it must be the case that $\{x,y\}=\{c_i^r,\hat{c}_j^r\}$ for some couples~$c_i$ and~$c_j$. 
Notice that $c_i^r$ defeats the nominee of~$P_r$, and thus, by Claim~\ref{clm:flatness}, must in turn be defeated by~$\hat{c}_j^r$. 
This implies that $\hat{c}_j^r$ must precede~$c_i^r$ both in~$F'_r$ and in~$F''_r$, which can only happen if $c_i=c_j$.
Hence, in this case we have $\{x,y\}=\{c_i^r,\hat{c}_i^r\}$ for some couple~$c_i \in C$, which means that $M(r)=\{c_i\}$, ensuring~(\ref{eqn:capacity-of-room}). 
\end{claimproof}

\medskip
\noindent
{\bf Direction ``$\Longrightarrow:$''}
    Assume now that $M \subseteq E$ is a complete matching for our instance of \textsc{Maximum Matching with Couples}.

    \begin{claim}
    \label{clm:matching-gives-flat}
        Given a complete matching for~$G$, there exists a set of nominations for the constructed instance such that for each team~$t \in T$, the restriction of~$\EE_x$ to the candidates of~$F_t$ is a flat election with three nominees.
    \end{claim}
    \begin{claimproof}
    For each edge~$sr \in E$ incident to a single~$S$, let party~$P_s^r$ nominate the candidate~$s^r$ if $sr \in M$, and $\neg s^r$ otherwise.
    Similarly, for each edge~$cr \in E$ incident to a couple~$c$, let $P^r_c$ nominate~$c^r$ (and $P_{\hat{c}}$ nominate~$\hat{c}^r$) if $cr \in M$, and~$\neg c^r$ ($\neg \hat{c}^r$, respectively) otherwise.
    
    Given a room~$r$, we define the nominee of~$P_r$ as follows.
    If $M(r)$ contains a single that appears as the first candidate in~$F_r$, then let $P_r$ nominate~$r$, and $r'$ otherwise.     
    Based on the definitions~(\ref{eqn:def-Fr-mixed})--(\ref{eqn:def-Fr-couples}), it is straightforward to check that these nominations
    fulfill the requirements of the claim for~$F_t$. 
    
    If $p \in S \cup C \cup C'$ is such that $p$ and~$p'$ are both singletons, then the team lists for~$p$ are defined by~(\ref{eqn:def-Fp-2choices}). Using that exactly one candidate of the form~$p^r$ and, consequently, one candidate of the form~$\neg p^r$ is nominated, it follows that the restriction of~$\EE_X$ to the three nominees in~$F_p$ is a flat election.
    
    By contrast, if $p \in S \cup C \cup C'$ is such that candidates~$p$ and~$p'$ belong to the same party, then 
    the team lists for~$p$ are defined by~(\ref{eqn:def-Fp-3choices}). Since~$p$ (or the couple whose copy is~$p$) is incident to exactly one edge of~$M$, we know that there are exactly two nominees in~$F_p$ of the form~$\neg p^r$ for some $r \in R$. Let us nominate candidate~$p'$ if the first candidate of~$F''_p$ is not nominated, i.e., it is $\neg p^r$ where $p r \in M$; otherwise nominate the candidate~$p$. It is straightforward to check that in this case, the nominations again fulfill the requirements for~$F_p$.
        
    Finally, it is trivial that all possible nominations fulfill the requirements for the team lists corresponding to some dummy~$d \in D$, as defined in~(\ref{eqn:def-Fd}).
    \end{claimproof}

    Consider the nominations guaranteed by Claim~\ref{clm:matching-gives-flat}.
    Let $t \in T$ be a team.
    Notice that by %
    Observation~\ref{obs:bijection-property}, a nominee in~$F_t$ defeats a nominee in~$F_{t'}$ for distinct teams~$t$ and~$t'$ if and only if $\psi^{-1}(t)$ defeats~$\psi^{-1}(t')$ in~$\EE_q$. Hence, Observation~\ref{obs:EEq-number-of-defeats} means that each nominee~$x$ in~$F_t$ defeats exactly $\frac{(3|T|-3)}{2}$ nominees from~$X$ that are not contained in~$F_t$. 
    By Claim~\ref{clm:matching-gives-flat}, we also know that $x$ defeats exactly one nominee in~$F_t$. 
    From this we get
    $\Cplodd(x)=\frac{(3|T|-3)}{2}+1=\frac{(3|T|-1)}{2}$. Hence, the proposed nominations for~$\EE_X$ indeed result in a flat election.
    This finishes the proof of the correctness of our reduction.

\subsection{Proof of Theorem~\ref{thm:4voters-s=2}}
\label{sec:proof-of-4votersNPc}
\thmfourvotersmaxsizetwo*

\begin{proof}
    We present a reduction from \textsc{3-Coloring} that uses some ideas from the proof of Theorem~\ref{thm:Copeland-v=s=2}. Let $G=(U,E)$ be our input graph with vertex set $U=\{u_1,\dots,u_n\}$. We construct an instance of \PP\ with voter set $\{v,v',w,w'\}$ and $\maxsize=2$ as follows. 
    
    We let $P=\{p\}$ be our distinguished party, and we create candidates~$q_i$ for $i \in [n]$, with each $q_i$ forming its own singleton party. For each color~$c \in [3]$, we further create a party $P_{u_i}^c=\{u^c_i,\ol{u}^c_i\}$ for each $i \in [n]$, and a party $P^c_e=\{e^c_i,e^c_j\}$ for each edge $e= u_i u_j \in E$. 
    Then the maximum party size is $\maxsize=2$.
    We shall use the notation 
    {
    \allowdisplaybreaks
    \begin{align*}    
    U^{[3]} & = \{u^c_i:i \in [n], c \in [3]\}, \\
    \ol{U}^{[3]} & = \{\ol{u}^c_i:i \in [n], c \in [3]\}, \\
    E^{[3]} & = \{e^c_i:e \in E \textrm{ is incident to } u_i, i \in [n], c \in [3]\}, \\
    Q &= \{q_1,\dots,q_n\},
    \end{align*}
    so that our candidate set is $C_0=\{p\} \cup Q \cup U^{[3]} \cup \ol{U}^{[3]} \cup E^{[3]}$.
    }
    
    To define the preferences of our voters, we will use the definition of blocks~$X^c$ and~$\wt{X}^c$ for each color $c \in [3]$ as defined in~(\ref{eq:Xblocks-def}), and we additionally introduce the following blocks:
\begin{align}
\label{eq:Zblocks-def}
\begin{split}
Z &= \ol{u}^1_1, \ol{u}^2_1, \ol{u}^3_1, q_1,
\ol{u}^1_2, \ol{u}^2_2, \ol{u}^3_2, q_2, \dots,
\ol{u}^1_n, \ol{u}^2_n, \ol{u}^3_n, q_n; \\
\wt{Z} &= \ol{u}^1_n, \ol{u}^2_n, \ol{u}^3_n, q_n,
\ol{u}^1_{n-1}, \ol{u}^2_{n-1}, \ol{u}^3_{n-1}, q_{n-1}, \dots,
\ol{u}^1_1, \ol{u}^2_1, \ol{u}^3_1, q_1.
\end{split}
\end{align}

Now we are ready to define the preferences of our voters:
\begin{align*}
    v &: p, \ora{Q}, X^1, X^2, X^3, \ora{\ol{U}^{[3]}}; \\
    v' &: \wt{X}^3, \wt{X}^2,\wt{X}^1, \ola{\ol{U}^{[3]}}, \ola{Q}, p; \\
    w &: \ora{E^{[3]}}, p, Z, \ora{U^{[3]}}; \\
    w' &: \wt{Z}, \ola{U^{[3]}}, p, \ola{E^{[3]}}.
\end{align*}

For some reduced election~$\EE$ with candidate set~$C \subset C_0$, let $\lambda_\EE(x)$ for each $x \in C$ denote the set of nominees that defeat~$x$ in~$\EE$. 
The following observations follow directly from the preferences of voters:
\begin{itemize}
    \item $\lambda_\EE(p)=\emptyset$;
    \item $\lambda_\EE(u^c_i)\supseteq Q$ for each nominee $u_i^c \in U^{[3]}$; 
    \item $\lambda_\EE(\ol{u}^c_i)\supseteq E^{[3]} \cap C$ for each nominee $\ol{u}_i^c \in \ol{U}^{[3]}$.
\end{itemize}

Notice also that candidates in~$Q$ either are tied  with or defeat every nominee in~$C \setminus \ol{U}^{[3]}$, and thus can be defeated only by nominees in~$\ol{U}^{[3]}$. However, due to the definition of blocks~$Z$ and~$\wt{Z}$, each nominee~$q_i \in Q$ is tied with all nominees in~$\ol{U}^{[3]} \setminus \{\ol{u}_i^1,\ol{u}_i^2,\ol{u}_i^3\}$. Thus, we also get
\begin{equation}
\label{eqn:qi_induces_one_color}
    \lambda_\EE(q_i) \subseteq \{\ol{u}_i^1,\ol{u}_i^2,\ol{u}_i^3\}
\end{equation}
for each $i \in [n]$.

Similarly, note that candidates in~$E^{[3]}$ are either tied  with or defeat every nominee in~$C \setminus U^{[3]}$, and thus can be defeated only by nominees in~$U^{[3]}$. However, due to the definition of blocks~$X^c$ and~$\wt{X}^c$ for some color~$c \in [3]$, each nominee~$e^c_i \in E^{[3]}$ is tied with all nominees in~$U^{[3]} \setminus \{u_i^c\}$. Thus, we also get
\begin{equation}
\label{eqn:eic_ensures_proper_coloring}
  \lambda_\EE(e_i^c) \subseteq \{u_i^c\}
\end{equation}
for each nominee~$e_i^c$ in~$\EE$.

Since $\Cpl[1](x)=t-1-\lambda_\EE(x)$ for each nominee $x \in C$, we get that
\begin{equation}
\label{eqn:3Col2-pscore}
    \Cpl[1](p) = t-1,     
\end{equation}
which shows that $p$ is the unique winner in~$\EE$ if and only if all nominees in~$Q \cup E^{[3]}$ are defeated by at least one other nominee.

Moreover, for each $i \in [n]$ and $c \in [3]$ we obtain that
\begin{align}
    \label{eqn:3Col2-Uscore}
    \Cpl[1](u^c_i) & \leq  t-1-|Q|=t-1-n && \textrm{if } u_i^c \in C; \\
    \label{eqn:3Col2-negUpscore}
    \Cpl[1](\ol{u}^c_i) & \leq  t-1-|E^{[3]} \cap C|=t-1-3|E| && \textrm{if } \ol{u}_i^c \in C. 
\end{align}

We claim that $p$ can be the unique winner in the reduced election~$\EE$ if and only if $G$ admits a proper $3$-coloring.

\medskip 
\noindent
{\bf Direction ``$\Longrightarrow$'':}
Suppose first that $p$ is the unique winner in~$\EE$. Since each candidate~$q_i \in Q$ must be defeated by at least one nominee in~$\EE$ due to~(\ref{eqn:3Col2-pscore}), by (\ref{eqn:qi_induces_one_color}) we know that at least one candidate in~$\{\ol{u}_i^1,\ol{u}_i^2,\ol{u}_i^3\}$ must be nominated in~$\EE$ for each $i \in [n]$. Thus, for each $i \in [u]$ we can fix one candidate~$\ol{u}_i^c$ that is nominated in~$\EE$, and assign the color~$c$ to~$u_i$. 
We argue that the coloring~$\chi$ thus defined is a proper coloring for~$G$.

For the sake of a contradiction, assume that there exists some edge~$e =u_i u_j \in E$ and color~$c \in [3]$ so that $\chi(u_i)=\chi(u_j)=c$. This means that both $\ol{u}_i^c$ and $\ol{u}_j^c$ are nominated in~$\EE$, and consequently, neither $u_i^c$ nor $u_j^c$ are nominated due to the definitions of parties~$P_{u_i}^c$ and~$P_{u_j}^c$. By (\ref{eqn:eic_ensures_proper_coloring}), this implies that if candidate~$e_i^c$ is nominated, then $\lambda(e_i^c)=\emptyset$, and similarly, if $e_j^c$ is nominated, then $\lambda(e_j^c)=\emptyset$. Therefore,  $\Cpl[1](P_e^c)=t-1=\Cpl[1](p)$ in both cases, contradicting our assumption that $p$ is the unique winner in~$\EE$. Hence, $\chi$ is indeed a proper 3-coloring for~$G$.

\medskip 
\noindent
{\bf Direction ``$\Longleftarrow$'':}
Suppose now that $\chi:U \rightarrow [3]$ is a proper 3-coloring of~$G$. We provide a nomination strategy for all parties which yields a reduced election~$\EE$ whose unique winner is $p$. 

First, for each~$i \in [n]$ we nominate $\ol{u}_i^{\chi(u_i)}$ together with the candidates in~$\{u_i^{c'}:c' \neq \chi(u_i)\}$. By~(\ref{eqn:qi_induces_one_color}), this immediately yields $\lambda_\EE(q_i)=\{\ol{u}_i^{\chi(u_i)}\}$ and hence 
\begin{equation}
\label{eqn:3Col2-qscore}
    \Cpl[1](q_i)=t-2.
\end{equation}

Next, for each $e \in E$ and color~$c \in [3]$, let us nominate the candidate~$e_i^c$ where $u_i$ is an (arbitrarily fixed) endpoint of~$e$ that does \emph{not} have color~$c$; as we have argued in the proof of Theorem~\ref{thm:Copeland-v=s=2}, $e$ must admit such an endpoint, because $\chi$ is a proper 3-coloring.
Note that in this case, candidate~$u^c_i$ is nominated (by $c \neq \chi(u_i)\}$), so the nominee $e_i^c$ is therefore defeated by~$u_i^c$. Hence, $\Cpl[1](e_i^c) \leq t-2$ follows. Together with our previous observations in (\ref{eqn:3Col2-pscore})--(\ref{eqn:3Col2-qscore}), this means that all nominated candidates have score at most~$t-2$, while $p$ has score~$t-1$. Thus $p$ is indeed the unique winner in~$\EE$. 
\end{proof}

\section{Additional material for Section~\ref{sec:parameter-t}}
\label{app:param-t}
All reductions in this section are from the 
\textsc{Multicolored Clique} problem which 
is $\mathsf{W}[1]$-hard when parameterized by the size of the desired clique~\cite{pietrzak-multicolored-2003,fellows-hermelin-rosamond-vialette-multicolored-hardness}.

\begin{center}
\fbox{ 
\parbox{0.95\columnwidth}{
{\bf \textsc{Multicolored Clique}}. \\
{\bf Input:} An undirected graph $G=(U,E)$ with the vertex set~$U$ partitioned into $k$ independent sets $U_1, \dots, U_k$. \\
{\bf Question:} Is there a clique of size~$k$ in~$G$?
}}
\end{center}

\medskip

We now define some notation that will be useful in  our reductions from \textsc{Multicolored Clique}.

For each pair of distinct indices~$i,j \in [k]$  let  $E_{\{i,j\}}$ denote the set of edges in~$G$ that run between~$U_i$ and~$U_j$; note that  $E_{\{i,j\}}=E_{\{j,i\}}$. 
We may assume w.l.o.g.\ that $|U_i|=n$ 
and $|E_{\{i,j\}}|=m$ for every $i,j \in [k]$ with $i \neq j$; 
thus, we will write $U_i=\{u_i^1,\dots,u_i^n\}$ and 
$E_{\{i,j\}}=\{e_{\{i,j\}}^1,\dots,e_{\{i,j\}}^m\}$. 
Indeed, to achieve this we can simply add a newly constructed independent edge set of appropriate size to each set~$E_{\{i,j\}}$ to guarantee $|E_{\{i,j\}}|=m$ for all $i \neq j$, and then add the necessary number of isolated vertices to each vertex set $U_i$ so that they all have the same size.

For each vertex $u \in U_i$, we let $E^j(u)=\{uv \in E: v \in U_j\}$ for each $j \in [k] \setminus \{i\}$, as well as the edge sets $E^{>i}(u)=\bigcup_{j \in [k], j>i} E^j(u)$ and $E^{<i}(u)=\bigcup_{j \in [k], j<i} E^j(u)$.

\subsection{Proof of Theorem~\ref{thm:4voters}.}
\label{sec:param-t-even}

\thmCopelandfourteven*

We will prove Theorem~\ref{thm:4voters} in two steps, dissecting it into Theorems~\ref{thm:4voters-alpha-lessthanone}
and~\ref{thm:Copeland-4voters-alpha}.
We start with the following result.

\begin{restatable}[$\star$]{theorem}{thmfourvotersalphalessthanone}
\label{thm:4voters-alpha-lessthanone}
For any constant $\alpha \in [0,1)$, \PP{} for Copeland$^{\, \alpha}$ is $\mathsf{W}[1]$-hard when parameterized by~$t$, the number of parties, even if there are only four voters.    
\end{restatable}

\begin{proof}  Let $G=(U,E)$ and $k$ be our instance of the \textsc{Multicolored Clique} with $U$ partitioned into sets $U_1, \dots, U_k$;
we use all the notation introduced at the beginning of Appendix~\ref{app:param-t}.
We define an instance of \PP{} as follows.

First, we set $P=\{p\}$ as our distinguished party.
We also introduce a party~$P'=\{p'\}$.
We add $U_i$ for each $i \in [k]$, as well as~$E_{\{i,j\}}$ for each pair $\{i,j\} \subseteq [k]$ with $i \neq j$ as parties. Thus, we have $t=2+k+\binom{k}{2}$ parties. Our four voters will be $v,v',w,$ and~$w'$. To define their preferences, we first construct the following blocks for each $i \in [k]$:
\begin{align*}
F^>_i &= u_i^1, \ora{E^{>i}(u_i^1)},
u_i^2, \ora{E^{>i}(u_i^2)},
\dots,
u_i^n, \ora{E^{>i}(u_i^n)}; \\
\wt{F}^>_i &= \ola{E^{>i}(u_i^1)}, u_i^1, 
\ola{E^{>i}(u_i^2)}, u_i^2, 
\dots,
\ola{E^{>i}(u_i^n)}, u_i^n; \\
F^<_i &= u_i^1, \ora{E^{<i}(u_i^1)},
u_i^2, \ora{E^{<i}(u_i^2)},
\dots,
u_i^n, \ora{E^{<i}(u_i^n)}; \\
\wt{F}^<_i &= \ola{E^{<i}(u_i^1)}, u_i^1, 
\ola{E^{<i}(u_i^2)}, u_i^2, 
\dots,
\ola{E^{<i}(u_i^n)}, u_i^n.
\end{align*}
Note that blocks~$F^>_i$ and~$\wt{F}^>_i$ contain exactly those edges that have one endpoint in~$U_i$ and the other endpoint in some set~$U_j$ with $j>i$.
Similarly, blocks~$F^<_i$ and~$\wt{F}^<_i$ contain exactly those edges that have one endpoint in~$U_i$ and the other endpoint in some set~$U_j$ with $j<i$.
This also means that candidates of~$E_{\{i,j\}}$ for some $1 \leq i<j \leq k$ are contained precisely in the blocks~$F_i^>,\wt{F}_i^>,F^<_j,$ and $\wt{F}^<_j$.

We are now ready to define the preferences of our voters:
\begin{align*}
v &: p, p', F^>_1, F^>_2, \dots, F^>_{k-1}, F^>_k; \\
v' &:  \wt{F}^>_k, \wt{F}^>_{k-1}, \dots, \wt{F}^>_2, \wt{F}^>_1, p, p'; \\
w &: p, p', F^<_1, F^<_2, \dots, F^<_{k-1}, F^<_k; \\
w' &:  \wt{F}^<_k, \wt{F}^<_{k-1}, \dots, \wt{F}^<_2, \wt{F}^<_1, p, p'.
\end{align*}

Observe that in every reduced election~$\EE$ resulting from some nominations, $p$ receives a score of $\Cpl(p)=1+\alpha(t-2)$, because it defeats~$p'$ and is tied with all remaining nominees. Since $\alpha<1$, this implies that $p$ is the unique winner of~$\EE$ if and only if no nominee other than~$p$ defeats any other nominee in~$\EE$.
We show that such an election~$\EE$ exists if and only if $G$ contains a clique of size~$k$.

\smallskip
\noindent
{\bf Direction ``$\Longrightarrow$'':} Suppose first that $p$ is the unique winner in~$\EE$, and let $u_i^{\kappa(i)}$ denote the nominee of party~$U_i$ 
for each $i \in [k]$.
We are going to show that the vertices in $K=\{u_i^{\kappa(i)}:i \in [k]\}$ form a clique in~$G$. 
Fix some $i,j \in [k]$ with $i<j$, and let $e=u_i^{h}u_j^{\ell}$ be the edge nominated by party~$E_{\{i,j\}}$. 

We claim that $e$ is incident to~$u_i^{\kappa(i)}$, that is, $h=\kappa(i)$. To see this, first note that $w$ and $w'$ 
rank $u_i^{\kappa(i)}$ (contained in blocks~$F^<_i$ and~$\wt{F}^<_i$) 
and $e$ (contained in blocks~$F^<_j$ and~$\wt{F}^<_j$) differently. 
Moreover, $v$ and $v'$ both prefer $u_i^{\kappa(i)}$ to~$e$ 
(both contained in blocks~$F^>_i$ and~$\wt{F}^>_i$) exactly if $\kappa(i) < h$, while 
they both prefer $e$ to~$u_i^{\kappa(i)}$ if and only if $\kappa(i)>h$. Since $u_i^{\kappa(i)}$ and~$e$ must be tied in~$\EE$ (as do every pair of candidates not involving~$p$), we get that only $h=\kappa(i)$ is possible, proving our claim. 

An analogous argument shows that $e$ is incident to~$u_j^{\kappa(j)}$, that is, $\ell=\kappa(j)$: 
First, note that $v$ and~$v'$ rank $e$ and $u_j^{\kappa(j)}$ differently, because the former is contained in blocks~$F_i^>$ and~$\wt{F}_i^>$, while the latter is contained in blocks~$F_j^>$ and~$\wt{F}_j^>$. 
Moreover, if $\ell<\kappa(j)$, then both $w$ and~$w'$ prefer~$e$ to~$u_j^{\kappa(j)}$, while if 
$\ell>\kappa(j)$, then both $w$ and~$w'$ prefer~$u_j^{\kappa(j)}$ to~$e$. Since these two candidates must be tied in~$\EE$, we obtain $\ell=\kappa(j)$ as required.

Hence, $e$ is the edge connecting $u_i^{\kappa(i)}$ and~$u_j^{\kappa(j)}$.
Since we proved this for arbitrary indices $i,j \in [k]$ with $i \neq j$, it follows that $K$ forms a clique in~$G$.

\smallskip
\noindent
{\bf Direction ``$\Longleftarrow$'':}
Suppose now that there exists a clique~$K$ in~$G$. We claim that nominating the vertices and edges of~$K$ results in an election~$\EE$ in which $p$ is the unique winner. By our previous observations, it suffices to prove that there exists no nominee~$c$ other than~$p$ that defeats another nominee in~$\EE$. 

First, $p'$ clearly cannot defeat any other candidate. Second, for every $i,j \in [k]$, the nominee of~$U_i$ cannot defeat the nominee of~$U_j$: voters~$v$ and~$v'$ rank them differently, because they are in different blocks within the preferences of~$v$ and~$v'$, and the same holds for voters~$w$ and~$w'$ as well.
Third, for indices $i,i',j,j' \in [k]$ with $i \neq j$ and~$i' \neq j'$, the nominee of $E_{\{i,j\}}$ cannot defeat the nominee of~$E_{\{i',j'\}}$, 
because $v$ and~$v'$, as well as $w$ and~$w'$, rank them differently, irrespective whether they belong to different blocks (in which case the ordering of blocks ensures our claim) or not (in which case they share an endpoint, and therefore the ordering of edges within a single block guarantees our claim).

Hence, it remains to show that the nominee of~$U_i$, say $u_i^{\kappa(i)}$, and the nominee of some party~$E_{\{i',j'\}}$ where $1 \leq i' < j' \leq k$ are tied in~$\EE$.
First, $i \neq i'$, then these two nominees are in different blocks in the votes of~$v$ and~$v'$, hence are ranked differently by them. If $i=i'$, then they are again ranked differently by~$v$ and~$v'$, because the nominee of $E_{\{i',j'\}}$ (an edge incident to some vertex in~$U_{i'}=U_i$, namely, to~$u_i^{\kappa(i)}$) must be contained in~$E^{>i}(u_i^{\kappa(i)})$, and hence is preferred to~$u_i^{\kappa(i)}$ by~$v'$, while is less preferred than~$u_i^{\kappa(i)}$ by~$v$. Hence, in either case, $v$ and~$v'$ rank these two nominees differently.
The analogous argument shows that $w$ and~$w'$ also rank them differently, implying that they are indeed tied in~$\EE$. 
This proves that $p$ is indeed the unique nominee that defeats another nominee, and hence is the unique winner in~$\EE$, proving the correctness of our reduction.
\end{proof}

Recall that the intractability result stated in Theorem~\ref{thm:4voters-alpha-lessthanone} does not hold for $\alpha=1$. The following theorem fills this gap by a more sophisticated reduction. 

\begin{restatable}[$\star$]{theorem}{thmCopelandfourvotersalpha}
\label{thm:Copeland-4voters-alpha}
For any constant $\alpha \in (0,1]$, \PP{} for Copeland$^{\, \alpha}$ is  
$\mathsf{W}[1]$-hard when parameterized by~$t$, the number of parties, even if there are only four voters.
\end{restatable}

\begin{proof}
Let $G=(U,E)$ and $k$ be our instance of the \textsc{Multicolored Clique} with $U$ partitioned into sets~$U_1, \dots, U_k$;
again, we use all the notation introduced at the beginning of Appendix~\ref{app:param-t}.
We define an instance of \PP{} as follows. 

First, we set $P=\{p\}$ as our distinguished party.
We also introduce a party~$P'=\{p'\}$.
We add $U_i$ for each $i \in [k]$, as well as~$E_{\{i,j\}}$ for each pair $\{i,j\} \subseteq [k]$ with $i \neq j$ as parties. 
We next define the following copies of these parties as follows:

\vspace{4pt}
\begin{tabular}{ll}
 $U_i^{(c)}=\{u_i^{h,(c)}:h \in [n] \}$ & for $i \in [k]$ and  $c \in [k] \setminus \{i\}$; 
 \\
 $E_{\{i,j\}}^{(c)}=\{e_{\{i,j\}}^{h,(c)}:h \in [m] \}$ & for $1 \leq i<j \leq k$ and $c \in \{0,1\}$; \\
 $A_{\{i,j\}}=\{a_{\{i,j\}}^h:h \in [m] \}$ & for $1 \leq i<j \leq k$; 
\end{tabular}
\vspace{4pt}

We remark that for some edge~$e \in E_{\{i,j\}}$ we may denote its two copies as~$e^{(0)}$ and~$e^{(1)}$, ignoring that this edge~$e$ must have the form $e=e_{\{i,j\}}^h$ for some $h \in [m]$.

It will be useful for us to define the following two sets of parties:
\begin{align*}
\SS^\orig& =\bigcup_{i \in [k]} \{U_i\} \cup \bigcup_{1 \leq i<j \leq k}  \{E_{\{i,j\}} \}; \\
\SS^\cop& = \bigcup_{1 \leq i<j \leq k}  \{ U_i^{(j)},U_j^{(i)},
E_{\{i,j\}}^{(0)},E_{\{i,j\}}^{(1)} \};
\end{align*}

We further create certain candidate sets that are \emph{not} parties; namely, we create sets~$C_{\{i,j\}}$ and~$D_{\{i,j\}}$ for each $i, j \in [k]$ with $i<j$, sets~$C_p$ and $D_p$,  
and sets~$B_S$ and~$D_S$ for each party $S \in \SS^\orig$. 
Each of these sets contains two candidates, both of them forming a singleton party.
Hence, the number of parties is 
\[t=2+
k^2+4 \binom{k}{2}+2\left(2\binom{k}{2}+2+2|\SS^\orig|\right)=k^2+12\binom{k}{2}+4k+6.
\]

Next, we let $v,v',w,$ and $w'$ be our four voters.
We will define the preferences of these voters in a piecewise manner: each preference list will be constructed by concatenating several \emph{blocks}, where each block is a list of candidates. 

We first construct \emph{vertex-copy blocks} $Y_i$ and $\wt{Y}_i$ for each $i \in [k]$:
\begin{multline}
\label{eq:Yblock-def}
Y_i = u_i^1, \ora{\bigcup_{c \in [k] \setminus \{i\}} \{u_i^{1,(c)}\}},
u_i^2, \ora{\bigcup_{c \in [k] \setminus \{i\}} \{u_i^{2,(c)}\}}, \\
\dots,
u_i^n, \ora{\bigcup_{c \in [k] \setminus \{i\}} \{u_i^{n,(c)}\}}; 
\end{multline}
\vspace{-10pt}
\begin{multline}
\label{eq:wtYblock-def}
\wt{Y}_i = u_i^n, \ola{\bigcup_{c \in [k] \setminus \{i\}} \{u_i^{n,(c)}\}},
u_i^{n-1}, \ola{\bigcup_{c \in [k] \setminus \{i\}} \{u_i^{n-1,(c)}\}}, \\
\dots,
u_i^1, \ola{\bigcup_{c \in [k] \setminus \{i\}} \{u_i^{1,(c)}\}}.
\end{multline}

Similarly, we construct \emph{edge-copy blocks} $F_{\{i,j\}}$ and $\wt{F}_{{\{i,j\}}}$ for each~$i, j \in [k]$ with $i<j$:
\begin{flalign}
\notag
F_{\{i,j\}} = e_{\{i,j\}}^1,a_{\{i,j\}}^1, e_{\{i,j\}}^{1,(0)},& e_{\{i,j\}}^{1,(1)},
e_{\{i,j\}}^2,a_{\{i,j\}}^2, e_{\{i,j\}}^{2,(0)},e_{\{i,j\}}^{2,(1)},&%
\\
\label{eq:Fblock-def}
& \dots, 
e_{\{i,j\}}^m,a_{\{i,j\}}^m, e_{\{i,j\}}^{m,(0)},e_{\{i,j\}}^{m,(1)},     
\end{flalign} 
\vspace{-10pt}
\begin{flalign}
\notag
\label{eq:wtFblock-def}
\wt{F}_{\{i,j\}} = e_{\{i,j\}}^m,a_{\{i,j\}}^m, e_{\{i,j\}}^{m,(1)},& e_{\{i,j\}}^{m,(0)}, 
e_{\{i,j\}}^{m-1},a_{\{i,j\}}^{m-1}, e_{\{i,j\}}^{m-1,(1)},e_{\{i,j\}}^{m-1,(0)},&%
\\
& \dots,
e_{\{i,j\}}^1,a_{\{i,j\}}^1, e_{\{i,j\}}^{1,(1)},e_{\{i,j\}}^{1,(0)}.
\end{flalign}

Next, 
we construct blocks $\myDC$ and $\wt{\myDC}$ by setting
\begin{align*}
\myDC &= \ora{D_{\{1,2\}}},\ora{C_{\{1,2\}}}, 
\ora{D_{\{1,3\}}},\ora{C_{\{1,3\}}},\dots, \ora{D_{\{k-1,k\}}},\ora{C_{\{k-1,k\}}}; \\
\wt{\myDC} &= \ola{D_{\{k-1,k\}}},\ola{C_{\{k-1,k\}}}, 
\ola{D_{\{k-2,k\}}},\ola{C_{\{k-2,k\}}},\dots, \ola{D_{\{1,2\}}},\ola{C_{\{1,2\}}}. 
\end{align*}
Let us clarify that the block $\myDC$ contains all lists $(\ora{D_{\{i,j\}}}, \ora{C_{\{i,j\}}})$ where $1 \leq i<j \leq k$
in some arbitrarily fixed order,
whereas
$\wt{\myDC}$ also contains all lists $(\ola{D_{\{i,j\}}}, \ola{C_{\{i,j\}}})$ where $1 \leq i<j \leq k$, but exactly in the reversed ordering. 

Similarly, 
we construct blocks $\myCA$ and $\wt{\myCA}$ by setting
\begin{align*}
\myCA &= \ora{C_{\{1,2\}}},\ora{A_{\{1,2\}}}, 
\ora{C_{\{1,3\}}},\ora{A_{\{1,3\}}}, 
\dots, \ora{C_{\{k-1,k\}}},\ora{A_{\{k-1,k\}}}; \\
\wt{\myCA} &= \ola{C_{\{k-1,k\}}},\ola{A_{\{k-1,k\}}}, \ola{C_{\{k-2,k\}}},\ola{A_{\{k-2,k\}}},\dots, \ola{C_{\{1,2\}}},\ola{A_{\{1,2\}}}. 
\end{align*}

Next, 
we construct a block~$\myDB$ 
by concatenating the lists $(\ora{D_S},\ora{B_S})$ for each party $S \in \SS^\orig$ in some fixed order, and we obtain the block~$\wt{\myDB}$ by concatenating the lists $(\ola{D_S},\ola{B_S})$ for $S \in \SS^\orig$ in the reversed order. 

Similarly, 
we create a block $\myBX$ 
by concatenating the lists $(\ora{B_S},\ora{S})$ for each party $S \in \SS^\orig$ in some fixed order, and we obtain the block~$\wt{\myBX}$ by concatenating the 
lists $(\ola{B_S},\ola{S})$ for~$S \in \SS^\orig$ in the reversed order.
We also define the set 
\[D=\bigcup_{S \in \SS^\orig}D_S \cup \bigcup_{1 \leq i<j\leq k} D_{\{i,j\}}.\]

Next, let us define \emph{incidence blocks} $H_{(i,j)}$ and $\wt{H}_{(i,j)}$  for each indices $i,j \in [k]$ with $i \neq j$. If $i<j$, then we define %
\begin{multline}
\label{eq:Hblock-def1}
    H_{(i,j)} = \ora{\bigcup_{e \in E^j(u_i^1)}\{ e^{(0)} \}}, u_i^{1,(j)}, 
\ora{\bigcup_{e \in E^j(u_i^2)}\{ e^{(0)} \}}, u_i^{2,(j)}, \\
\dots, 
\ora{\bigcup_{e \in E^j(u_i^n)}\{ e^{(0)} \}}, u_i^{n,(j)}; 
\end{multline}
\vspace{-10pt}
\begin{multline}
\label{eq:wtHblock-def1}
\wt{H}_{(i,j)} = \ola{\bigcup_{e \in E^j(u_i^n)}\{ e^{(0)} \}}, u_i^{n,(j)}, 
\ola{\bigcup_{e \in E^j(u_i^{n-1})}\{ e^{(0)} \}}, u_i^{n-1,(1)}, \\
\dots,
\ola{\bigcup_{e \in E^j(u_i^1)}\{ e^{(0)} \}}, u_i^{1,(1)}.
\end{multline}
Otherwise, i.e., if $i>j$, then we define 
\begin{multline}
\label{eq:Hblock-def2}
H_{(i,j)} = \ora{\bigcup_{e \in E^j(u_i^1)}\{ e^{(1)} \}}, u_i^{1,(j)}, 
\ora{\bigcup_{e \in E^j(u_i^2)}\{ e^{(1)} \}}, u_i^{2,(j)}, \\
\dots, 
\ora{\bigcup_{e \in E^j(u_i^n)}\{ e^{(1)} \}}, u_i^{n,(j)};     
\end{multline}
\vspace{-10pt}
\begin{multline}
\label{eq:wtHblock-def2}
\wt{H}_{(i,j)} = \ola{\bigcup_{e \in E^j(u_i^n)}\{ e^{(1)} \}}, u_i^{n,(j)}, 
\ola{\bigcup_{e \in E^j(u_i^{n-1})}\{ e^{(1)} \}}, u_i^{n-1,(1)}, \\
\dots,
\ola{\bigcup_{e \in E^j(u_i^1)}\{ e^{(1)} \}}, u_i^{1,(1)}. 
\end{multline}

We are now ready to define the preference lists of all voters:
\begin{align*}
v &: \myDB,\myDC,Y_1,\dots,Y_k,F_{\{1,2\}},\dots,F_{\{k-1,k\}},p',p,\ora{D_p},\ora{C_p}; \\
v' &: \ola{D_p},\ola{C_p},p',p,\wt{F}_{\{k-1,k\}},\dots,\wt{F}_{\{1,2\}},\wt{Y}_k,\dots,\wt{Y}_1,\wt{\myDC},\wt{\myDB}; \\
w &: \myBX,\myCA,H_{(1,2)},\dots,H_{(k,k-1)},p,\ora{D},\ora{C_p},p'; \\
w' &: \ola{C_p},p',p, H_{(k,k-1)},\dots,H_{(1,2)}, \ola{D},\wt{\myCA},\wt{\myBX}.
\end{align*}

It is straightforward to verify that all candidates are listed in the preferences of each of the four voters. It is clear that the  presented reduction can be computed in polynomial time. Moreover, since the number of parties created is $t=O(k^2)$, it is  a parameterized reduction. %
It remains to prove its correctness.

\medskip
We start with some observations regarding the relationship between the parties in the constructed instance. 
Observing the preference lists of~$v$ and~$v'$, we obtain that both $v$ and~$v'$ 
\begin{itemize}
\item  prefer candidates in $D_S$ to candidates in~$B_S$ for each $S \in \SS^\orig$ due to the definition of~$\myDB$ and~$\wt{\myDB}$; %
\item  prefer candidates in $D_{\{i,j\}}$ to candidates in~$C_{\{i,j\}}$ for each $i,j \in [k]$ with $i<j$ due to the definition of~$\myDC$ and~$\wt{\myDC}$; %
\item  prefer candidate~$p'$ to $p$;
\item  prefer candidates in $D_p$ to candidates in~$C_p$.
\end{itemize}
Notice also that each pair of candidates that falls into one of the above categories is ranked differently by~$w$ and~$w'$.

Observing the
preference lists of~$w$ and~$w'$, we obtain that both~$w$ and~$w'$
\begin{itemize}
\item prefer  candidates in $B_S$ to candidates in~$S$ for each $S \in \SS^\orig$;
\item prefer candidates in $C_{\{i,j\}}$ to candidates in~$A_{\{i,j\}}$ for each $i,j \in [k]$ with $i<j$;
\item prefer  candidates in some party in~$\SS^\cop \cup\{P\}$ to candidates in~$D$;
\item prefer  candidates  in~$C_p$ to candidate~$p'$.
\end{itemize}
Again, each pair of candidates that falls into one of the above categories is ranked differently by~$v$ and~$v'$. 
Taking into account the observations of the previous paragraph (and considering again the preferences of voters where necessary),
we obtain that
\begin{enumerate}[label=(\alph*)]
\item\label{item-fact:p-rel} $p$ defeats all candidates in~$D$, is defeated by candidate~$p'$, and is tied with all remaining candidates;
\item\label{item-fact:copy-vs-D} each party in $\SS^\cop$ defeats all candidates in~$D$;
\item each party~$S \in \SS^\orig$ is defeated by both candidates in~$B_S$;
\item both candidates in~$B_S$ are defeated by both candidates in~$D_S$ for all $S \in \SS^\orig$;
\item both candidates in~$C_{\{i,j\}}$ defeat both candidates in~$A_{\{i,j\}}$ and are defeated by all candidates in~$D_{\{i,j\}}$;
\item\label{item-fact:DpCpBp} both candidates in~$C_p$ defeat candidate~$p'$ and are defeated by both candidates in~$D_p$.
\end{enumerate}
In particular, statement~\ref{item-fact:p-rel} means that for any reduced election~$\EE$ resulting from some nomination strategy, we get that the Copeland$^{\alpha}$-score of $p$ is exactly 
\begin{equation}
\label{eq:p-score}
\Cpl(p)=|D|+\alpha(t-|D|-2).    
\end{equation}

The key properties of a nomination strategy we need is captured in the following definition.

\begin{definition}
\label{def:valid_nomination}    
We call a set of nominations by parties in $\SS^\orig \cup \SS^\cop$ \emph{valid}, if
\begin{enumerate}[label=(\roman*)]
\item \label{cond:valid-U} the nominee of $U_i$ precedes the nominees of $U_i^{(c)}$ for each $c \in [k] \setminus \{i\}$ in both blocks $Y_i$ and $\wt{Y}_i$, for each $i \in [k]$;
\item \label{cond:valid-E}  the nominee of $E_{\{i,j\}}$ precedes the nominees of $E_{\{i,j\}}^{(c)}$ for both $c \in \{0,1\}$ in both blocks $F_{\{i,j\}}$ and $\wt{F}_{\{i,j\}}$, for each $i,j \in [k]$ with $i<j$;
\item \label{cond:valid-EU}  the nominee of $E_{\{i,j\}}^{(c)}$, where $c=0$ if $i<j$ and $c=1$ if $i>j$, precedes the nominee of $U_i^{(j)}$ in both blocks $H_{(i,j)}$ and $\wt{H}_{(i,j)}$, for each $i,j \in [k]$ with $i \neq j$.
\end{enumerate}
\end{definition}

\begin{claim}
\label{clm:validity=clique} 
There exists a valid nomination for all parties in~$\SS^\orig \cup \SS^\cop$ if and only if $G$ admits a clique of size~$k$.
\end{claim}

\begin{claimproof}
Suppose that there exists a valid nomination. Let $u_i^{\kappa(i)}$ denote the nominee of $U_i$ for each $i \in [k]$, 
and let $e_{\{i,j\}}^{\kappa(i,j)}$ denote the nominee of $E_{\{i,j\}}$ for each $i,j \in [k]$ with $i<j$.

Condition~\ref{cond:valid-U} of validity can be satisfied only if
for each $c \in [k] \setminus \{i\}$, the nominee~$u_i^{h,(c)}$ of $U_i^{(c)}$ satisfies, on the one hand, $h \geq \kappa(i)$ due to the definition of~$Y_i$ in~(\ref{eq:Yblock-def}) and, on the other hand, $h \leq \kappa(i)$ due to the definition of~$\wt{Y}_i$ in~(\ref{eq:wtYblock-def}).
Therefore, 
$U_i^{(c)}$ must nominate the candidate~$u_i^{\kappa(i),(c)}$ for each $c \in [k] \setminus \{i\}$ in~$\EE$. 
With an analogous argument based on the definitions of the blocks~$F_{\{i,j\}}$ and~$\wt{F}_{\{i,j\}}$ as given in~(\ref{eq:Fblock-def})--(\ref{eq:wtFblock-def}), condition~\ref{cond:valid-E} of validity 
implies that
$E_{\{i,j\}}^{(c)}$ must nominate the candidate~$e_{\{i,j\}}^{\kappa(i,j),(c)}$
in~$\EE$ for both $c \in \{0,1\}$, for each $i,j \in [k]$ with $i<j$. 

Consider now condition~\ref{cond:valid-EU} of validity for some 
indices $i,j \in [k]$ with $i<j$.
On the one hand, by the definition of the block~$H_{(i,j)}$ as given in~(\ref{eq:Hblock-def1}), we get that the edge~$e_{\{i,j\}}^{\kappa(i,j)}$ must be incident to a vertex in $\{u_i^h:1 \leq h \leq \kappa(i)\}$.
On the other hand, by the definition of~$\wt{H}_{(i,j)}$ as given in~(\ref{eq:wtHblock-def1}), we get that 
the edge~$e_{\{i,j\}}^{\kappa(i,j)}$ must be incident to a vertex in $\{u_i^h:\kappa(i) \leq h \leq n\}$. Therefore, $e_{\{i,j\}}^{\kappa(i,j)}$ must be incident to~$u_i^{\kappa(i)}$. 
Consider now the condition~\ref{cond:valid-EU} of validity for the index pair $(j,i)$; recall that $j>i$. The definition of the blocks~$H_{(j,i)}$ and~$\wt{H}_{(j,i)}$ as given in~(\ref{eq:Hblock-def2}) and~(\ref{eq:wtHblock-def2}) implies through an analogous reasoning that the edge~$e_{\{i,j\}}^{\kappa(i,j)}$ must also be incident to~$u_j^{\kappa(j)}$. 
Thus, the vertices $u_i^{\kappa(i)}$ and $u_j^{\kappa(j)}$ must be connected by an edge. Since this holds for each $i,j \in [k]$ with $i<j$, we obtain that $G$ indeed contains a clique of size~$k$.

For the other direction, it is straightforward to verify that given a clique~$K$ of size~$k$ in~$G$, nominating the vertices and edges of~$K$, as well as all their copies, yields a set of nominations for parties in~$\SS^\orig \cup \SS^\cop$ that is valid.
\end{claimproof}

Using Claim~\ref{clm:validity=clique}, it is now not hard to prove the correctness of the reduction. 
\begin{claim}
\label{clm:dir1} 
If $P$ is the unique winner in an election~$\EE$ resulting from some nominations, 
then $G$ contains a clique of size~$k$.
\end{claim}

\begin{claimproof}
Notice that for $P$ to become the unique winner in~$\EE$, all nominated candidates in~$\EE$ must have Copeland$^\alpha$-score less than~$\Cpl(p)$. 
Consider the party~$U_i^{(j)}$ for some~$i,j \in [k]$ with $i \neq j$.   
By $U_i^{(j)} \in \SS^\cop$ and statement~\ref{item-fact:copy-vs-D}, we know that $U_i^{(j)}$ defeats all candidates in~$D$; recall that each of these candidates forms a singleton party, yielding a partial score of~$|D|$ for~$U_i^{(j)}$. 
Observing the preference lists, we can see that 
the only parties that may defeat $U_i^{(j)}$ are $U_i$ and 
either $E_{\{i,j\}}^{(c)}$  where $c=0$ if $i<j$, and $c=1$ if $i>j$; 
all remaining $t-|D|-3$ parties are tied with~$U_i^{(j)}$ in~$\EE$.
Therefore, by (\ref{eq:p-score}) we get 
\[\Cpl(U_i^{(j)}) \geq |D|+\alpha(t-|D|-3
) = \Cpl(p).\]
Moreover, 
unless both of the two parties mentioned above---namely, $U_i$ and $E_{\{i,j\}}^{(c)}$ with $c$ defined as~$0$ if $i<j$ and as~$1$ otherwise---defeat $U_i^{(j)}$, we know that the nominee of~$U_i^{(j)}$ obtains an additional score of at least $\alpha$ from the comparison with these two parties. That would result in $U_i^{(j)}$ having Copeland$^\alpha$-score at least $|D|+\alpha(t-|D|-2)$, preventing $p$ from becoming the unique winner of~$\EE$, due to~(\ref{eq:p-score}). 
Therefore, both of these two parties must defeat~$U_i^{(j)}$, which means precisely that conditions~\ref{cond:valid-U} and~\ref{cond:valid-EU} must hold.

Consider now the party~$E_{\{i,j\}}^{(c)}$ for some~$i,j \in [k]$ with $i<j$ and $c \in \{0,1\}$. 
By $E_{\{i,j\}}^{(c)} \in \SS^\cop$  and statement~\ref{item-fact:copy-vs-D}, we know that $E_{\{i,j\}}^{(c)}$ defeats all candidates in~$D$. 
Observing the preference lists again, we can see that 
the only parties that may defeat $E_{\{i,j\}}^{(c)}$ are $E_{\{i,j\}}$ and~$A_{\{i,j\}}$. 
Therefore, by~(\ref{eq:p-score}) we get 
\[\Cpl(E_{\{i,j\}}^{(c)}) \geq |D|+\alpha(t-|D|-3) = \Cpl(p).\]
Again, 
$E_{\{i,j\}}^{(c)}$ additionally obtains at least~$\alpha$ score from the comparison with parties~$E_{\{i,j\}}$ and~$A_{\{i,j\}}$, thereby preventing $p$ from becoming the unique winner in~$\EE$,
unless both~$E_{\{i,j\}}$ and~$A_{\{i,j\}}$ defeat~$E_{\{i,j\}}^{(c)}$. 
Hence, condition~\ref{cond:valid-E} of validity most also hold. 

This proves that the nominations that result in the election~$\EE$ must be valid, and thus  by Claim~\ref{clm:validity=clique} we know that $G$ admits a clique of size~$k$.
\end{claimproof}

\begin{claim}
\label{clm:dir2} 
If $G$ contains a clique~$K$ of size~$k$, then $P$ is a winner in an election~$\EE$ resulting from some nominations.
\end{claim}

\begin{claimproof}
For each $i \in [k]$, define $u_i^{\kappa(i)}$ as the vertex of~$K$ in~$U_i$, 
and let $e_{\{i,j\}}^{\kappa(i,j)}$ denote the unique edge in $E_{\{i,j\}}$ connecting two vertices of~$K$. 
Let $\EE$ be the election resulting from the following nominations: 
\begin{itemize}
\item party $U_i$ nominates~$u_i^{\kappa(i)}$ for each $i \in [k]$;
\item party $U_i^{(j)}$ nominates $u_i^{\kappa(i),(j)}$ for each $i,j \in [k]$ with $i \neq j$;
\item party $E_{\{i,j\}}$ nominates $e_{\{i,j\}}^{\kappa(i,j)}$
for each $i,j \in [k]$ with $i < j$;
\item  party $E_{\{i,j\}}^{(c)}$ nominates $e_{\{i,j\}}^{\kappa(i,j),(c)}$ for  each $i,j \in [k]$ with $i < j$ and $c \in \{0,1\}$;
\item party $A_{\{i,j\}}$ nominates $a_{\{i,j\}}^{\kappa(i,j)}$ for each $i,j \in [k]$ with $i < j$.
\end{itemize}

It is straightforward to verify that the above nominations are valid.
Taking into account statements~\ref{item-fact:copy-vs-D}--\ref{item-fact:DpCpBp}, it follows that each party except for~$P$ is defeated by at least two other parties. Note also that each party defeats at most~$|D|$ parties. Hence, if some party other than~$P$ defeats $k \leq |D|$ other parties in~$\EE$, then it must have Copeland$^\alpha$-score at most~$k+\alpha(t-k-3)$, which is less than~$\Cpl(p)$ due to~(\ref{eq:p-score}) and $0<\alpha\leq 1$. Therefore, each party other than~$P$ has score less than~$\Cpl(p)$, so $P$ is indeed the unique winner of the election~$\EE$.\end{claimproof}

The correctness of the reduction now follows from Claims~\ref{clm:dir1} and~\ref{clm:dir2}.
\end{proof}

\subsection{Proof of Theorem~\ref{thm:5voters-t}}
\label{sec:proof-of-5voterst}

Before proving Theorem~\ref{thm:5voters-t}, we state a lemma which will be useful for us.
This lemma creates the possibility to construct nearly flat elections, where we can increase the score of certain candidates slightly by shifting them ``to the left'' (i.e., closer to the top) in the preferences of one voter. 

\begin{lemma}
\label{lem:shiftable-candidates}
    If $q \geq \lceil \log_3 \ell \rceil +
    \lceil \log_3 \ell' \rceil +1$, 
    then $C_q$ contains $\ell$ candidates $c_1,\dots, c_\ell$ and $\ell$  candidate sets $C_{\leftarrow 1}, \dots,C_{\leftarrow \ell}$,  each of size~$\ell'$, such that: 
    \begin{itemize}
        \item for each $i \in [\ell]$, candidate~$c_i$ is directly preceded by the candidates of~$C_{\leftarrow i}$ in the preferences of~$w$;
        \item for each $i \in [\ell]$, each candidate in~$C_{\leftarrow i}$ defeats~$c_i$ in~$\EE_q$;
        \item the sets~$\{c_i\} \cup C_{\leftarrow i}$ for $i \in [\ell]$ are pairwise disjoint.
    \end{itemize}
\end{lemma}

\begin{proof}
Set $q'=\lceil \log_3 \ell' \rceil $, and consider the election $\EE_{q'+1}$, defined by (\ref{prefs-of-EEq}) with $q'$ taking the role of~$q$. Let $c^\star$ denote the first candidate in $\copb{L_{q'}(w)}$, and notice that $c^\star$ is preceded by $3^{q'} \geq \ell'$ candidates in the preferences of~$w$, each of whom defeats $c^\star$ in~$\EE_{q'+1}$; let $C^\star$ denote this candidate set. 

Recall that each $(q'+1)$-level group within~$\EE_q$ is a copy of $\EE_{q'+1}$, and hence contains a copy of~$c^\star$ and $C^\star$. 
By $3^{q-(q'+1)} \geq 3^{\lceil \log_3 \ell \rceil} \geq \ell$, there are at least $\ell$ such $(q'+1)$-level groups in~$\EE_q$. Fixing an arbitrary ordering over these groups, let $c_i$ and $C_{\leftarrow i}$ be the copy of~$c^\star$ and~$C^\star$, respectively, in the $i$-th such group for $i=1,\dots,\ell$. It is then clear that 
the sets $\{c_i\} \cup C_{\leftarrow i}$ are pairwise disjoint, and that each $c_i$ is defeated by all candidates of~$C_{\leftarrow i}$. 
Notice also that candidates belonging to the same group (at some level) appear consecutively in each of the preference lists in~$\EE_q$, and therefore $c_i$ is directly preceded by the candidates of~$C_{\leftarrow i}$ in the preference list of~$w$ in~$\EE_q$.
\end{proof}

\thmfivevoterst*
\begin{proof}

Let $G=(U,E)$ and $k$ be our instance of the \textsc{Multicolored Clique} with $U$ partitioned into sets~$U_1, \dots, U_k$;
again, we use all the notation introduced at the beginning of Appendix~\ref{app:param-t}.
We define an instance of \PP{} as follows. 

\medskip
\noindent
{\bf Candidates and parties.}
We set~${P=\{p\}}$ as our distinguished party. We also add~$U_i$ for each $i \in [k]$, as well as $E_{\{i,j\}}$, $A_{(i,j)}$, and~$A_{(j,i)}$ for each $i,j \in [k]$ with $i<j$ as a party where $A_{(i,j)}$ and~$A_{(j,i)}$ are two copies of the edge set~$E_{\{i,j\}}$ defined as $A_{(i,j)}=\{a_{(i,j)}^1,\dots,a_{(i,j)}^m\}$ for each $i,j \in [k]$ for $i \neq j$. 
We will also add a set~$D$ of dummy candidates, each of whom forms its own singleton party.
We set the total number of dummies as 
\begin{align}
    \label{def:5v-dummies-number}
    |D| &=3^{q+1}-3\binom{k}{2}-k-1 \quad \text{where } \\
\label{eqn:5v-qdef}
q& =\lceil \log_3 (k+1) \rceil +
\lceil \log_3 k \rceil +2; 
\end{align}
observe that $|D|=O(k^2)$.
It will be useful for us to set aside a set~$\wt{D} \subseteq D$ of dummy candidates for indexing purposes, and to associate with each $d \in \wt{D}$ two copies in~$D \setminus \wt{D}$, denoted by $d'$ and~$d''$, such that $|\{d,d',d'':d \in \wt{D}\}|=3|\wt{D}|$. 
We will also use the notation $D_0=D \setminus \{d,d',d'':d \in \wt{D}\}$.
We set the sizes of these sets such that  
$|D_0|=2k+2$ and thus by (\ref{def:5v-dummies-number}) we have 
\begin{equation}
\label{eqn:5v-no-of-dummies}
    |\wt{D}|=\frac{|D|-2k-2}{3}=3^q-\binom{k}{2}-k-1.
\end{equation}

The number of parties is $t=1+k+3\binom{k}{2}+|D|=O(k^2)$.

\medskip
\noindent
{\bf High-level description.}
We will use the techniques developed in Sections~\ref{sec:flat-elections} and~\ref{sec:3votersreduction} for proving Theorem~\ref{thm:3voters-NPh}. 
First, we create a flat 
election~$\EE_q$ 
with three voters, namely $v,v',$ and $v''$, with 
sufficiently $\vphantom{L^{L^L}}$ large number of candidates so that we can associate a candidate in~$\EE_q$ with each of the parties~$P$ and $U_i$ for~$i \in [k]$, as well as with each edge set~$E_{\{i,j\}}$ for $1 \leq i<j \leq k$. We  then substitute each candidate of~$\EE_q$ in the preference list of a voter in~$\{v,v',v''\}$ with corresponding \emph{team lists}, as done in the proof of Theorem~\ref{thm:3voters-NPh}.
Relying on Lemma~\ref{lem:shiftable-candidates}, we will then shift all candidates in the party~$U_i$ for $i \in [k]$, so that each of these parties  obtains $k$ more points in a Copeland$^\alpha$ election than the rest of the candidates except for certain dummy candidates, who lose one point each. Similarly, we shift candidate~$p$ so that it gains~$2$ additional points. 

Next, we introduce voters $z$ and~$z'$ whose preferences are \emph{almost} exactly opposed: they rank each pair of candidates differently, with the exception that a candidate~$a_{(i,j)}^h$ of~$A_{(i,j)}$ may defeat candidate~$u \in U_i$ in the preferences of both~$z$ and~$z'$ whenever the edge~$e_{\{i,j\}}^h$ is incident to~$u$. 
Hence, by choosing nominees for the parties~$A_{(i,j)}$ for $i,j \in [k]$ with $i \neq j$ so that they correspond to edges that are incident to the nominees of the parties~$U_i$, $i \in [k]$, we can decrease the score of each $U_i$ by $k-1$ points at the cost of increasing the score of each party~$A_{i,j}$ by~$1$, enabling~$p$ to become the unique winner. With a minor adjustment of the scores of the parties $E_{\{i,j\}}$ for $i,j \in [k]$ with $ i \neq j$, we can also ensure that the team lists associated with the edge set~$E_{\{i,j\}}$ must induce a flat election, which in turn guarantees that the nominations in~$A_{(i,j)}$ and~$A_{(j,i)}$ are consistent in the sense that they correspond to the same edge of~$E_{\{i,j\}}$ for each $i,j \in [k]$, implying that the nominees of the parties~$U_i$ for~$i \in [k]$ must form a clique in~$G$.

\medskip
\noindent
{\bf Teams and their associated lists.}
Let us define the following families of parties: 
\begin{align*}
\mathcal{P}_U &= \{U_i:i \in [k]\}, \\
\mathcal{P}_E &= \{E_{\{i,j\}}:1 \leq i<j\leq k\}, \\
\mathcal{P}_{\wt{D}} &=\{ \{d\}:d \in \wt{D}\}.
\end{align*}
We will refer to the family 
\[T=\{P\} \cup \mathcal{P}_U \cup \mathcal{P}_E 
\cup  \mathcal{P}_{\wt{D}}
\] 
of parties as the set of our \emph{teams}. Notice that 
by (\ref{eqn:5v-no-of-dummies}) we know $|T|=3^q$.
To define the preferences of  voters $v$, $v'$, and~$v''$, we introduce \emph{team lists} $F_t,F'_t$, and~$F''_t$ for each team $t \in T$.

With the distinguished party~$P=\{p\}$, we associate the following team lists, where $d_p$ and~$d'_p$ are two dummies in~$D_0$:
    \begin{equation}
    \begin{array}{@{}r@{\hspace{2pt}}l@{}}    
    F_P &= p,d_p,d'_p; \\
    F'_P &= d'_p,p,d_p; \\    
    F''_P &= d_p,d'_p,p.
    \end{array}
    \label{eqn:def-Fp-5v}
    \end{equation}
Similarly, with each party~$U_i$ for $i \in [k]$, 
we associate the following team lists, where $d_i$ and~$d'_i$ are two dummies in~$D_0$:
    \begin{equation}
    \begin{array}{@{}r@{\hspace{2pt}}l@{}}    
    F_{U_i} &= \ora{U_i},d_i,d'_i; \\
    F'_{U_i} &= d'_i,\ora{U_i},d_i; \\    
    F''_{U_i} &= d_i,d'_i,\ora{U_i}.
    \end{array}
    \label{eqn:def-FUi-5v}
    \end{equation}
When defining the above team lists, we always use distinct dummies, so that the team lists in~(\ref{eqn:def-Fp-5v})--(\ref{eqn:def-FUi-5v}) contain a total of $2k+2 =|{\{d_p,d'_p\}} \cup {\{d_i,d'_i:i \in [k]\}}|$ dummy candidates, using up all of~$D_0$.

Next, we define a team list for each party~$E_{\{i,j\}}$ where $i,j \in [k]$ with $i<j$:
    \begin{equation}
    \begin{array}{@{}r@{\hspace{2pt}}l@{}}    
    F_{E_{\{i,j\}}} &= 
    a_{(i,j)}^1, e_{\{i,j\}}^1, a_{(j,i)}^1, \dots, 
    a_{(i,j)}^m, e_{\{i,j\}}^m, a_{(j,i)}^m; \\[4pt]
    F'_{E_{\{i,j\}}} &= 
    a_{(j,i)}^1, a_{(i,j)}^1, e_{\{i,j\}}^1, \dots, 
    a_{(j,i)}^m, a_{(i,j)}^m, e_{\{i,j\}}^m; \\[4pt]
    F''_{E_{\{i,j\}}} &= 
    e_{\{i,j\}}^1, a_{(j,i)}^1, a_{(i,j)}^1, \dots, 
    e_{\{i,j\}}^m, a_{(j,i)}^m, a_{(i,j)}^m.
    \end{array}
    \label{eqn:def-FEij-5v}
    \end{equation}

Finally, we define the team list for each $d \in \wt{D}$ as 
    \begin{equation}
    \begin{array}{@{}r@{\hspace{2pt}}l@{}}    
    F_{\{d\}} &= d,d',d''; \\[2pt]
    F'_{\{d\}} &= d'',d,d'; \\[2pt]    
    F''_{\{d\}} &= d',d'',d.
    \end{array}
    \label{eqn:def-Fd-5v}
    \end{equation}

\medskip
\noindent
{\bf Preferences.} To define the preferences of our five voters, we create an election~$\EE^{G,k}$ in three steps. 
First, based on Definition~\ref{def:EEq} we create an election~$\EE^\varphi$ that is obtained from the flat election~$\EE_q$ by using a bijection~$\varphi$ between our team set~$T$ and the candidate set~$C_q$ of~$\EE_q$, and substituting each candidate of~$C_q$ with the team list of the corresponding team. 
Second, we shift certain candidates in~$\EE^\varphi$ so that $p$ gains $2$ points but parties in~$\P_U$ gain $k$ points each; we denote by $\EE^{\varphi,\textup{sh}}$ the resulting election. 
Third, we add two additional voters together with their preferences to obtain~$\EE^{G,k}$.
Let us explain these three steps more in detail. 

\smallskip 
\noindent
{\bf Step 1: Constructing the election~$\EE^\varphi$.}
Consider the election~$\EE_q$ defined over candidates set~$C_q$ as in Definition~\ref{def:EEq}. Recall that the number of candidates in~$\EE_q$ is $|C_q|=3^q$.

Due to our choice of~$q$ as defined in~(\ref{eqn:5v-qdef}), we can apply Lemma~\ref{lem:shiftable-candidates} 
to~$\EE_{q-1}$ for $\ell=k+1$ and~$\ell'=k$. We obtain that there are candidate sets~$\{c_i\} \cup C_{\leftarrow i}$ in~$\EE_{q-1}$ for $i\in [k+1]$ such that 
\begin{itemize}
    \item candidate~$c_i$ is directly preceded by the candidates of~$C_{\leftarrow i}$ in the preferences of~$w$;
    \item for each $i \in [\ell]$, each candidate in~$C_{\leftarrow i}$ defeats~$c_i$ in~$\EE_{q-1}$;
    \item the sets~$\{c_i\} \cup C_{\leftarrow i}$ for $i \in [k+1]$ are pairwise disjoint and have size~$k+1$.
\end{itemize} 
We define a bijection~$\varphi$ from~$T$ into the set $C_q$ of candidates in~$\EE_q$ such that $\varphi(U_i)=\copa{c_i}$ for each $i \in [k]$, 
$\varphi(P)=\copa{c_{i+1}}$, 
and moreover,
teams in~$\P_E$ are mapped to candidates in~$\{\copb{c}:c \in C_{q-1}\}$.
Notice that such a function exists and can be easily computed based on~$\EE_{q-1}$, because 
$|\mathcal{P}_{E}|=\binom{k}{2} \leq 3^{q-1}$ holds 
due to our choice of~$q$, as defined in~(\ref{eqn:5v-qdef}).   

Now, we create the election~$\EE^\varphi$ as follows. For each candidate in~$c \in C_q$, we substitute $c$ in the preference list of~$v,v',$ and~$v''$ with the team list~$F_{\varphi^{-1}(c)}$, $F'_{\varphi^{-1}(c)}$, and $F''_{\varphi^{-1}(c)}$, respectively.

Before proceeding to the next step, let us establish the following claim.
\begin{claim}
\label{clm:quasiflatness-of-EEvarphi}
Let $N$ be a set of nominees containing exactly one candidate from each party in the constructed instance. 
Then  in the election~$\EE^\varphi_N$
\begin{itemize}
    \item each nominee $c \in N$ defeats exactly $3\cdot (3^q-1)/2$ nominees among those not contained in the same team lists as~$c$;
    \item $\Cplodd[\EE^{\varphi}_{N}](P)=
\Cplodd[\EE^{\varphi}_{N}](U_i)=
\Cplodd[\EE^{\varphi}_{N}](d)=
\frac{3^{q+1}-1}{2}
$ for each $U_i \in \P_U$ and each dummy~$d \in D$;
\item candidate~$p$, as well as the nominee of~$U_i$ for each $i \in [k]$, is defeated by each of the $3k$ nominees directly preceding it in the preference list of~$v$ in~$\EE^{\varphi}_N$;
    \item the nominee of $U_i$ defeats the nominees of~$A_{(i,j)}$ and of~$E_{\{i,j\}}$ for each $i,j \in [k]$ with~$i \neq j$.
\end{itemize}
\end{claim}
\begin{claimproof}
    First notice that for each party, there is exactly one team list containing all of its candidates. In particular, this means that for any nomination by the parties, each team list contains exactly three nominees. Using that $\varphi$ is a bijection, the first statement of the claim follows from Observation~\ref{obs:EEq-number-of-defeats}.

    The second statement follows from the observation that the team lists containing candidates of~$P$, of~$\bigcup_{i \in [k]}U_i$ or of~$D$ are all cyclic in the sense that they can be interpreted as copies of the election~$\EE_1$; see their definitions in~(\ref{eqn:def-Fp-5v}), (\ref{eqn:def-FUi-5v}), and~(\ref{eqn:def-Fd-5v}). Hence, these candidates obtain a score of~$3(3^q-1)/2+1$ in~$\EE^\varphi_N$.
    
    To see the the third statement, note that $p$ as well as the nominee of each party in~$\P_U$ is contained in the team list~$F_{\varphi^{-1}(c_i)}$ for some candidate~$c_i \in C_{q-1}$, $i \in [k+1]$, resulting from Lemma~\ref{lem:shiftable-candidates}, and in particular is listed first among the nominees in~$F_{\varphi^{-1}(c_i)}$ in the preferences of~$v$. By our conditions on~$c_i$ and~$C_{\leftarrow i}$, all $3k$ nominees in the team lists associated with the $k$ teams that are mapped to candidates of~$C_{\leftarrow i}$ by~$\varphi$ defeat all nominees in~$F_{\varphi^{-1}(c_i)}$ in~$\EE^\varphi_N$. Since these $3k$ nominees are exactly those that in~$\EE^\varphi_N$ directly precede the first nominee within~$F_{\varphi^{-1}(c_i)}$ in the preferences of~$v$, the statement follows.

    Finally, the fourth statement follows from the property of the election~$\EE_q$ that each candidate in~$\{\copa{c} :c \in C_{q-1}\}$ defeats every candidate in~$\{\copb{c} :c \in C_{q-1}\}$; see Definition~\ref{def:EEq}. Note that we also use our condition on~$\varphi$ that it maps 
    parties in~$\P_U$ to candidates in~$\{\copa{c} :c \in C_{q-1}\}$ and maps
    parties in~$\P_E$ to candidates in~$\{\copb{c} :c \in C_{q-1}\}$, and hence, all nominees in the team lists of some team in~$\P_U$ defeat
    all nominees in the team lists of some team in~$\P_E$.
\end{claimproof}

\smallskip 
\noindent
{\bf Step 2: Constructing the election~$\EE^{\varphi,\textup{sh}}$.}
We modify~$\EE^\varphi$ by shifting every candidate in~$U_i$ to the left (i.e., towards the most-preferred position) by~$k$ positions for each $i \in [k]$, and shifting the candidate~$p$ to the left by~$2$ positions in the preferences of~$v$. 

Recall now that due to the third statement of Claim~\ref{clm:quasiflatness-of-EEvarphi},
candidate~$p$, as well as the first candidate in~$\ora{U_i}$ for each $i \in [k]$, is preceded by~$3k$ dummy candidates in the preferences of~$v$ in $\EE^\varphi$, each of whom defeats it.
Therefore, the shifts applied to the election~$\EE^\varphi$ increase the score of~$p$ by exactly~$2$, increase the score of each~$U_i$, $i \in [k]$, by exactly~$k$, and decrease the score of certain dummy candidates by~$1$.
Consequently,  in 
any reduced election~$\EE^{\varphi,\textup{sh}}_{N}$
we get that
\begin{align}
\label{eq:5v-score-p-aftershift}
\Cplodd[\EE^{\varphi,\textup{sh}}_{N}](P) &=
\frac{3^{q+1}-1}{2}+2 \quad \textrm{and}
\\ 
\label{eq:5v-score-Ui-aftershift}
\Cplodd[\EE^{\varphi,\textup{sh}}_{N}](U_i) &=
\frac{3^{q+1}-1}{2}+k.
\end{align}

Notice that for each candidate~$e \in E$, the applied shifts do not change the relative position of~$e$ with respect to any other candidate.

\smallskip 
\noindent
{\bf Step 3: Constructing the election~$\EE^{G,k}$.}
We are going to add two more voters, $z$ and~$z'$, to the election~$\EE^{\varphi,\textup{sh}}$, resulting in the election~$\EE^{G,k}$. To define the preferences of these two additional voters, we need to define a set $A(u)$ for each vertex $u$ as \[A(u)=\left\{a^h_{(i,j)}: e^h_{\{i,j\}} \text{ is incident to $u$, and $i$ is such that } u \in U_i\right\},\] 
containing one copy of each edge incident to~$u$.  
Then we can define the blocks~$Y_i$ and $\wt{Y}_i$ for each $i \in [k]$ as 
\begin{align}
\label{eq:blocks-Y-zz'}
\begin{split}
    Y_i &: \ora{A(u_i^1)},u_i^1, \dots, \ora{A(u_i^n)},u_i^n; \\
    \wt{Y}_i &:  \ora{A(u_i^n)},u_i^n, \dots, \ora{A(u_i^1)},u_i^1.
\end{split}
\end{align}

We also need to fix an arbitrary dummy~$d^\star$ that is contained in the team list of some team mapped to a candidate in~$\{\copa{c}:c \in C_{q-1}\}$ by~$\varphi$; since all such teams contain a dummy in their lists (recall that teams in~$\P_E$ are all mapped to~$\{\copb{c}:c \in C_{q-1}\}$ by~$\varphi$), such a dummy exists. Notice that $d^\star$ defeats all candidates of each party in~$\P_E$ in the election~$\EE^{\varphi}$, and hence, also in the election~$\EE^{\varphi,\textup{sh}}$.

Finally, we set the preferences of~$z$ and~$z'$ as follows: 
\begin{align}
\label{eq:prefs-zz'}
\begin{split}
z&: p,\ora{E}, d^\star, \ora{D \setminus \{d^\star\}},Y_1,Y_2,\dots,Y_k; \\
z'&: \wt{Y}_k,\wt{Y}_{k-1},\dots,\wt{Y}_1,
\ola{D \setminus \{d^\star\}},
\ola{E},d^\star,p.
\end{split}
\end{align}

The following claim is immediate from the preferences of voters~$z$ and~$z'$, our choice of~$d^\star$, and (\ref{eq:5v-score-p-aftershift}).
\begin{claim}
    \label{clm:5v-Eij-vs-dummystar}
    Let $N$ be a set of nominees containing exactly one candidate from each party in the constructed instance. 
    Then the nominee of each party~$E_{\{i,j\}} \in \P_E$  
is defeated by~$d^\star$ in~$\EE^{\varphi,\textup{sh}}_N$  but defeats~$d^\star$ in~$\EE^{G,k}:N$; moreover,
\begin{equation}
    \label{eq:5v-score-p-final}
    \Cplodd[\EE^{G,k}_{N}](P) =
    \Cplodd[\EE^{\varphi,\textup{sh}}_{N}](P) =
    \frac{3^{q+1}-1}{2}+2
\end{equation}
\end{claim}

\medskip
\noindent
{\bf Correctness.}
We claim that there exists a set~$N$ of nominees (containing exactly one candidate from each party) such that $p$ is the unique winner in the reduced election~$\EE^{G,k}_N$ if and only if $G$ contains a clique of size~$k$.

\medskip
\noindent
{\bf Direction ``$\Rightarrow$'':}
Suppose that $p$ is the unique winner in~$\EE^{G,k}_N$, that is, each nominee in~$N \setminus \{p\}$ obtains at most $\frac{3^{q+1}-1}{2}+1$ points. 
Let $u_i^{\kappa(i)}$ denote the nominee of~$U_i$ for each $i \in [k]$. We are going to show that the vertex set~$K=\{u_i^{\kappa(i)}:i \in [k]\}$ induces a clique in~$G$.

Fix some~$i \in [k]$.
By (\ref{eq:5v-score-Ui-aftershift}) and (\ref{eq:5v-score-p-final}) we know that there must exist a set~$N_i^\star$ of  at least $k-1$ nominees that are defeated by~$U_i$ in~$\EE^{\varphi,\textup{sh}}_N$ but defeat~$U_i$ in~$\EE^{G,k}_N$. Clearly, each nominee in~$N_i^\star$ must precede the nominee of~$U_i$, that is, $u_i^{\kappa(i)}$ in the preferences of both~$z$ and~$z'$. 
Observing the preferences of these voters as given in~(\ref{eq:prefs-zz'}) and the definition of the blocks $Y_i$ and~$\wt{Y}_i$ in~(\ref{eq:blocks-Y-zz'}),
we deduce that $N_i^\star \subseteq A(u_i^{\kappa(i)})$ for each $i \in [k]$, and hence, by $|N_i^\star|\geq k-1$ we get that $A(u_i^{\kappa(i)})$ must contain the nominee of each of the $k-1$ parties $A_{(i,j)}$ for $j \in [k] \setminus \{i\}$.
In other words, the edges corresponding to the nominees of the parties~$A_{(i,j)}$, $j \in [k] \setminus \{i\}$, must all be incident to~$u_i^{\kappa(i)}$.

Consider now a party~$A_{(i,j)}$ for some $i,j \in [k]$ with $i \neq j$.
The above arguments imply that the nominee of~$A_{(i,j)}$ defeats  at least $\frac{3 \cdot (3^q-1)}{2}+1$ candidates not contained in the team list $F_{E_{\{i,j\}}}$: the $3 \cdot (3^q-1)/2$ nominees it already defeated in~$\EE^\varphi_N$ (due to the first statement of Claim~\ref{clm:quasiflatness-of-EEvarphi}), plus the nominee $u_i^{\kappa(i)}$ that it defeats in~$\EE^{G,k}_N$ due to voters~$z$ and~$z'$ but not in~$\EE^\varphi_N$ (as shown also by the third statement of Claim~\ref{clm:quasiflatness-of-EEvarphi}).

Similarly, the nominee of~$E_{\{i,j\}}$ defeats at least $\frac{3 \cdot (3^q-1)}{2}+1$ candidates not contained in the team list $F_{E_{\{i,j\}}}$: the $3 \cdot (3^q-1)/2$ nominees it already defeated in~$\EE^\varphi_N$ (due to the first statement of Claim~\ref{clm:quasiflatness-of-EEvarphi}) and hence also defeated in~$\EE^{\varphi,\textup{sh}}_N$ (as the shifts applied in Step~2 do not affect~$E_{\{i,j\}}$), plus the dummy $d^\star$ that it defeats in~$\EE^{G,k}_N$ due to voters~$z$ and~$z'$ but not in~$\EE^{\varphi,\textup{sh}}_N$ (due to Claim~\ref{clm:5v-Eij-vs-dummystar}).

This implies that the restriction of~$\EE^{G,k}$ to the three nominees in the team list~$F_{E_{\{i,j\}}}$ must be flat, as otherwise one of the candidates nominated by parties~$A_{(i,j)}, A_{(j,i)}$, and~$E_{\{i,j\}}$ would gain two additional points (against the nominees within in the team list~$F_{E_{\{i,j\}}}$) and thus would obtain a score of~$\frac{3 \cdot (3^q-1)}{2}+3=\frac{3^{q+1}-1}{2}+2$, which by (\ref{eq:5v-score-p-final}) contradicts the assumption that $p$ is the unique winner in~$\EE^{G,k}_N$.

Considering the team lists associated with $E_{\{i,j\}}$ as defined in~(\ref{eqn:def-FEij-5v}), it follows in a straightforward way that the restriction of~$\EE^{G,k}$ to the three nominees in the team list~$F_{E_{\{i,j\}}}$ can only be flat if $A_{(i,j)}$ and $A_{(j,i)}$ both nominate the copies of the same edge, i.e., their nominees are~$a_{(i,j)}^h$ and~$a^h_{(j,i)}$, respectively, for some $h \in [m]$. Therefore, $a_{(i,j)}^h \in  A(u_i^{\kappa(i)})$
and $a_{(j,i)}^h \in  A(u_j^{\kappa(j)})$ must hold, in other words, $e_{\{i,j\}}^h$ is incident to both~$u_i^{\kappa(i)}$ and~$u_j^{\kappa(j)}$. Since we have proved this for arbitrary $i,j \in [k]$ with $i \neq j$, it follows that $K$ forms a clique in~$G$.

\medskip
\noindent
{\bf Direction ``$\Rightarrow$'':}
Suppose now that 
$K=\{u_i^{\kappa(i)}:i \in [k]\}$ induces a clique in~$G$.
Consider the election obtained by nominating each vertex of~$K$, as well as each edge~$e_{\{i,j\}}^h$ in~$G[K]$ together with its copies~$a_{(i,j)}^h$ and~$a_{(j,i)}^h$,
besides nominating~$p$ and all dummies. Let $N$ denote the set of nominees.

Notice that restricting~$\EE^\varphi$ to the nominees contained in some arbitrarily chosen team list yields a flat election. Hence, it is immediate by Observation~\ref{obs:EEq-number-of-defeats} that $\EE^\varphi_N$ is a flat election.
The shifts applied in Step~2 to obtain~$\EE^{\varphi,\textup{sh}}$ increase the score of~$p$ by~$2$, increase the score of each~$U_i \in \P_U$ by~$k$, and do not change the score of~$E_{\{i,j\}}$, irrespective of the nominations; dummies' score may decrease but not increase.
Finally, since the nominee of~$A_{(i,j)}$ for some $i,j \in [k]$ with $i \neq j$ precedes the candidate of~$U_i$ in the preferences of both~$z$ and~$z'$ (since it is a copy of an edge incident to the vertex in~$K \cap U_i$),
we get that the addition of voters~$z$ and~$z'$
increases the score of~$A_{(i,j)}$ by~$1$ for each $i,j \in [k]$ with $i \neq j$, and decreases the score of~$U_i$ by~$k-1$ for each~$i \in [k]$. Taking into account Claim~\ref{clm:5v-Eij-vs-dummystar} and the observation that no dummy's score may increase due to the addition of voters~$z$ and~$z'$, we get that 
each party in~$\P_U \cup \P_E$ as well as each party~$A_{(i,j)}$, $i,j \in [k]$ with $i \neq j$, obtains 
a score of~$\frac{3^{q+1}-1}{2}+1$ in~$\EE^{G,k}_N$. Hence, $P$ is indeed a unique winner  in~$\EE^{G,k}_N$ by (\ref{eq:5v-score-p-final}), finishing the proof of correctness for our reduction.
\end{proof}

\section{Additional material for Section~\ref{sec:Maximin}}

\subsection{Proof of Theorem~\ref{thm:Maximin_hard}}
\label{sec:proof-of-Maximinhard}

\thmMaximinhard*

We first prove the statement for four voters in Theorem~\ref{thm:Maximin_V4_s2}, and then for five voters in Theorem~\ref{thm:Maximin_V5_s2}; as it is possible to add two voters with opposite preferences without changing the election outcome, these results imply Theorem~\ref{thm:Maximin_hard}.

\begin{theorem}
    \label{thm:Maximin_V4_s2}
    \PP\ for Maximin is $\NP$-complete even for instances with four voters and maximum party size~$\maxsize=2$.
\end{theorem}

\begin{proof}
    We present a reduction from $\NP$-complete \textsc{3-SAT} problem. 
    Let the input instance for \textsc{3-SAT} be a formula $\varphi=\bigwedge_{k=1}^m c^k$ over a set~$X=\{x_1,\dots,x_n\}$ of variables that is the conjunction of a set 
    $\C=\{c^k: k \in [m]\}$  of clauses, where each clause $c \in \C$ contains exactly three literals. 
    We may assume w.l.o.g.\ that $\varphi$ is not trivial in the sense that there is at least one variable that appears both as a positive and a negative literal in~$\varphi$.

    We are going to define an instance of \PP\ for Maximin as follows. We let $V=\{v,v',w,w'\}$ be the set of voters. 
    We let $P=\{p\}$ be our distinguished party, and we define another party $P'=\{p'\}$.
    For each variable~$x \in X$ we further introduce a party $P_x=\{x,\ol{x}\}$, and for each clause~$c \in \C$
    we introduce three parties $P_{c,j}=\{c_j,c^-_j\}$, $j \in [3]$. 
    Notice that each party has size at most two.    

    To provide the preferences of the voters, we introduce some notation. Let 
\begin{equation}
    \label{eq:def-c-minus}
    \C^-=\{c_j^-: c \in \C, j \in [3]\}
    \end{equation}
        and $L=\{x,\ol{x}:x \in X\}$ be the set of all literals.
    For each literal $\ell \in L$ we define 
    \begin{equation}
    \label{eq:def-v4s2-setA}
    A(\ell)=\{c_j: c \in \C, j \in [3], \ell \textrm{ is the $j$-th literal in }c \}.  \end{equation}
    We also define the following blocks for each~$c \in \C$ and each~$\ell \in L$: 
    \begin{align}
    \label{eq:def-v4s2-blocks}
    \begin{split}
        Y_c &= c_1, c_2, c_3, c^-_1, c^-_2, c^-_3; \\
        \wt{Y}_c &= c_3, c_2, c_1, c^-_3, c^-_2, c^-_1; \\
        F_{\ell} &= \ell, \ora{A(\ell)}; \\
        \wt{F}_{\ell} &= \ell, \ola{A(\ell)}. 
    \end{split}
    \end{align}
Now we are ready to define the preferences of our voters:
    \begin{align*}
        v &: p, p', \ora{L}, Y_{c^1}, \dots, Y_{c^m}; \\
        v' &: p, \wt{Y}_{c^m}, \dots, \wt{Y}_{c^1}, p', \ola{L}; \\
        w &: \ora{\C^-}, p', 
        F_{x_1}, \dots, F_{x_n}, 
        F_{\ol{x}_1}, \dots, F_{\ol{x}_n}, p; \\
        w' &:         
        F_{\ol{x}_n}, \dots, F_{\ol{x}_1},
        F_{x_n}, \dots, F_{x_1}, 
        \ola{\C^-}, p', p.
    \end{align*}

    We claim that $P$ can be the unique winner in 
    some reduced election~$\EE$
    if and only if $\varphi$ admits a satisfying truth assignment. 
    
    Note first that candidate~$p$ has Maximin-score $\MM_\EE(p)=2$ because it is tied with every other nominee in~$\EE$. Thus, $p$ is the unique winner in~$\EE$ if every other nominated candidate obtains Maximin-score at most one, that is, gets defeated by at least one other nominee.
    
    Suppose that this indeed happens. Note that it is not possible that each of the three candidates~$c_1^-, c_2^-$, and $c_3^-$ is nominated in~$\EE$ for some $c \in \C$, as in that case none of the remaining nominees can defeat these nominees. Hence, for each $c \in \C$ there exists some $j \in [3]$ for which $c_j$ is nominated by party~$P_{c,j}$. 
    Based on the preferences of the voters, it is straightforward to verify that candidate~$c_j$ can only be defeated by  the nominee of a party~$P_x$.  
    In particular,$\vphantom{L^{L^{\hat{a}}}}$ 
    $c_j$ can only be defeated by the literal $\ell$ for which $c_j \in A(\ell)$. Thus, we get that the $j$-th literal in~$c$ must be nominated. 
    This immediately implies that the truth assignment that sets exactly the nominated candidates from~$L$ as true satisfies~$\varphi$, since for each clause~$c \in \C$ it sets at least one literal true.

    \medskip
    Assume now that we have a satisfying truth assignment~$\alpha$ for $\varphi$. 
    Consider the election~$\EE$ where each party~$P_x$, $x \in X$, nominates the literal set to true by~$\alpha$, 
    and where each party~$P_{c,j}$, $c \in \C$ and $j \in [3]$, nominates~$c_j$ if and only if the $j$-th literal in~$c$ is set to true by~$\alpha$.
    It is straightforward to check that $p$ is the unique winner in~$\EE$, since each nominee other than~$p$ gets defeated by some other nominee: 
    \begin{itemize}
        \item $p'$ gets defeated by every nominee in~$\C^-$ (note that there is at least one nominee in~$\C^-$ by our assumption that $\varphi$ is nontrivial);
        \item each nominee~$c_j^- \in \C^-$ gets defeated by all nominees in $\{c_1,c_2,c_3\}$, and at least one of these candidates is nominated because $\alpha$ satisfies~$\varphi$;
        \item each nominee~$c_j$ for some $c \in \C,j \in [3]$ gets defeated by the $j$-th literal in~$c$, which is a literal set to true by~$\alpha$ (since $c_j$ is nominated) and thus nominated in~$\EE$;
        \item each nominated literal~$\ell \in L$ is defeated by~$p'$.
    \end{itemize}
    This proves the correctness of our reduction.
\end{proof}

\begin{theorem}
    \label{thm:Maximin_V5_s2}
    \PP\ for Maximin is $\NP$-complete even for instances with five voters and maximum party size~$\maxsize=2$.
\end{theorem}

\begin{proof}
We provide a polynomial-time reduction from $3$-SAT based on the reduction presented in the proof of Theorem~\ref{thm:Maximin_V4_s2}.
We keep all notation regarding the input instance~$\varphi$ of 3-SAT, as well as the notation introduced in~(\ref{eq:def-c-minus}), (\ref{eq:def-v4s2-setA}) and~(\ref{eq:def-v4s2-blocks}).
Furthermore, the set of parties and candidates contains all candidates and parties defined in the the proof of Theorem~\ref{thm:Maximin_V4_s2}, as well as an additional party~$P''=\{p''\}$. 

We define the set of voters as $V=\{v,v',w,w',z\}$, and define their preferences as follows: 
    \begin{align*}
        v &: p, \ora{L}, Y_{c^1}, \dots, Y_{c^m}, p', p''; \\
        v' &: p, \wt{Y}_{c^m}, \dots, \wt{Y}_{c^1}, p', p'', \ola{L}; \\
        w &: p'', 
        F_{x_1}, \dots, F_{x_n}, 
        F_{\ol{x}_1}, \dots, F_{\ol{x}_n},
        \ora{\C^-}, p', p; \\
        w' &: p', p'',
        F_{\ol{x}_n}, \dots, F_{\ol{x}_1},
        F_{x_n}, \dots, F_{x_1}, \ola{\C^-}, p
        \\
        z &:    
        \ora{\C^-}, p',p'', \ora{L}, \ora{\{c_j:c \in \C, j \in [3]\}}, p.
    \end{align*}

    We claim that $P$ can be the unique winner in the election~$\EE$ resulting from some nominations for the parties if and only if $\varphi$ admits a satisfying truth assignment. 
    
    Observe first that candidate~$p$ has Maximin-score $\MM_\EE(p)=2$ in~$\EE$. 
    Let us further observe the following facts: 
    \begin{itemize}
        \item $p''$ defeats each nominee in~$L$ by $4:1$;
        \item $p'$ defeats~$p''$ by $4:1$;
        \item each nominee in~$\C^-$ defeats~$p'$ by $4:1$;        
    \end{itemize}

    Now assume that $p$ indeed is the unique winner in~$\EE$. Then each candidate $c_j^- \in \C^-$ that is nominated in~$\EE$ must be defeated by $4:1$ by some nominee in~$\EE$; however, it is immediate from the preferences of the voters that this can only be done by some nominee in~$\{c_1,c_2,c_3\}$. Hence, we get that for each $c \in \C$ there exists some $j \in [3]$ such that $c_j$ is nominated in~$\EE$.
    
    Next consider some nominee $c_j$ where $c \in \C, j \in [3]$. 
    Again by construction, it follows that the only nominee that may defeat $c_j$ by $4:1$ is the literal~$\ell$ for which $c_j \in A(\ell)$, that is, the $j$-th literal in clause~$c$. 
    Hence, for each clause~$c$ there must exist a literal in~$c$ that is nominated in~$\EE$. This immediately yields a satisfying assignment that sets a literal to true if and only if it is nominated in~$\EE$.

    \medskip
    Assume now that $\varphi$ admits a satisfying truth assignment~$\alpha$. Consider the election~$\EE$ where each party $P_x$, $x \in X$, nominates the literal set to true by~$\alpha$, and where each party~$P_{c,j}$, $c \in \C$ and $j \in [3]$, nominated $c_j$ if and only the $j$-th literal in~$c$ is set to true. Again, it is straightforward to check that $p$ is the unique winner in~$\EE$, since every other nominee obtains a Maximin-score at most~$1$. To see this, 
    it suffices to complement our earlier observations with the following: 
    \begin{itemize}
        \item there is at least one nominee in~$\C^-$ by our assumption that $\varphi$ is nontrivial;
        \item for each $c \in \C$ there is at least one nominee in~$\{c_1,c_2,c_3\}$, and this nominee defeats each nominee in~$\{c_1^-,c_2^-,c_3^-\}$ by~${4:1}$;
        \item 
        for each nominee $c_j$, $c \in \C$ and $j \in [3]$, the $j$-th literal~$\ell$ in~$c_j$ is set to true by~$\alpha$ and hence is nominated in~$\EE$; moreover, $\ell$ defeats $c_j$ by~${4:1}$.
    \end{itemize}
    This proves the correctness of our reduction.
\end{proof}

\subsection{Proof of Claim~\ref{clm:latrules}}
\label{app:proof-of-clmlatrules}

\clmlatrules*

\begin{proof}
Given an instance~$I=(D,H,\gamma)$ of \LabSubDigraph\, we will refer to a subdigraph of~$H$ that is $\gamma$-isomorphic to~$D$ as a \emph{solution} for~$I$.

Consider Rule~\ref{latrule:leaf} when applied to an instance~$I=(D,H,\gamma)$, and let $U=\{x \in \Gamma_u: N_H^+(x) \cap \Gamma_v = \emptyset\}$.
First, if $\wt{H}$ is a subdigraph of~$H$ that is $\gamma$-isomorphic to~$D$, then 
the unique vertex of~$\wt{H}$ with label~$u$ must have an out-neighbor with label~$v$. Hence,  $\wt{H}$ cannot contain a vertex from~$U$, which means that deleting $U$ from~$H$ yields an instance~$I'$ equivalent to~$I$. 

The instance~$I''$ constructed by Rule~\ref{latrule:leaf} is obtained by further deleting $v$ from~$D$ and $\Gamma_v$ from~$H-U$. We show that $I''$ is equivalent to~$I'$ and, hence, to~$I$. Assume that $\wt{H}$ is a solution for~$I''$. 
Then $\wt{H}$ contains a vertex $x \in V(H-U-\Gamma_v)$ with label~$u$, and because $x \notin U$, there must exist a vertex~$y$ in~$H$ among the out-neighbors of~$x$ in~$H$ that has label~$v$. 
Adding the vertex~$y$ and the arc~$(x,y)$ to~$\wt{H}$ yields a solution for the instance~$I'$. 
Conversely, the solvability of~$I'$ implies the solvability of~$I''$, because 
deleting the unique vertex with label~$v$
from a solution to~$I'$ yields a solution for~$I''$.

\medskip
Consider now Rule~\ref{latrule:cycle}  applied to instance
$I=(D,H,\gamma)$, and let $I'$ be the constructed instance with digraphs~$D'$ and~$H'$.

First, if $\wt{H}$ is a subdigraph of~$H$ that is $\gamma$-isomorphic to~$D$ witnessed by an isomorphism~$f$, then $(f(u),f(w)) \in A_{uw}$ is clear, due to~$f(v) \in N_H^+(u) \cap N_H^-(w) \cap \Gamma_v$. Hence, deleting $f(v)$ from~$\wt{H}$ and adding the arc~$(f(u),(w))$ yields a subdigraph of~$H'$ that is $\gamma$-isomorphic to~$D'$. 

Conversely, if $\wt{H}$ is a solution to~$I'$, then it contains an arc~$(x,y) \in A_{uw}$ such that $x \in \Gamma_u$ and $y \in \Gamma_w$, because $(u,w)$ is an arc in~$D'$. By the definition of~$A_{uw}$, it follows that there is a vertex $z \in N_H^+(u) \cap N_H^-(w) \cap \Gamma_v$ in~$H$. Adding $z$ and replacing the arc~$(x,y)$ with arcs~$(x,z)$ and~$(z,y)$ yields a subdigraph of~$H$ that is $\gamma$-isomorphic to~$D$. This shows that $I'$ is indeed equivalent with~$I$.
\end{proof}

\subsection{Proof of Theorem~\ref{thm:MM-fpt}}
\label{app:proof-of-fptcorrect}
\thmMMfpt*

\begin{lemma}
\label{lem:AlgMM-correct-1}
If \AlgMM\ returns ``yes'', then the input instance of \PP{} for Maximin is a ``yes''-instance.
\end{lemma}

\begin{proof}

\AlgMM\ may only return ``yes'' in Step~7, which requires that 
instance~$J$ of~\LabSubDigraph\ admits a solution, i.e., a subdigraph~$\wt{H}$ of~$H$ that is $\gamma$-isomorphic to~$D$. 
Note that $C':=V(\wt{H})$ contains exactly candidate from each party in~$\P \setminus \{P^\star\}$; let $\EE$ be the election where the nominees are exactly $C' \cup \{p\}$. 
By Step~3 we know $C' \cap X=\emptyset$, implying $\MM_\EE(p)=\min_{c \in C'} N(p,c) \geq s^\star$. 
Consider now some $c \in C'$ nominated by a party~$P$, and let $c' \in C'$ be the vertex of~$\wt{H}$ to which $\delta(P)$ is mapped by the isomorphism witnessing the solution~$\wt{H}$ for~$J$. As $(\delta(P),P)$ is an arc in~$D$, there must be an arc~$(c',c)$ in~$\wt{H}$ as well. By construction, this means $N(c,c')<s^\star$. Hence, we get
$\MM_\EE(c)=\min_{\hat{c} \in C'\cup \{p\}} N(c,\hat{c}) \leq N(c,c')<s^\star$. Thus, $p$ is indeed the unique winner in~$\EE$.
\end{proof}

\begin{lemma}
\label{lem:AlgMM-correct-2}
For each ``yes''-instance of \PP\ for Maximin, \AlgMM\ returns ``yes''.
\end{lemma}

\begin{proof}
Assume that there exists a set of nominations resulting in an election~$\EE$ where $P^\star$ is the unique winner. Let $C'$ denote the set of nominees in~$\EE$, and let $c_P$ denote the nominee of party~$P$ for each $P \in \P \setminus \{P^\star\}$.

Consider the branch of \AlgMM\ where 
the algorithm's guesses are correct, that is, it correctly guesses the nominee $p$ of~$ P^\star$, the Maximin-score~$s^\star$ of~$p$ in~$\E$, as well as a party~$\delta(P)$
for each party~$P \in \P \setminus \{P^\star\}$ 
whose nominee fulfills $N_{\E}(c_P,c_{\delta(P)})<s^\star$. 
Clearly, such a party~$\delta(P)$ exists, as otherwise $P$ would have Maximin-score at least~$s^\star$ in~$\EE$. 
We will use the fact that $N_{\E_0}(c,c')=N_{\E}(c,c')$ holds in any reduced election~$\E$ of~$\E_0$ for each two candidates~$c$ and~$c'$ that take part in the election~$\E$.

First, note that \AlgMM\ cannot return ``no'' in Step~5, because for each party $P \in \P \setminus \{P^\star\}$ we know that 
\begin{itemize}
\item[(i)] $N_{\E_0}(p,c_P)= N_{\E}(p,c_P)\geq \MM_\EE(p)=s^\star$, and that
\item[(ii)] if $\delta(P)=P^\star$, then $N_{\E_0} (c,p)=N_{\E} (c,p)<s^\star$, 
\end{itemize}
and therefore
$c_P$ is not deleted in Steps~3 and~4, so $c_P \notin X$. 
Hence, the algorithm reaches Step~7. 

It now suffices to show that the subgraph $H[C']$ of~$H$ spanned by all nominees in~$\EE$ contains a graph $\gamma$-isomorphic to~$D$. We show that the function $f(P)=c_P$ is a good witness for this, i.e, 
for each arc~$(P_1,P_2)$ in~$D$ there is an arc~$(f(P_1),f(P_2))$; note that the label constraint $\gamma(f(P))=\gamma(c_P)=P$ is satisfied automatically for each $P \in \P \setminus \{P^\star\}$.
Recall that $(P_1,P_2)$ is an arc in~$D$ if and only $P_1=\delta(P_2)$. As we assume that the algorithm's guesses are correct, we get $N(c_{P_2},c_{\delta(P_2)})<s^\star$. 
Since $f(P_1)=c_{P_1}=c_{\delta(P_2)}$ and $f(P_2)=c_{P_2}$, this is equivalent with saying that $N(f(P_2),f(P_1))<s^\star$. 
Hence, by the definition of~$H$ we get that $(f(P_1),f(P_2))$ is an arc in~$H$, as promised. 
This shows that \AlgMM\ returns ``yes'' in Step~7.
\end{proof}

Lemmas~\ref{lem:AlgMM-correct-1} and~\ref{lem:AlgMM-correct-2} together imply the correctness of \AlgMM, so to finish the proof of Theorem~\ref{thm:MM-fpt}, it remains to determine it's running time.

The number of possible guesses that the algorithm has to explore due to Steps~1 and~2 is $|P^\star| \cdot |V| \cdot (t-1)^{t-1}$.
For each such guess, performing Steps~3--6 takes $O(|V|\cdot |C|^2)$ time, as this is the time necessary to compute $N_{\E_0}(c,c')$ for each two candidates~$c$ and~$c'$ in~$\EE_0$. 
Finally, running the algorithm for \LabSubDigraph\ guaranteed by Lemma~\ref{lem:labelled-subdigraph} takes $O(|C|^2)$ time. Hence, the overall running time of \AlgMM\ is $O(t^t \cdot |V|^2 \cdot |C|^3)$, which is clearly fixed-parameter tractable with respect to parameter~$t$.

\bibliographystyleapp{ACM-Reference-Format}
\bibliographyapp{citeme}

\end{appendices}

\fi

\end{document}